\documentclass[pdflatex,sn-nature]{sn-jnl}

\usepackage{graphicx}%
\usepackage{multirow}%
\usepackage{amsmath,amssymb,amsfonts,amsthm}%
\usepackage{thm-restate}
\usepackage{thmtools,nameref}
\usepackage[capitalise]{cleveref}
\usepackage{mathtools} 
\usepackage{physics} 
\usepackage{bm} 
\usepackage{bbm} 

\usepackage{bookmark}
\usepackage[title]{appendix}
\usepackage{apptools}
\usepackage{titletoc}
\usepackage{setspace}
\usepackage{xspace}

\usepackage{tikz}
\usetikzlibrary{arrows.meta,calc,positioning,decorations.pathreplacing}
\usepackage{xcolor} 
\definecolor{setoneblue}{RGB}{199,210,235}
\definecolor{settwopeach}{RGB}{244,208,177}
\definecolor{settwoyellow}{RGB}{255,255,153}
\definecolor{settwogreen}{RGB}{180,214,118}

\theoremstyle{thmstyleone}%
\newtheorem{theorem}{Theorem}
\declaretheorem[name=Proposition,
refname={proposition,propositions},
Refname={Proposition,Propositions}]{prop}
\declaretheorem[name=Lemma,
refname={lemma,lemmas},
Refname={Lemma,Lemmas},
numberlike=prop]{lemma}
\declaretheorem[name=Corollary,
refname={corollary,corollaries},
Refname={Corollary,Corollaries},
numberlike=prop]{corollary}
\declaretheorem[name=Definition,
refname={definition,definitions},
Refname={Definition,Definitions},
numberlike=prop]{definition}
\theoremstyle{remark}
\newtheorem{remark}[prop]{Remark}
\newtheorem*{remark*}{Remark}

\AtAppendix{\numberwithin{prop}{section}}
\AtAppendix{\numberwithin{lemma}{section}}
\AtAppendix{\numberwithin{theorem}{section}}
\AtAppendix{\numberwithin{definition}{section}}
\AtAppendix{\numberwithin{corollary}{section}}
\AtAppendix{\numberwithin{remark}{section}}

\newcommand{\app}{\text{supplementary information}}

\newcommand{\doc}{\text{manuscript}}  

\newcommand{\Meth}{\text{Technical Material}}

\newcommand{\TrX}[2]{\mathrm{tr}_{#1}\left[ #2 \right]}

\newcommand{\ktensor}{\mathbin{\otimes_K}}

\newcommand{\rtensor}{\mathbin{\otimes_{\textup{R}}}}
\newcommand{\rtensordual}{\mathbin{\otimes_{\textup{R}^*}}}
\newcommand{\Rprod}{\mathbin{\otimes_R}}

\newcommand{\nn}{{\mathbbm{N}}}

\newcommand{\rr}{{\mathbbm{R}}}

\newcommand{\cc}{{\mathbbm{C}}}
\newcommand{\zz}{{\mathbbm{Z}}}

\newcommand{\kk}{{\mathbbm{K}}}

\newcommand{\id}{{\mathbbm{1}}}
\newcommand{\Aa}{\textup{A}}
\newcommand{\Bb}{\textup{B}}
\newcommand{\Xx}{\textup{X}}
\newcommand{\Yy}{\textup{Y}}

\newcommand{\Gs}{\Gamma}

\newcommand{\G}[1]{\Gamma\left\{ #1 \right\}}
\newcommand{\I}[1]{\overline{I}^{{\scriptscriptstyle( #1 )}}}
\newcommand{\J}[1]{\overline{J}^{{\scriptscriptstyle( #1 )}}}

\newcommand{\ReP}[1]{\mathrm{Re}\left( #1 \right)}
\newcommand{\ImP}[1]{\mathrm{Im}\left( #1 \right)}

\newcommand{\HermC}[1]{\mathsf{Herm}_{#1}\left(\cc\right)}

\newcommand{\SYR}[1]{\mathsf{SY}_{#1}(\rr)}

\usepackage{multibib}
\newcites{app}{Supplementary References}
\crefname{equation}{Eq.}{Eqs.}
\crefname{appendix}{Supplementary Information}{Supplementary Information}
\Crefname{appendix}{S.I.}{S.I.}

\usepackage{comment}
\usepackage[normalem]{ulem}

\usepackage{xcolor}

\newcommand\mpwS[1]{{\let\helpcmd\sout\parhelp#1\par\relax\relax} }
\long\def\parhelp#1\par#2\relax{%
	\helpcmd{#1}\ifx\relax#2\else\par\parhelp#2\relax\fi%
}

\begin{document}

\title{Quantum theory based on real numbers cannot be experimentally falsified}

\author*[1]{\fnm{Timothée} \sur{Hoffreumon}}\email{t.hoffreumon@gmail.com}

\author*[2,3]{\fnm{Mischa P.} \sur{Woods}}\email{mischa.woods@gmail.com}

\affil*[1]{\orgdiv{Mathematical Institute}, \orgname{Slovak Academy of Sciences}, \orgaddress{
\city{Bratislava}, \country{Slovakia}}}

\affil*[2]{\orgdiv{ENS Lyon}, \orgname{Inria}, \orgaddress{
\city{Lyon}, \country{France}}}
\affil[3]{\orgdiv{Université Grenoble Alpes}, \orgname{Inria}, \orgaddress{
\city{Grenoble}, \country{France}}}

\abstract{Whether the complex numbers of standard quantum theory are experimentally indispensable has remained open for decades. Real quantum theory (RQT), obtained by replacing complex amplitudes with real ones while retaining the usual Kronecker-product composition rule, reproduces all single-party and bipartite Bell correlations of quantum theory (QT), but its lack of local tomography suggested that the two theories might diverge in more general local experiments. This possibility appeared to be confirmed by Renou \textit{et al.}~\cite{Renou2021}, who argued that a bilocal network experiment can falsify RQT without falsifying QT. Here we show that this conclusion relies on an experimentally untestable assumption. The key distinction is between product-state independence, which constrains the mathematical form of source states, and operational independence, which is defined entirely by the absence of observable cross-source correlations. We prove that, once source independence is imposed operationally, every finite network correlation achievable in QT is also achievable in RQT with the same locality structure of the measurements. We then extend this equivalence to arbitrary finite sequential multipartite protocols involving channels and measurements with prescribed locality structure. Thus, as long as no violation of QT is observed, RQT cannot be experimentally falsified. Our results restore the empirical indistinguishability of QT and RQT, while showing that they support markedly different pictures of the correlation structure underlying the same observed world.
}


\maketitle


\addcontentsline{toc}{part}{Main Text}
\textbf{\textit{Real quantum theory}} (RQT) is obtained by replacing complex numbers in the standard Hilbert-space formulation of \textbf{\textit{quantum theory}} (QT) with real numbers. In the mixed-state formalism, density operators, observables and POVMs are then represented by real symmetric matrices rather than complex Hermitian ones, while the Born rule, the L\"uders update rule, and the use of the Kronecker product to combine systems are left unchanged. Since the 1990s, RQT has served as a canonical foil theory~\cite{Chiribella_2016} for probing the role played by complex amplitudes in QT~\cite{Wooters1990,Myrheim_1999,Caves2001,Rudolph_2002,Aaronson_2004,McKague,Aleksandrova_2013,Wootters_2016,Fuchs_2022}. Unlike more recent alternative representations of QT based on composition rules other than the Kronecker product~\cite{Barrios2025,RNQT,Ying_2025,Erba_2025}, RQT is generally regarded as a different theory from QT. Its best-known structural difference from QT is that the theory is no longer locally tomographic~\cite{Wooters1990}: the state of a multipartite system cannot, in general, be reconstructed from local measurement statistics alone. This feature has been widely studied~\cite{Araki_1980,Klay_1987,Hardy_2009,Hardy_2011,Barnum_2013,Barnum_2023,Selby_2023,Centeno_2025}, yet it does not by itself imply an experimentally accessible difference from QT. RQT reproduces all predictions of QT in single-party scenarios~\cite{Stueckelberg1960,Myrheim_1999,McKague}.

{\it What does this mean?} It means that the quantum description of an experiment, which uses complex matrices, can always be mapped to a description using real matrices preserving the measurement statistics~\cite{Stueckelberg1960,Rudolph_2002,Myrheim_1999,Wooters1990}. The classic example is the mapping
\begin{equation}
    \ket{\psi}=\sum_j c_j\ket{j} \mapsto |\widetilde{\psi}\rangle=\sum_j \ReP{c_j}\ket{0}\ket{j} +\ImP{c_j}\ket{1}\ket{j}\:,
\end{equation}
where $\{\ket{j}\}_j$ is an orthonormal basis for the original QT system, and $\{\ket{0}, \ket{1}\}$ is an orthonormal basis of a rebit system\footnote{A rebit system is the RQT analogue of the qubit for QT: a two-level system with real coefficients.}. A similar mapping exits for measurements. (The apparent mismatch in dimensions is not operationally meaningful, both because a complex space of dimension $d$ is already a real space of dimension $2d$, and because the dimension of a quantum system cannot itself be certified in a theory-independent way~\cite{Brunner_2008}.) But this real description maps local measurements to nonlocal ones, thus prompting a natural question, already emphasised by Gisin~\cite{Gisin_2007} and others~\cite{McKague}: does the equivalence fail once locality is imposed on the measurements? In other words, if spatially separated parties are restricted to performing local operations on their own subsystems, can QT and RQT then be experimentally distinguished? After a first negative answer for the special case of Bell experiments~\cite{Pal_2008,Moroder_2013}, this question became particularly pressing when Renou \textit{et al.}~\cite{Renou2021} claimed that the answer is yes for the bilocal (or entanglement-swapping) experiments~\cite{Branciard2010}. In that scenario, three parties are connected by two independent sources, and the authors derived a Bell-like inequality whose maximal value in RQT would be strictly smaller than in QT. Because the inequality is theory-independent, any observed value above the real bound would seem to falsify RQT while remaining compatible with QT. Experiments have since reported violations consistent with the QT prediction~\cite{Li2022,Lancaster2025}. At face value, this suggested that the distinction between QT and RQT might become observable precisely when locality and source independence are both taken seriously.

Here we show that this conclusion relies on an experimentally untestable assumption. The key point is that two distinct notions are being conflated: the locality of measurements and the independence of sources. These play very different conceptual roles. Locality concerns the operations performed by separated parties. In a network of parties $\Xx_1,\ldots,\Xx_n$, if party $\Xx_i$ acts only on subsystem $i$, then—assuming finite dimensions—the corresponding POVM elements may be represented without loss of generality as local operators, and joint local measurements may be represented in tensor-product form. This locality structure is therefore trustworthy: it reflects the physical fact that separated operations commute, and it can be imposed equally in QT and in RQT. In particular, when one studies Bell inequalities, network inequalities, or more general nonlocal games played by separated parties, the restriction to local measurements is a physically motivated and operationally meaningful constraint on the model.

Source independence is more subtle. Consider several sources distributing systems to the parties in such a network. In QT, one often encodes the statement that two sources are independent by requiring their joint state to be a product state, that is, a bipartite state $\Psi_{\textup{XY}}$ of the form 
\begin{equation}\label{def:LIR}
    \Psi_{\textup{XY}} = \rho_{\textup{X}}\otimes_K\rho_{\textup{Y}} \:, 
\end{equation} 
for arbitrary states $\rho_{\textup{X}}$, $\rho_{\textup{Y}}$ on systems $X$, $Y$, and where $\ktensor$ denotes the Kronecker product. We call this \textbf{\textit{product-state independence}}. But this is not itself something the parties directly observe. What they can observe is weaker and more operational: if the sources are independent, then local measurements on the systems emitted by different sources should reveal no cross-source correlations; the outcomes of these local measurements should be independent as random variables. This motivates the notion of \textbf{\textit{operational independence}}. For two systems, one may say that a bipartite state $\Psi_{\textup{XY}}$ is operationally independent when, for every pair of local POVMs $\{X_{\alpha}\}_\alpha$ and $\{Y_{\beta}\}_\beta$, the joint probability distributions of the outcomes $\alpha$ and $\beta$, $p(\alpha, \beta)=\TrX{}{\Psi_\textup{XY}(X_{\alpha}\ktensor Y_{\beta})}$, factorise into marginal distributions:
\begin{equation}\label{def:LIQ}
    p(\alpha,\beta) = p(\alpha)p(\beta) \:.
\end{equation} 
The crucial point is that this notion is experimentally motivated: it refers only to the absence of observable correlations in the statistics gathered by the local parties. By contrast, product-state independence is a statement about the formal expression assigned to the state inside a particular mathematical model. The parties can test the former by measurement and data comparison; they cannot directly inspect the latter.

In standard QT these two notions coincide. A bipartite quantum state is operationally independent if and only if it is a product state. This coincidence is so familiar that it is easy to mistake it for a general physical principle. But it is not. In RQT the two notions come apart. Every real product state is operationally independent, but there also exist real states that are \emph{not} product states and yet remain operationally independent. These states are non-separable in the matrix representation, but they never generate observable correlations under local measurements. This is a direct consequence of the fact that RQT is not locally tomographic. The real formalism is simply less able to detect certain correlations than the complex one. One may therefore regard these states as exhibiting only an ``entanglement of representation''~\cite{Caves2001}: they appear entangled in the Hilbert-space description, but that entanglement has no operational manifestation in the network experiment. This is not to say that RQT is `wrong' or can be a priori ruled out as a description of nature---it simply means that in a world described by RQT, not all correlations are observable. 

This observation already changes the interpretation of the bilocal result of Renou \textit{et al.}~\cite{Renou2021}. Their analysis assumes that independent sources in RQT must be represented by real product states, in direct analogy with QT. But once operational independence and product-state independence are recognised as distinct in RQT, that restriction is no longer compulsory. It is an additional modelling assumption, stronger than what is experimentally warranted. Against that background, the correct question is not whether a particular restricted subclass of real models can reproduce certain network correlations, but whether RQT itself, understood with the operational notion of source independence, can be distinguished from QT without violating QT.

Our first main result answers this in the negative for arbitrary network scenarios~\cite{Tavakoli_2022}. Consider any finite network of local parties and sources, with each party performing a measurement on potentially several systems emitted by different sources. The only experimentally motivated assumption placed on the sources is operational independence: no local test performed across distinct sources may reveal correlations between them. Within this general setting, we show that every experiment admitting a QT description also admits an RQT description with exactly the same observable statistics. Moreover, the real construction preserves the locality structure of the measurements. Thus the well-known real embedding of unrestricted QT can be strengthened so that it remains compatible with the trusted local structure of the measurement devices. The equivalence between QT and RQT therefore survives the introduction of local measurements and independent sources, provided that `independent sources' is understood operationally. This is not to say that product-state independence is an `improper' way of defining independence; it is simply that it is not experimentally verifiable in the network setting.

\begin{figure}[t]
	\centering
	\begin{tikzpicture}[
		x=1.55cm,
		y=1.15cm,
		>=Latex,
		line width=0.55pt,
		source/.style={
			draw,
			circle,
			fill=gray!18,
			minimum size=7mm,
			inner sep=0pt
		},
		party/.style={
			draw,
			rounded corners=3pt,
			fill=blue!8,
			minimum width=9mm,
			minimum height=8mm,
			inner sep=2pt
		},
		qwire/.style={gray!70},
		cwire/.style={-{Latex[length=2.4mm,width=1.7mm]}},
		lab/.style={font=\small}
		]
		
		\node[party] (A1) at (0,0) {$A_1$};
		\node[party] (A2) at (1.5,0) {$A_2$};
		\node[lab]   at (3.0,0) {$\cdots$};
		\node[party] (Ak) at (4.5,0) {$A_{n-1}$};
		\node[party] (AN) at (6.0,0) {$A_n$};
		
		\node[source] (S1) at (0.75,1.7) {$S_1$};
		\node[source] (S2) at (2.25,1.7) {$S_2$};
		\node[lab]    at (3.75,1.7) {$\cdots$};
		\node[source] (SL) at (5.25,1.7) {$S_L$};
		
		\draw[qwire] (S1) -- (A1);
		\draw[qwire] (S1) -- (A2);
		
		\draw[qwire] (S2) -- (A2);
		\draw[qwire] (S2) -- (Ak);
		
		\draw[qwire] (SL) -- (Ak);
		\draw[qwire] (SL) -- (AN);
		
\draw[qwire,densely dashed] (S2) -- +(0.95,-1.10);
\draw[qwire,densely dashed] (S2) -- +(0.18,-1.00);
		
		\draw[cwire] (-0.95,0) -- (-0.28,0);
		\node[left,lab] at (-0.95,0) {$\mu_1$};
		
		\draw[cwire] (0.55,-0.85) -- (1.22,-0.22);
		\node[below left,lab] at (0.55,-0.85) {$\mu_2$};
		
\draw[cwire] (3.49,-0.85) -- (4.16,-0.22);
\node[below left,lab] at (3.51,-0.85) {$\mu_{n-1}$};
		
		\draw[cwire] (6.95,0) -- (6.28,0);
		\node[right,lab] at (6.95,0) {$\mu_n$};
		
		\draw[cwire] (0, -0.45) -- (0,-1.10);
		\node[below,lab] at (0,-1.10) {$x_1$};
		
		\draw[cwire] (1.5,-0.45) -- (1.5,-1.10);
		\node[below,lab] at (1.5,-1.10) {$x_2$};
		
		\draw[cwire] (4.5,-0.45) -- (4.5,-1.10);
		\node[below,lab] at (4.5,-1.10) {$x_{n-1}$};
		
		\draw[cwire] (6.0,-0.45) -- (6.0,-1.10);
		\node[below,lab] at (6.0,-1.10) {$x_n$};
		
		
	\end{tikzpicture}
	\caption{
		General $n$-party network considered in Theorem~1. Independent sources $S_1,\dots,S_L$ distribute physical systems to subsets of the local parties $A_1,\dots,A_n$. Each party performs a local measurement specified by a classical input $\mu_i$ (their \textit{setting}) and returns a classical output $x_i$ (their \textit{outcome}). The locality structure of the measurements is trusted and represented explicitly by the separation into local parties, whereas source independence is imposed only operationally, namely through the absence of observable cross-source correlations in the statistics.
	}
	\label{fig:network_theorem1}
\end{figure}
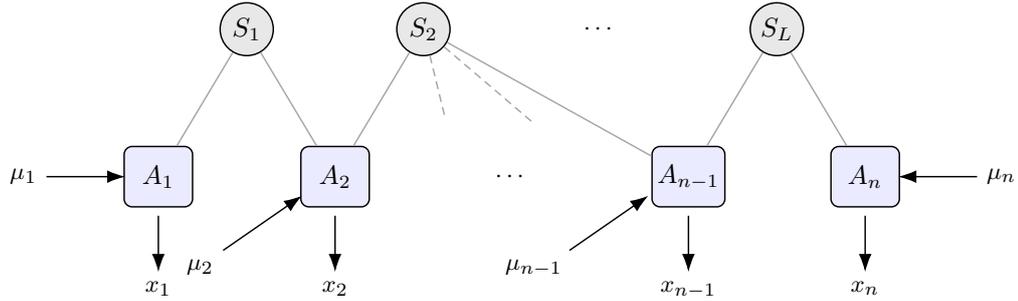

Informally, the theorem may be stated as follows (see the \Meth{} for an exact phrasing and the proof)
\begin{restatable}{theorem}{QTnetworksareRQTnetworks}\label{theo:QTnetworksareRQTnetworks}
	For any finite network of independent sources and locally measuring parties, if the sources are required only to be operationally independent, then every outcome distribution predicted by the QT model of the network can also be predicted by an equivalent RQT model.
\end{restatable}

This result subsumes the standard Bell scenario, bilocal entanglement-swapping scenarios, and more general network nonlocality experiments (such as the triangle and star networks; see Ref.~\cite{Tavakoli_2022}). It implies that Bell inequalities and Bell-like network inequalities cannot distinguish QT from RQT once the operational notion of independent sources is adopted and as long as a violation of QT is not observed. The reason the real model succeeds is that the required source states, although typically non-product in the real formalism, are locally indistinguishable from the product ones; RQT lacks the tomographic power to expose their non-separability.

This already establishes that no local game can distinguish QT from RQT in arbitrary network scenarios. 
But one may still wonder whether the equivalence breaks down once the experimental scenario is broadened beyond single-shot measurements. 
After all, real laboratory protocols often involve sequences of operations: state preparation, intermediate dynamics, adaptive measurements, and channels acting on selected subsystems before later rounds of detection. It is therefore natural to ask whether QT and RQT remain indistinguishable in sequential multipartite protocols, where both the initial state and subsequent measurements and transformations each carry a prescribed locality structure. See~\cref{fig:sequential_theorem2}.

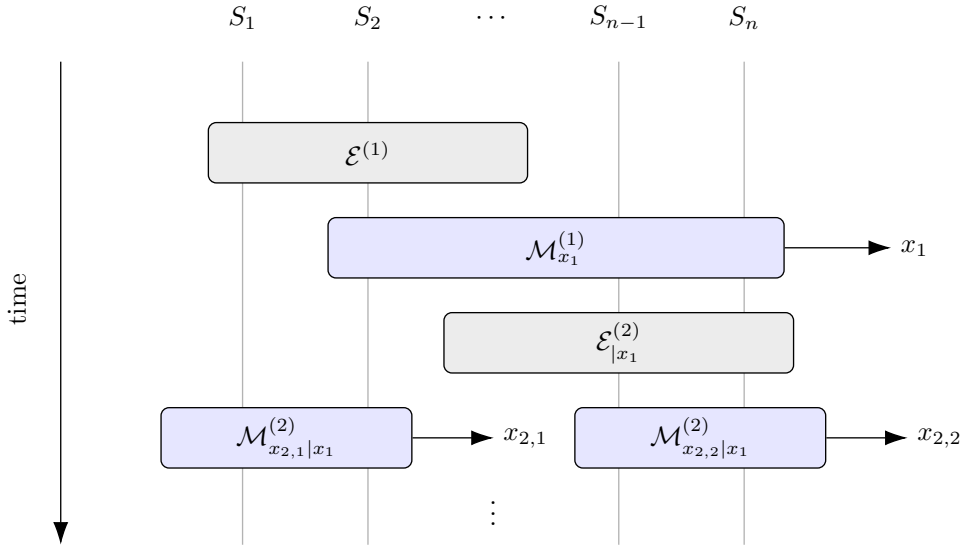
\begin{figure}[t]
	\centering
	\begin{tikzpicture}[
		x=1.65cm,
		y=1.05cm,
		>=Latex,
		line width=0.5pt,
		worldline/.style={gray!60},
		chan/.style={
			draw,
			rounded corners=3pt,
			fill=gray!15,
			minimum height=8mm,
			inner xsep=8pt,
			inner ysep=3pt
		},
		meas/.style={
			draw,
			rounded corners=3pt,
			fill=blue!10,
			minimum height=8mm,
			inner xsep=8pt,
			inner ysep=3pt
		}
		]
		
		\node at (0,0.55) {$S_1$};
		\node at (1,0.55) {$S_2$};
		\node at (2,0.55) {$\cdots$};
		\node at (3,0.55) {$S_{n-1}$};
		\node at (4,0.55) {$S_n$};
		
		\foreach \x in {0,1,3,4} {
			\draw[worldline] (\x,0) -- (\x,-6.1);
		}
		
		\draw[-{Latex[length=3mm,width=2mm]}] (-1.45,0) -- (-1.45,-6.1);
		\node[rotate=90] at (-1.80,-3.0) {time};
		
		\node[chan, minimum width=4.2cm] (E1) at (1.0,-1.15) {$\mathcal{E}^{(1)}$};
		
		\node[meas, minimum width=6.0cm] (N1) at (2.5,-2.35) {$\mathcal{M}_{x_1}^{(1)}$};
		\draw[-{Latex[length=3mm,width=2mm]}] (N1.east) -- ++(0.85,0);
		\node[right] at ($(N1.east)+(0.85,0)$) {$x_1$};
		
		\node[chan, minimum width=4.6cm] (E2) at (3.0,-3.55) {$\mathcal{E}^{(2)}_{|x_1}$};
		
		\node[meas, minimum width=3.3cm] (N21) at (0.35,-4.75) {$\mathcal{M}_{x_{2,1}|x_1}^{(2)}$};
		\draw[-{Latex[length=3mm,width=2mm]}] (N21.east) -- ++(0.65,0);
		\node[right] at ($(N21.east)+(0.65,0)$) {$x_{2,1}$};
		
		\node[meas, minimum width=3.3cm] (N22) at (3.65,-4.75) {$\mathcal{M}_{x_{2,2}|x_1}^{(2)}$};
		\draw[-{Latex[length=3mm,width=2mm]}] (N22.east) -- ++(0.65,0);
		\node[right] at ($(N22.east)+(0.65,0)$) {$x_{2,2}$};
		
		\node at (2.0,-5.6) {$\vdots$};
		
	\end{tikzpicture}
	\caption{
		Consider an $n$-partite system undergoing a sequence of rounds. In each round, a quantum channel is applied whose locality structure may be arbitrary but fixed: the channel may act jointly on some chosen subset of subsystems while factorising across the rest. This may be followed by a measurement round, again with a specified locality structure: a local measurement on one subsystem, a joint measurement on a designated block, or more generally any POVM that factorises according to a given partition of the parties. One may then iterate this pattern, allowing a general interleaving of channels and measurements with changing locality structure from round to round. The figure illustrates two such iterations: at round 1, the channel $\mathcal{E}^{(1)}$ is applied, then followed by a measurement in which outcome $x_1$ was obtained, yielding the post-measurement state-update channel (also known as a quantum instrument) $\mathcal{M}^{(1)}_{x_1}$; at round 2, $\mathcal{E}^{(2)}_{|x_1}$ is applied followed by a bipartite local measurement in which outcome $x_2=(x_{2,1},x_{2,2})$ was obtained, yielding the instrument $\mathcal{M}^{(2)}_{x_{2,1}|x_1} \ktensor\mathcal{M}^{(2)}_{x_{2,2}|x_1}$. Such protocols encompass not only standard nonlocal games but also sequential network experiments, adaptive tests, and multi-stage information-processing scenarios such as LOCC (Local Operations and Classical Communication).
	}
	\label{fig:sequential_theorem2}
\end{figure}

Our second main result shows that the indistinguishability persists in this more general setting.

For this entire class of protocols, we prove that the QT description can again be translated into RQT without changing any observable statistics and without altering the locality structure of the channels or the measurements (see~\Cref{sec:proof} for exact phrasing and proof).

\begin{restatable}{theorem}{RQTisQT}\label{theo:RQTisQT}
	Consider any finite multipartite protocol consisting of any initial state, followed by an arbitrary finite sequence of quantum channels and POVM measurements, each endowed with a specified locality structure. Then, every such protocol in QT admits an equivalent protocol in RQT with identical outcome statistics and the same locality structure at every stage. 

    This correspondence is component-wise and context-independent: each RQT channel	depends only on its corresponding QT channel, and not on the initial state nor on any other operations appearing elsewhere in the protocol; likewise, each RQT measurement depends only on its corresponding QT measurement; and the initial RQT state depends only on the initial QT state.
\end{restatable}
A direct corollary of \Cref{theo:RQTisQT} is that, in RQT, the preparation of product-state-independent sources cannot be certified, even using global operations; otherwise, our result would contradict the separation obtained by Renou \textit{et al.}~\cite{Renou2021}. This may appear counterintuitive. One might expect that an entanglement-breaking channel could erase the hidden correlations in the RQT description and thereby expose a statistical difference between RQT and QT. What our results show is that while such a channel may exist, without observing a violation of QT, an alternative RQT channel always exists that explains the observed data: a channel that destroys all operationally visible RQT correlations while leaving intact those correlations that remain inaccessible to local observation in the real formalism.

Taken together, these results place the burden of experimental distinction in a precise location. If an experiment is compatible with QT, then, within the operational framework considered here, it is also compatible with RQT. This remains true whether one considers a single round of local measurements on a network of independent sources or a multi-round protocol with channels and measurements carrying arbitrary locality structure. One can thus never experimentally falsify RQT. On the other hand, one might be able to falsify QT, by experimentally detecting a statistic consistent with RQT and inconsistent with QT. One such example would be the observation of two independent systems which are locally indistinguishable but globally distinguishable. Since QT is locally tomographic, this option---while possible in RQT---is impossible in QT (see \Cref{sec:RQTisnotQT} for a detailed discussion). This raises a potential security concern: because RQT cannot be ruled out experimentally, cryptographic protocols regarded as secure within QT may in principle be open to compromise through non-QT correlations.

\begin{figure}[t]
	\centering
	\begin{tikzpicture}[
		x=1cm,
		y=1cm,
		line width=0.45pt,
		>=Latex,
		panel/.style={
			rounded corners=2pt,
			draw=black!18,
			fill=black!1
		},
		paneltitle/.style={
			font=\normalsize
		},
		planetA/.style={
			draw=black!55,
			fill=blue!10
		},
		planetB/.style={
			draw=black!55,
			fill=black!10
		},
		dust/.style={
			draw=none,
			fill=black!12,
			opacity=0.55
		},
		star/.style={
			circle,
			fill=black!45,
			inner sep=0.45pt
		},
		dustgrain/.style={
			circle,
			fill=black!28,
			inner sep=0.35pt
		},
		corr/.style={
			draw=blue!45!black,
			opacity=0.38,
			line width=0.55pt
		}
		]
		
		\newcommand{\VisibleUniverse}{
			\draw[panel] (0,0) rectangle (6.3,3.7);
			
			\foreach \x/\y in {
				0.45/3.10, 0.95/2.85, 1.65/3.18, 2.10/2.92, 2.85/3.05,
				3.65/3.15, 4.20/2.95, 4.80/3.18, 5.45/2.88, 5.95/3.10,
				0.70/2.30, 1.95/2.15, 4.95/2.18, 5.70/1.95
			}{
				\node[star] at (\x,\y) {};
			}
			
			\fill[dust] (2.55,1.95) ellipse (0.95 and 0.46);
			\fill[dust,opacity=0.35] (2.95,1.65) ellipse (0.62 and 0.28);
			
			\foreach \x/\y in {
				2.10/2.00, 2.25/1.82, 2.40/2.15, 2.58/1.72, 2.70/1.98,
				2.88/1.80, 3.05/2.05, 3.18/1.68, 3.28/1.90, 2.78/2.18
			}{
				\node[dustgrain] at (\x,\y) {};
			}
			
			\draw[planetA] (1.05,1.02) circle (0.34);
			\draw[black!40] (1.05,1.02) ellipse (0.55 and 0.13); 
			\draw[planetB] (4.62,0.98) circle (0.26);
			\draw[planetB] (5.18,2.35) circle (0.18);
			\draw[planetA, fill=black!7] (3.72,2.68) circle (0.22);
			
			\foreach \x/\y/\r in {
				4.08/2.18/0.040,
				4.28/2.05/0.030,
				4.43/2.18/0.025,
				5.55/1.45/0.035,
				5.72/1.28/0.022,
				0.55/1.72/0.030
			}{
				\fill[black!25] (\x,\y) circle (\r);
			}
		}
		
		\begin{scope}[shift={(0,0)}]
			\VisibleUniverse
			\node[paneltitle] at (3.15,4.15) {\textbf{QT}};
		\end{scope}
		
		\begin{scope}[shift={(7.3,0)}]
			\VisibleUniverse
			\node[paneltitle] at (3.15,4.15) {\textbf{RQT}};
			
			\draw[corr] (1.05,1.02) to[out=20,in=190] (2.42,1.88);
			\draw[corr] (1.05,1.02) to[out=5,in=205] (4.62,0.98);
			\draw[corr] (2.80,1.70) to[out=0,in=200] (5.18,2.35);
			\draw[corr] (3.72,2.68) to[out=-100,in=110] (4.62,0.98);
			\draw[corr] (3.72,2.68) to[out=8,in=170] (5.18,2.35);
			\draw[corr, densely dashed] (2.32,2.08) to[out=35,in=-165] (3.72,2.68);
			\draw[corr, densely dashed] (2.48,1.62) to[out=-18,in=155] (4.62,0.98);
			\draw[corr, densely dashed] (1.05,1.02) to[out=55,in=-160] (3.72,2.68);
		\end{scope}
		
	\end{tikzpicture}
	\caption{
		Conceptual illustration of the interpretive difference between QT and RQT. The two panels depict the same visible universe and therefore the same observable phenomena. In the RQT description (right), however, the underlying state may contain additional correlations---shown schematically as faint curved links---that remain operationally inaccessible. The figure is purely illustrative: it conveys that RQT can attribute a denser correlation structure to its description of the world while remaining empirically indistinguishable from QT in experiments whose statistics lie within the quantum set. 
	}
	\label{fig:qt_rqt_correlations}
\end{figure}
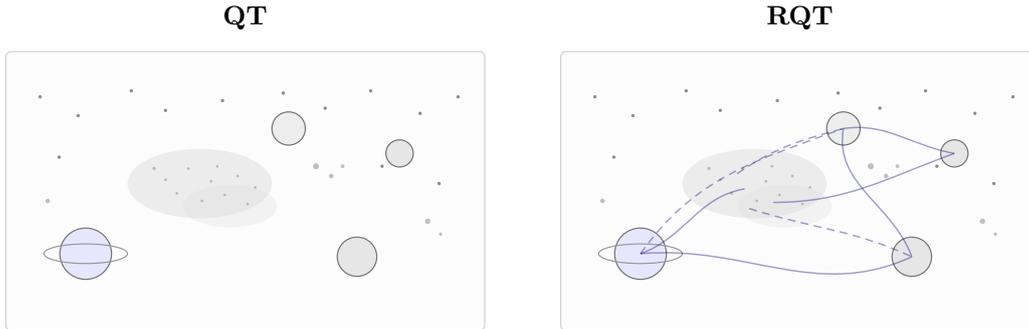

On a broader scope, our results highlight the complexity of testing a foil theory~\cite{Chiribella_2016} against quantum theory. The very accommodating formalism of quantum theory cannot simply be taken for granted in alternative theories. Before proposing experiments to challenge such alternatives against QT, one must first verify that each element of the alternative framework carries the same operational meaning as its counterpart in QT. Without doing so, the foil is condemned to remain a toy theory, and its scope cannot pretend to extend from the blackboard to the lab bench.

In the case of RQT versus QT, locality of measurements and channels does have such a shared meaning, but product-state independence does not. Once a notion of independence that does have shared meaning is imposed, we proved that the difference~\cite{Renou2021} in the predictions of RQT and QT vanishes. The scope of the falsification of RQT is confined to the blackboard.

An important consequence is interpretive. Within the framework considered here, all experiments to date that are consistent with QT are equally consistent with RQT. Yet the picture of the world suggested by the two theories is then rather different. In the standard quantum description, independent systems are typically taken to be genuinely uncorrelated unless correlations are created by interaction or common preparation. In the real description, by contrast, correlations may be far more widespread than they appear: systems that are operationally independent need not be uncorrelated at the level of the underlying real-state representation. What changes is not the observed data, but the theory's ability to resolve the correlations that may already be present. From this perspective, RQT describes a universe in which correlations are more pervasive, yet often remain experimentally invisible because the real formalism has less tomographic power than QT. Such hidden correlations need not be regarded as contrived: in a universe with a common early history, they could simply reflect correlations inherited from a shared origin and subsequently propagated across widely separated systems. Our results therefore show that present-day experiments do not compel the sparse-correlation ontology usually associated with QT; they are equally compatible with an interpretation in which correlations are effectively ubiquitous, but largely beyond operational reach. See~\cref{fig:qt_rqt_correlations}.

Our work therefore restores a simple but far-reaching conclusion. Real quantum theory is not experimentally distinguishable from standard quantum theory in any network or sequential protocol whose statistics remain compatible with quantum theory. 


\bmhead{Acknowledgements}
We thank useful conversations with Pedro Barrios Hita, Marco Erba, Thomas Galley, Miguel Navascu{\'{e}}s, Alexander Wilce, and William K. Wooters.

T.H. acknowledges support from the Slovak grants APVV-22-0570 and VEGA 2/0128/24. 

\section*{\Meth }
This technical section includes a high-level proof of the main statements as well as a more rigorous formulation of the claims of the main text. The technical core of the proof relies on a theory of real matrices we previously introduced in Ref.~\cite{RNQT} called Real-Number Quantum Theory (RNQT). Since this theory is isomorphic to QT, the overall proof strategy consists of showing that RNQT models are always contained within RQT models, which explains why RQT predictions can always explain QT predictions. This RNQT has a different notion of tensor-product locality than RQT---called $R$-product locality as opposed to Kronecker-product locality. The main technical hurdle is to prove that these distinct notions of product locality are interchangeable for measurements and channels, without altering the statistical outcomes for the relevant sets of states.

As a warm-up, we first revisit the cases where an RQT explanation is already known to exist, namely single-partite~\cite{Stueckelberg1960,McKague} and Bell (or local) experiments~\cite{Pal_2008,Moroder_2013}. Recovering these proofs using RNQT is almost immediate: one simply needs to identify the subset of states and effects of RQT that corresponds to the states and effects of RNQT. 

The main difficulty to overcome when extending the argument to arbitrary network scenarios (which we shall call `multilocal' in opposition to `local' experiments) is the treatment of independent sources. In particular, assuming product-state independence is insufficient to guarantee that RNQT appears as a submodel of RQT~\cite{Renou2021}; one needs to extend it to operational independence. 

Once this point is recognised, the general proof consists of the following steps: 1) Product-state-independent RQT states cannot be locally distinguished from product-state-independent RNQT states, and both classes satisfy operational independence. Therefore, replacing the former by the latter entails no loss of generality.
2) The existence of these product-state-independent RNQT states within the RQT state space is sufficient to ensure that the RQT model contains an RNQT submodel.
3) Since the RNQT submodel is isomorphic to the corresponding QT model, any bound on local inequalities achievable within (complex) quantum theory can also be achieved within a real quantum theory model that assumes operational independence.

The detailed constructions required for the proof (how to build the RNQT model and how it is isometrically isomorphic to a QT model) are largely technical. For this reason, these have been delayed into two appendices, \cref{sec:RNQT,sec:real_units}.

\subsection*{Preliminaries and the definitions of independent sources}
Let us first formalise our language: By a `\textbf{theory}', we mean here the general set of rules believed to provide an accurate mathematical description of physical phenomena. Following the recent RQT literature, we limit the scope of the theory to a definition of the set of accessible states of systems, the set of possible measurements performed on a system, the prescription of how to associate settings and outcomes to the measurements, and the rule to obtain probabilities out of a state and effect pair. 

This can be done because the theories under consideration in this paper---quantum theory (QT), real quantum theory (RQT), and real number quantum theory (RNQT)---are all \textit{foils} of quantum theory: theories that are sufficiently close to be essentialised as variations of the density-matrix formalism of quantum theory. Passing from one theory to the other simply amounts to changing the constraints defining the set of valid states and the set of valid effects, and, for RNQT, also to changing the matrix representation of the composition rule.

\begin{definition}
    In the following, \textbf{Quantum Theory (QT)} refers to the usual quantum-mechanical density-operator description of systems with a finite number of degrees of freedom, as described for example in Ref.~\cite{Nielsen2012}.

    \textbf{Real Quantum Theory (RQT)} refers to the restriction of QT under the following rule: For a chosen basis, the operators of RQT are only the operators of QT that are represented by real matrices in that basis.

    \textbf{Real-Number Quantum Theory (RNQT)} refers to the further modification of RQT under the following rules: 1) Every density operator and effect must commute with a given matrix $J$ such that $J^2 = - \id$; 2) The state and effects composition rules are respectively changed by the bilinear forms $\rtensor$ and $\rtensordual$. (The exact forms of $J$, $\rtensor$, and $\rtensordual$ are representation-dependent specificities of RNQT explicited in \Cref{sec:RNQT}.)
\end{definition}
In other words, RQT can always be seen as a sub-theory of QT, and RNQT as a sub-theory of RQT with modified composition rules. Nevertheless, while RNQT might appear `smaller' than RQT (in the sense of inclusions), it has been proven to yield the same predictions as QT~\cite{Barrios2025,RNQT}; the difference is that a system with $d$ degrees of freedom in QT is described as a system with $2d$ degrees of freedom in RNQT. This increase is allowed since the number of degrees of freedom of a system is not something that can be experimentally bounded~\cite{Brunner_2008}, and so it is treated as a parameter of the model of an experiment.

By a `\textbf{model}', we mean here the set of theory-independent assumptions made in order to \textit{model} a given experimental procedure as a set of probability distributions. It prescribes which and how systems are prepared at their source(s), how they are split into subsystems, and how these subsystems are subsequently distributed among the `labs' of the `parties' that will measure them. It further prescribes how the measurement settings and their outcomes are split amongst the parties. Once the underlying theory has been set, these prescriptions shape the mathematical expression used by the theory to generate the outcome distribution. 

Since QT, RQT, or RNQT all assign probabilities according to the Born rule, the prescriptions of the model result in the same formal expression regardless of the underlying theory. For instance, consider an experiment where Alice and Bob share a bipartite system. Irrespective of whether the underlying is QT, RQT, or RQNT, the outcomes-given-settings probability distribution is modelled by an expression of the form
\begin{equation}\label{eq:model}
    p(a,b|x,y) = \TrX{}{\rho_{\Aa\Bb}\: F_{a,b|x,y}} \:,
\end{equation}
where $\rho_{\Aa\Bb}$ is the state of the shared system, and where $F_{a,b|x,y}$ is the POVM element (the effect) representing Alice observing outcome $a$ and Bob observing outcome $b$ from a POVM $\{F_{a,b|x,y}\}_{(a,b)}$ (the measurement) they chose according to their respective settings $x$ and $y$. 

Because of that, the model can then be seen as the specific shape taken by the right-hand side of \cref{eq:model}. In this case, the information relative to the model of the experiment is the fact that there are two parties Alice and Bob, hence two tensor factors of the Hilbert spaces respectively labelled by $\Aa$ and $\Bb$, and that there are two pairs of random variables $(a,x)$ and $(b,y)$ associated with each respective party and representing their outcome and setting pair, hence a dependence of the POVM on these four variables. 

In the following, when both the theory and the model are fixed in an equation like \cref{eq:model}, we will use wording such as `RQT model' to refer to the right-hand side of this equation while implying that the matrices involved are all real matrices. We will use `QT model' and `RNQT model' similarly.

Nonetheless, the model is usually more than a prescription for associating parties with tensor factors and classical variables with operators; it often contains additional assumptions about the experimental context. These assumptions are then specialised depending on the underlying theory. This specialisation takes the form of extra constraints: it can be a constraint on the dimensions of the systems involved in the experiment, on a specific symmetry obeyed by the operations, etc. As we will now detail, the two constraints relevant to this work are those encoding the locality of parties and the independence of sources. 

By `\textbf{locality}', we mean here the constraint imposed on the model to capture the experimental fact that the parties operate in spacelike-separated labs\footnote{Or at least a situation where the parties refrain from performing measurements that affect other parties' systems in any observable way.}. Under this condition, the parties cannot influence each other's actions. In the model, this requirement is expressed through the commutativity assumption: any measurement performed in one laboratory must commute with measurements performed in another laboratory. 

Following the previous literature on RQT, we also assume that all systems have finite dimensions for simplicity. This will allow us to represent commuting sets of measurements as pairs of POVMs in tensor product\footnote{Representing commuting measurements as tensor-product measurements requires the Tsirelson conjecture to hold, which in practice amounts to restricting the model to systems of finite dimension or to measurements with finitely many outcomes~\cite{Scholz2008}.}. This avoids many of the subtleties that arise when equating locality with the commutativity assumption\footnote{See some relevant reviews~\cite{Fritz2012,Berkovitz2016,Fewster2016}
of these issues and the sources within.}.

For example, in an experiment where Alice and Bob are in local labs but share a bipartite system, the model of the experiment is again formally the same for QT, RQT and RNQT:
\begin{equation}
    p(a,b|x,y) = \TrX{}{\rho_{\Aa\Bb}\: (A_{a|x} \otimes B_{b|y})} \:.
\end{equation}
In this case, the joint POVM is now $\{A_{a|x} \otimes B_{b|y}\}_{a,b}$ where $\{A_{a|x}\}_a$ (respectively, $\{B_{b|y}\}_b$) is the local POVM picked by Alice (resp., Bob) in her (resp., his) lab, and where $\otimes$ is (a matrix representation of) the tensor product for the theory at hand (for QT and RQT it is represented by $\ktensor$, the Kronecker product; for RNQT it is represented by another bilinear form noted $\rtensordual$).

Finally, by `\textbf{independent sources}', we mean here the constraint imposed on the model to reflect the experimental fact that certain systems are assumed uncorrelated. For example, when the systems are locally prepared in spacelike-separated laboratories. As discussed in the main text, this idea has two distinct translations into a mathematical constraint on the states modelling the independent sources. One possibility is to impose independence directly at the level of the theory, requiring the states to factor into a product state; we refer to this as `\textit{product-state independence}' (\cref{def:LIR}).  The alternative is to define independence operationally, by requiring the measurement statistics to be uncorrelated; we refer to this as `\textit{operational independence}' (\cref{def:LIQ}). 
\begin{definition}[Independent sources]\label{def:independent_sources}
    Let there be a model that assumes an $n$-partition between parties A(lice), B(ob), ..., and underlaid by a theory in which the tensor product is noted $\otimes$ such that the joint distributions of outcomes predicted by the model have the form
    \begin{equation}
        p(a,b,...|x,y,...) = \TrX{}{\rho_{\Aa\Bb ... } \: (A_{a|x} \otimes B_{b|y} \otimes \ldots )} 
    \end{equation}
    for any joint state of the shared system $\rho_{\Aa\Bb ... }$ and for any POVMs $\{A_{a|x}\}_a$,  $\{B_{b|y}\}_b$, ...
    Then, the model assumes \textbf{product-state independent states} with respect to the $n$-partition if and only if the joint state has the form
    \begin{equation}
        \rho_{\Aa\Bb ... } = \rho_\Aa \otimes \rho_\Bb \otimes ... \:,
    \end{equation}
    where $\rho_\Aa$ is a valid state of Alice's subsystem, $\rho_\Bb $ is one of Bob's, etc.

    Alternatively, the model assumes \textbf{operationally independent states} with respect to the $n$-partition if and only if the distribution is uncorrelated for any local measurement with respect to the $n$-partition. That is, if and only if there exists $n$ conditional distributions $p(a|x)$, $p(b|y)$, ... such that
    \begin{equation}
        p(a,b,...|x,y,...) = p(a|x)p(b|y) ... \:.
    \end{equation}
\end{definition}

Like locality, be aware that independence is defined with respect to a partitioning of the subsystems. Be aware that some models tend to have different partitioning between systems preparation and the subsequent local measurements. For example, in the Bell scenario, the bipartite state is global during preparation, i.e. it comes from a single independent source prepared w.r.t. the partition (AB), whereas there are two measurements local with respect to the partitions (A)(B). A more complicated scenario is the bilocal scenario~\cite{Branciard2010} (the one considered by Renou et al.~\cite{Renou2021}): it models a four-partite system which is prepared by two independent sources w.r.t. the partition (AB$_1$)(B$_2$C) and then locally measured by three parties w.r.t. the partition (A)(B$_1$B$_2$)(C).

As mentioned in the main text, the two meanings of `independent sources' are interchangeable in models assuming QT (or the equivalent RNQT) without loss of generality. But the same cannot be said about RQT~\cite{Caves2001}.
\begin{prop}\label{prop:independent_sources}
    In QT, the independent sources assumption is equivalently represented by a model with either product-state or operationally independent states.
    
    In RQT, product-state independence is only sufficient for operational independence, making the independent sources assumption stronger when represented by a model with product states rather than operationally independent states.
\end{prop}
\begin{proof}
    We prove it in the bipartite case; the multipartite case is a straightforward generalisation.
    
    It is direct to see that product-state independence is sufficient for operational independence since
    \begin{equation}
        \TrX{}{(\rho_\Aa \ktensor \sigma_\Bb)(A_{a|x} \ktensor B_{b|y})} = \TrX{}{\rho_\Aa \: A_{a|x}} \TrX{}{\sigma_\Bb \: B_{b|y}} \:,
    \end{equation}
    for every pair of local states and POVMs, leading to the operational independence condition $p(a,b|x,y)= p(a|x)p(b|y)$.

    To show that the condition is necessary for QT, we can use the affine decomposition of bipartite states, $\rho_{\Aa\Bb} = \sum_i q_i (\rho_i \ktensor \sigma_i)$ where $\{\rho_i\}_i$ (respectively, $\{\sigma_i\}_i$) is a collection of local states of Alice's system (resp., Bob) and where $\{q_i\} \subset \rr$ are affine coefficients. Injecting this into a local model, one gets
    \begin{equation}
        \TrX{}{\rho_{\Aa\Bb}(A_{a|x} \ktensor B_{b|y})} = \sum_i q_i \TrX{}{\rho_i A_{a|x}} \TrX{}{\sigma_i B_{b|y}} = \sum_i q_i p(a|x,i)p(b|y,i) \:.
    \end{equation}
    Because the choice of POVMs are arbitrary, one can always pick one such that there is a subset $\{A_{a'|x} \ktensor B_{b'|y}\}_{a',b'} \subseteq \{A_{a|x} \ktensor B_{b|y}\}_{a,b} $ of the POVM elements for which $A_{a'|x} \ktensor B_{b'|y} = c_i \rho_i \ktensor \sigma_i $ with $c_i$ a positive constant. Because of that, the only way for $\sum_i q_i p(a|x,i)p(b|y,i) = p(a|x)p(b|y)$ to be true for every POVM is to require that $\{q_i\}_i$ only possesses one non-zero member, meaning that $\TrX{}{\rho_{\Aa\Bb}(A_{a|x} \ktensor B_{b|y})} = p(a|x)p(b|y)$ only if $\rho_{\Aa\Bb} = (\rho_\Aa \ktensor \sigma_\Bb)$ for some states $\rho_\Aa$ and $\sigma_\Bb$ of Alice and Bob.
    
    For RQT, the above necessity proof no longer holds because the affine decomposition requires a closed number field. In other words, the affine decomposition cannot be performed over $\rr$, as it may involve local matrices that are positive semi-definite and trace-one but that cannot be expressed as purely real matrices, making them invalid RQT states. An example of such a bipartite state of two two-dimensional systems is given by any state of the form~\cite{Caves2001}:
    \begin{multline}\label{eq:II_alphaJJ}
        \frac{1}{4}\left(I \ktensor I - \alpha J \ktensor J\right) \\
        = \frac{1}{2} \left[\frac{1}{2}(I + i\sqrt{\alpha}J) \ktensor \frac{1}{2}(I + i\sqrt{\alpha}J) \right] + \frac{1}{2} \left[\frac{1}{2}(I - i\sqrt{\alpha}J) \ktensor \frac{1}{2}(I - i\sqrt{\alpha}J) \right] \:,
    \end{multline}
    with $0 \leq \alpha \leq 1$, $I = \left( \begin{smallmatrix} 1 & 0 \\ 0 & 1 \end{smallmatrix}\right)$ and $J = \left( \begin{smallmatrix} 0 & -1 \\ 1 & 0 \end{smallmatrix}\right)$. As emphasised by the presence of an $i$ in the right-hand side of the equation, the decomposition can only be done using complex matrices, thus within QT. This matrix is thus separable when seen as a QT state, but entangled when seen as an RQT state. However, if the only accessible local POVMs are RQT ones (i.e., positive semi-definite real matrices), remark that $\TrX{}{( J \ktensor J)(A_{a|x} \ktensor B_{b|y})} = 0$ because any pair from $\{A_{a|x}\}_a$ and $\{B_{b|y}\}_{b}$ are real symmetric matrices, which are orthogonal to $J$, since $J$ is an antisymmetric matrix. Hence,
    \begin{equation}
        \TrX{}{\frac{1}{4}\left(I \ktensor I - \alpha J \ktensor J\right)(A_{a|x} \ktensor B_{b|y})} = \frac{1}{2}\TrX{}{A_{a|x}} \: \frac{1}{2}\TrX{}{B_{b|y}} \:.
    \end{equation}
    While this RQT state is not a product state, it is yet an operationally independent one. This provides a counterexample to product-state independence being necessary for operational independence in the RQT case, concluding the proof.
\end{proof}

\subsection*{Single-source experiments have an RQT description}

We start by reminding the reader that in experiments involving a single source and a single measurement, QT can always be alternatively explained by RQT. This has been known since the beginning of quantum theory and formally established by the seminal work of St{\"u}ckleberg~\cite{Stueckelberg1960} (who defined a theory that can be seen as single-partite RNQT; see also the relevant but more abstract approach taken by Dyson~\cite{Threefold}). 
This construction was later studied by McKague et al. as a way to prove that a single-partite RQT model always contains a sub-model isomorphic to a QT model~\cite{McKague}. 
We recover this result through the following (see \cref{sec:RNQTsingle} for a proof).
\begin{restatable}[Single-partite QT models are sub-models of single-partite RQT models.]{prop}{singleRQTissingleQT}\label{prop:singleRQTissingleQT}
    A single-partite QT model involving a single source producing a state $\rho$ measured by a POVM $\{A_{a|x}\}_a$ can always be embedded into an RQT model of twice the matrix size.
    
    That is, there exists a pair of real-linear injective maps $(M,E) : \cc^{d_A \times d_A} \rightarrow \rr^{2d_A\times 2d_A}$ such that for any QT state $\rho$, $M(\rho)$ is a valid RQT state, such that for any QT POVM $\{A_{a|x}\}_a$, $E(A_{a|x})$ is a valid RQT POVM effect and $\{E(A_{a|x})\}_a$ a valid RQT POVM, and such that every distribution of the QT model is mapped to one of the RQT model as:
    \begin{equation}
        \forall a: \quad \TrX{}{\rho A_{a|x}} = \TrX{}{M(\rho) E(A_{a|x})} \:.
    \end{equation}
\end{restatable}
The image of the QT model through the pair of maps $(M,E)$ is what we refer to as the `QT-image' of the RQT. The QT-image normally differs from the RNQT model by its system-composition rule. But since single-partite models are not concerned with the system-composition rule, these are the same. Hence, any single-partite, single-source RQT model always contains a sub-model isomorphic to a QT model of half the matrix size. It means that any distribution of outcomes obtained with QT states and effects could have been obtained with RQT states and effects of twice the matrix size. Therefore, any single-source, single-partite QT model can be alternatively explained by an RQT model, making the two indistinguishable from experimental statistics alone.

But what about multipartite models\footnote{Remark that this question arises because the tensor product structure is assumed. If the action of local parties were represented by subsets of commuting maps all defined over the same Hilbert space instead, we could simply use \cref{prop:singleRQTissingleQT} to prove that every QT model has an explanation as an RQT model. This would be the case because the pair of maps $(M,E)$ is actually a *-algebra homomorphism and thus preserves sets of commuting maps.}? If these are viewed as global scenarios without any locality constraints, they fall under the umbrella of the above proposition. In the presence of locality and independence constraints, we face the problem that there does not exist a pair of maps that can preserve all pure Kronecker products of states and effects while keeping the distributions invariant. This is what motivated us to change the system-combination rule of RQT to define RNQT~\cite{RNQT}. 
\begin{restatable}[QT models are RNQT models.]{prop}{QTisRNQT}\label{prop:QTisRNQT}
    A $n$-partite QT model can always be embedded into an RNQT model of $2^n$ times the matrix size. 
    
    That is, there exists a pair of real-linear injective maps $(M_{\Aa\Bb ... }, E_{\Aa\Bb \dots}):  \cc^{d_Ad_B ... \times d_Ad_B...} \rightarrow \rr^{2d_A2d_B ...\times 2d_A2d_B...}$ that respectively map any valid QT state and effect to a valid RQT state and effect that commute with a given $J$ matrix, such that every distribution predicted by the QT model has a real-matrix counterpart and such that the system-composition rules of QT are mapped to those of RNQT. I.e.,
    \begin{subequations}\label{eq:Kronecker_to_Rprod}
        \begin{align}
            M_{\Xx\Yy}(\rho_\Xx \ktensor \sigma_{\Yy}) &=  M_{\Xx}(\rho_\Xx) \rtensor M_{\Yy}(\sigma_{\Yy}) \:; \\
             E_{\Xx\Yy}(X_{\alpha|\mu} \ktensor Y_{\beta|\nu}) &=  E_{\Xx}(X_{\alpha|\mu}) \rtensordual E_{\Yy}(Y_{\beta|\nu}) \:; 
        \end{align}
    \end{subequations}
    where $\Xx$ and $\Yy$ are any two subsets of the parties $\Aa\Bb\dots$, $\rho_\Xx$ and $\sigma_{\Yy}$ (respectively, $X_{\alpha|\mu}$ and $Y_{\beta|\nu}$) are valid QT states (resp. effetcs) of $\Xx$ and $\Yy$, and $(M_{\Xx}, E_{\Xx})$ denotes a local application of the maps on the $\Xx$ subset.
\end{restatable}

The proof is presented in \cref{sec:RNQT_model}. Remark that this embedding can always be refined into an isometry by using a factor of $2$ instead of $2^n$ (see our previous work on the subject~\cite{RNQT}). Using the above, 
finding the RQT sub-model equivalent to a QT model is reduced to finding a way to replace the RNQT composition rules on the right-hand sides of \cref{eq:Kronecker_to_Rprod} by the RQT composition rules, i.e., by the Kronecker product $\ktensor$. This is easier than it appears because the replacement should not hold in general, but only within the Born rule and under the locality constraints. For instance, in the local (or Bell-like) scenario, since there is only a single source, one only needs to replace the composition rule on the effect side in every Born rule.

\begin{restatable}[Local QT models are sub-models of local RQT models.]{prop}{localRQTislocalQT}\label{prop:localRQTislocalQT}
    A multipartite QT model with a single source producing a state $\rho_{\Aa\Bb\ldots}$ and $n$ local measurements represented by the POVMs $\{A_{a|x}\}_a$, $\{B_{b|y}\}_b$, ..., can always be embedded into an RQT model of $2^n$ the matrix size.
    
    That is, there exists a pair of real-linear injective maps $(M_{\Aa\Bb ... }, E_{\Aa} \ktensor E_{\Bb} \ktensor ...)$ from $ \cc^{d_Ad_B ... \times d_Ad_B...}$ to $\rr^{2d_A2d_B ...\times 2d_A2d_B...}$ such that $M_{\Aa\Bb ... }(\rho_{\Aa\Bb ... })$ is a valid RQT state of the multipartite system; such that $E_\Xx(X_{\alpha|\mu}) \in \rr^{2d_\Xx \times 2d_\Xx}$ is a valid RQT effect on the local system associated with party $\Xx$ and $\{E_\Xx(X_{\alpha|\mu})\}_\mu$ is a valid RQT POVM for every party $\Xx$; and such that every distribution of the QT model is mapped to one of the RQT model as:
    \begin{multline}
        \forall a, \forall b , ...: \TrX{}{\rho_{\Aa\Bb\ldots} (A_{a|x} \ktensor B_{b|y} \ktensor ... )} \\
         = \TrX{}{M_{\Aa\Bb...}(\rho_{\Aa\Bb\ldots}) (E_{\Aa}(A_{a|x}) \ktensor E_{\Bb}(B_{b|y}) \ktensor ...)} \:.
    \end{multline}
\end{restatable}
The proof is presented in \cref{sec:RNQTlocal}. As mentionned, this does not imply that the mapping verifies $E_{\Aa\Bb\dots}(A_{a|x} \ktensor B_{b|y} \ktensor ...) =  E_{\Aa}(A_{a|x}) \ktensor E_{\Bb}(B_{b|y}) \ktensor ...$, only that they are equivalent when evaluated over every QT-image state. 
But this is sufficient to conclude that the RQT local models can reproduce the predictions of their corresponding QT models. 
Thus, this proposition generalises the known result~\cite{Pal_2008,Moroder_2013} that every single-source Bell-like scenario has an RQT explanation. 

As a remark, notice that the POVMs can be more generally $m$-partite for $m\leq n$. However, one can always regroup the subsystems corresponding to the same local party in order to ensure that $m = n$. 
Also, notice that we use a matrix-size rescaling factor of 2 growing with the power of the number of parties, instead of the factor of 2 used in our previous work~\cite{RNQT}. 
A factor that grows with the number of parties might appear unreasonable; worse, it appears to make the proposition partition-dependent. While this is correct in the absolute, it is merely a formulation chosen to simplify the proof of our main statement. Using the results of our previous work (specifically by replacing the $\Rprod$-product on which this proposition is based by the isomorphic $\otimes_r$-product), one could prove a partition-independent version of it that only requires a factor of 2 between the dimensions of the RQT and QT systems. However, doing so would induce more technicalities to overcome, unnecessarily lengthening the proof. 

\subsection*{Description of multiple sources in RQT}\label{sec:methods_source_certification}

We saw that in experiments involving a single source, every model assuming QT can be explained by a model assuming RQT. But can the same be said about experiments with multiple sources? 
Renou et al. showed that in the bilocal scenario~\cite{Branciard2010}, the models assuming RQT cannot explain all distributions predicted by the corresponding models~\cite{Renou2021} \textit{if independent sources are represented by product states}. 
However, as we argued above and in the main text, their product-state independence assumption is theory-dependent.

In QT, one could use a theory-independent approach to circumvent the issue, representing the joint state of two independent systems via an operational independence rather than product-state independence (we defined these two notions in \cref{def:independent_sources}). 
As we showed in \cref{prop:independent_sources}, the same cannot be done in RQT: the product-state independence is stronger than operational independence. This implies that not all independent state-preparation procedures correspond to product states in RQT. Requiring the state of independent sources to be in a product form, therefore, amounts to restricting the set of states that the parties can prepare solely on the basis of model-dependent considerations. 

\begin{figure}[t]
    \centering
    \begin{tikzpicture}[x=.07\linewidth,y=.07\linewidth]

        \fill[setoneblue,draw=black,rounded corners=0.6cm]
            (0,0) rectangle (14,8);

        \node[align=center] at (7,7.3)
            {Set 1: All bipartite states in RQT between parties $\Xx$ and $\Yy$.};
\begin{scope}[shift={(0,-0.4)}]  
        \fill[settwopeach,draw=black,rounded corners=0.5cm]
            (0.8,1.2) rectangle (13.2,7.1);

        \node[align=center,text width=8cm] at (7,6.5)
        {\singlespacing\noindent Set 2: RQT operationally independent states.\par{\setstretch{0.5}\footnotesize States for which all local \mbox{measurements} lead to factorising probabilities.\par}};

        \fill[settwoyellow,draw=black]
            (3.8,3.4) ellipse (2.7 and 1.9);

        \node[align=center,text width=5.2cm] at (3.8,3.8)
        {\singlespacing\noindent Set 2.1: RQT product states\smallskip\newline{\centering \setstretch{0.5} $\rho_{\Xx\Yy} = \rho_\Xx \ktensor \rho_\Yy$.\par}};

        \draw[thick, dashed, red] (3.8,2.9) ellipse (3.2 and 2.4);
        \node[anchor=west, align=center, text width=6cm, text=red] at (5.5,1.1)
        {\singlespacing\noindent Set 3: `partially independent' states};

        \fill[settwogreen,draw=black]
            (10.2,3.4) ellipse (2.7 and 1.9);

        \node[align=center,text width=5.2cm] at (10.3,3.8)
        {\singlespacing\noindent Set 2.2: 
        RNQT product states\smallskip\newline {\centering \setstretch{0.5} $\rho_{\Xx\Yy} = \rho_\Xx \rtensor \rho_\Yy$.\par}};

\end{scope}
    \end{tikzpicture}
    \caption{{\bf Relations between sets of bipartite states in real quantum theory (RQT).} \Cref{prop:independent_sources} shows that in RQT, the sets 2 and 2.1 are disjoint. \Cref{prop:locallyIndistinguishable} further shows that many states are locally indistinguishable from the product states. An example of such locally equivalent non-product-yet-non-correlated states is provided by set 2.2: the image of QT product states in RNQT (set 2.2) is valid set of RQT states (within set 1) that are operationally independent (within set 2), but distinct from the set of product states (outside of set 2.1), yet locally indistinguishable from it.
    The red circle, set 3, sketches the relaxation of the product-state assumption considered in Refs.~\cite{Renou2021,Weilenmann2025}. Our argument is that the operationally motivated relaxation of set 2.1 should rather be set 2.
    }
    \label{fig:VennDiagram}
\end{figure}
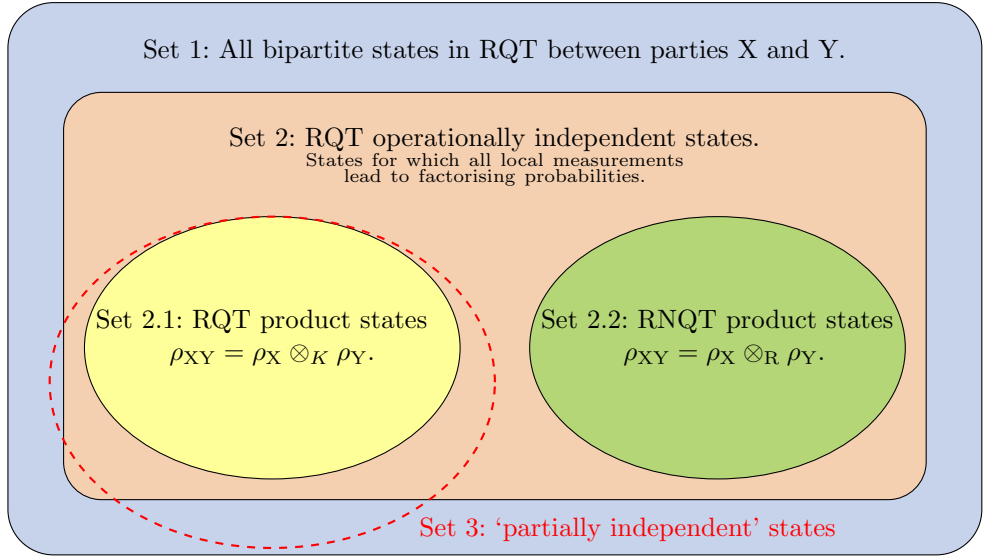

To what extent is it a problem? To begin with, it implies that the sources cannot be locally certified: in RQT, there are states that cannot be distinguished from the product states by local measurements alone. An important class of these are the operationally independent RQT states that are in the QT-image as in \cref{prop:localRQTislocalQT}. Even when they are entangled, these states must correspond to a QT product state because of \cref{prop:independent_sources}. However, the RQT product states must already correspond to the QT product states, so it implies that more than one state locally describes the same situation.
\begin{prop}\label{prop:locallyIndistinguishable}
    Let $\rho_{\Aa\Bb\dots}$ be an $n$-partite RQT state that is operationally independent as in \cref{def:independent_sources}. 
    Then, there exists $n$ single-partite RQT states $\rho_{\Aa}$, $\rho_{\Bb}$, $\dots$ such that $\rho_{\Aa\Bb\dots}$ is locally indistinguishable from $\rho_{\Aa} \ktensor \rho_{\Bb} \ktensor \dots$

    That is, for all RQT POVMs $\{A_{a|x}\}_a, \{B_{b|y}\}_b, \dots$, the following equality holds:
    \begin{equation}\label{eq:locallyIndistinguishable}
        \TrX{}{\rho_{\Aa\Bb\dots}(A_{a|x} \ktensor B_{b|y} \ktensor \dots)} = \TrX{}{(\rho_{\Aa} \ktensor \rho_{\Bb} \ktensor \dots)(A_{a|x} \ktensor B_{b|y} \ktensor \dots)}\:.
    \end{equation}
\end{prop}
\begin{proof}
    If $\rho_{\Aa\Bb\dots}$ is operationally independent, it verifies $\TrX{}{\rho_{\Aa\Bb\dots}(A_{a|x} \ktensor B_{b|y} \ktensor \dots)} = p(a|x)p(b|y)\dots$. One can then set $\rho_{\Aa} = \TrX{\Bb\textup{C}\dots}{\rho_{\Aa\Bb\dots}},\rho_{\Bb} = \TrX{\Aa\textup{C}\dots}{\rho_{\Aa\Bb\dots}},$ etc. such that $p(a|x) = \TrX{}{\rho_{\Aa} A_{a|x}}$, $p(b|y) = \TrX{}{\rho_{\Bb} B_{b|y}}$, etc., which implies that $ p(a|x)p(b|y)\dots = \TrX{}{(\rho_{\Aa} \ktensor \rho_{\Bb} \ktensor \dots)(A_{a|x} \ktensor B_{b|y} \ktensor \dots)}  = \TrX{}{\rho_{\Aa\Bb\dots}(A_{a|x} \ktensor B_{b|y} \ktensor \dots)} $.
\end{proof}
This does not happen in locally tomographic theories like QT. In that case, \cref{eq:locallyIndistinguishable} necessarily implies $\rho_{\Aa\Bb\dots} = \rho_{\Aa} \ktensor \rho_{\Bb} \ktensor \dots$ because operational independence implies product-state independence. But for RQT, the product states might be locally indistinguishable from entangled states.

A class of such locally indistinguishable entangled states is actually given by \cref{prop:QTisRNQT}: the image of QT product states, i.e., the RNQT product states which are of the form $\rho_\Aa \rtensor \rho_\Bb \rtensor \dots$, are indistinguishable from the corresponding RQT product states. Indeed, for every RQT POVMs, $\{A_{a|x} \}_a$, \{$B_{b|y}\}$, etc., one has (see the corollaries of \cref{lemm:Rprod_contaminates}):
\begin{multline}
    \TrX{}{(\rho_{\Aa} \rtensor \rho_{\Bb} \rtensor \dots)(A_{a|x} \ktensor B_{b|y} \ktensor \dots)} \\
    = \TrX{}{(\rho_{\Aa} \ktensor \rho_{\Bb} \ktensor \dots)(A_{a|x} \ktensor B_{b|y} \ktensor \dots)}\:.
\end{multline}
Yet, the state $\rho_{\Aa} \rtensor \rho_{\Bb} \rtensor \dots$ is actually entangled with respect to the Kronecker product (the bilinear map $\rtensor$ is entangling; see the \cref{def:Rprod} for the exact form of this map). 

This proposition entails that assuming the systems have been prepared in a product state cannot be certified locally. In a local experiment, it only makes sense to require independent preparation procedures to yield operationally independent states, since being uncorrelated under every local measurement is the only experimentally accessible way to define independence to begin with\footnote{One may still remark that, using global tomography, product states are distinguishable from non-product but operationally independent ones. While this is true, it does not imply that the latter states are correlated in any observable manner; it only means they differ from the former in a way that can be revealed only by using correlated measurements. This is nonetheless an intriguing property of RQT, discussed in more detail in \Cref{sec:RQTisnotQT}, and one might use it to study a stronger notion of independence using entropic or resource considerations.}. Yet, Renou et al.'s proof \textit{requires} the state of independent sources to be a product state. This uncertifiable assumption amounts to removing from the model every non-product state in the equivalence classes implied by \cref{prop:locallyIndistinguishable}. But if a state is locally indistinguishable from a product state, why should it not correspond to a valid preparation procedure of independent sources as well, since it is as uncorrelated as a product state can be?  

Aware of this limitation, the authors of Ref.~\cite{Renou2021} allowed separable but non-product states in their proof. Going further, some of these authors sought to strengthen the result by studying how allowing entangled states to correspond to independent preparation procedures narrows the gap between QT and RQT for the Renou et al. game~\cite{Weilenmann2025}. They did not realise that some RQT entangled states---the operationally independent ones---correspond to uncorrelated systems and thus verify the heuristics of independent preparation procedures. 
These are precisely the states of independent systems needed to show (using our proof technique) that the RQT models can always explain the prediction of the QT models, as we will prove in the next section. What the Renou et al. result~\cite{Renou2021} actually showed, therefore, is that in the RQT model of entanglement swapping scenarios, the optimal state to saturate their bilocal game is not of the Kronecker-product form. The only conclusion that can be drawn from the experimental data~\cite{Li2022,Lancaster2025}, therefore, is that the RQT description of the state in which the system started the experiment was not a product state.

\subsection*{Multiple-source experiments have an RQT description}
As announced, RQT models with operationally independent sources can explain QT models. Our result holds for every possible network scenario, which we will call `multilocal' as a shorthand for `local experiments with multiple sources'. Since single-source networks do not require defining independent sources, this result recovers the local and single-partite cases previously mentioned as special cases.
\begin{restatable}[Multilocal QT models are sub-models of multilocal RQT models.]{prop}{MulitilocalRQTisMultilocalQT}\label{prop:MulitilocalRQTisMultilocalQT}
    A multipartite QT model involving $n$ parties, $m\leq n$ measurements, and $p\leq n$ operationally independent sources, can always be embedded into an RQT model of $2^n$-th the matrix size.

    Specifically, let $\mathcal{N}:=\{A,B,...\}$ be the set of parties, $\abs{\mathcal{N}} = n$; let $\mathcal{M} = \{M_1,M_2,..., M_m\}$ be the set of subsets of $\mathcal{N}$ representing how the parties are regrouped to perform the $m$ measurements (i.e., $M_i \subseteq \mathcal{N}$ and $\bigcup_{i=1}^m M_i = \mathcal{N} $); and let the sources be represented by a joint state $\rho_{\Aa\Bb\ldots}$ operationally independent w.r.t. the given partitionning into $p$ independent subsets of the parties.
    Then, there exists a pair of map $(M_{\Aa\Bb\dots}, E_{\Aa\Bb\dots})$ from $ \cc^{d_Ad_B ... \times d_Ad_B...}$ to $\rr^{2d_A2d_B ...\times 2d_A2d_B...}$ such that $M_{\Aa\Bb\dots}$ injectively sends the QT state to an RQT state while preserving operational independence, such that  $E_{\Aa\Bb\dots}$ injectively sends the QT POVMs to RQT POVMs while preserving locality (i.e., $E_{\Aa\Bb\dots} = E_{M_1} \ktensor E_{M_2} \ktensor \dots \ktensor E_{M_m}$), and such that the pair of map preserves all inner products between any states and effects.
\end{restatable}
The proof combines all the points discussed so far: it relies on showing that for every QT model, the equivalent RNQT model, given by \cref{prop:QTisRNQT}, can be turned into a sub-model of an RQT model. To do so requires combining three observations: first, that the RNQT model is a sub-model of the RQT model but with different system-composition rules. Second, that the RNQT composition rule of effects $\rtensordual$ can always be turned into a RQT composition rule $\ktensor$, like in \cref{prop:localRQTislocalQT}. And third, that the RNQT product states correspond to operationally independent RQT states, as discussed in the previous section. The technical phrasing of this intuitive proof is provided in \cref{sec:RNQTmultilocal}.

This proposition immediately implies our first result, that any network scenario with a QT model has a corresponding RQT model.
\QTnetworksareRQTnetworks*
\begin{proof}
    Since the dimension of systems is experimentally impossible to certify~\cite{Brunner_2008}, every network experiment which has a QT model can be alternatively explained by its corresponding RQT model of $2^n$ times the matrix size given by \cref{prop:MulitilocalRQTisMultilocalQT}.
\end{proof}

Naturally, the network scenarios considered so far may not represent every conceivable locality experiment, in the sense of every distributed quantum information processing task. For instance, we have not yet considered adaptive strategies over several rounds of operations, in which the subsystems are redistributed among the parties between rounds (and possibly destroyed or created). 

Fortunately, this does not change our conclusion: regardless of the observed statistics, if the state of the subsystems after each round and the linear maps representing the operations at a given round have a QT description, they also have an RQT description.
\RQTisQT*
The proof is presented in length in \cref{sec:proof}.

\bibliography{references} 






\crefalias{section}{appendix}

\newpage
\appendix
\appendixpage
\addcontentsline{toc}{part}{Appendices}

\startcontents[sections]
\printcontents[sections]{l}{1}{\setcounter{tocdepth}{2}}


\section{Real representation of quantum networks}\label{sec:RNQT}
The proofs of the main statements of this paper rely on the fact that models assuming real quantum theory (RQT) and systems of dimension $2d$ will always contain a sub-model (i.e., a subset of states and effects) that is isometrically isomorphic to a model assuming quantum theory (QT) and systems of at most dimension $d$. These QT-isomorphic submodels within RQT can be seen as models assuming an alternative theory isomorphic to QT, which we call \textbf{real number quantum theory (RNQT)}. We developed it in a previous work~\cite{RNQT}. 

To simplify the proofs in this work, we will not use the RNQT model of the larger possible dimension. Rather, for an $n$-partite RQT model, we will fix the dimension of its system to be $2d_A 2d_B 2d_C ... = 2^n d_A d_B d_C...$ and we will focus on the RNQT sub-model isomorphic to a QT model with systems of dimension $d_Ad_Bd_C...$, rather than the `biggest' possible RNQT sub-model (i.e., the one isomorphic to a QT model with systems of dimension $2^{n-1}d_Ad_Bd_C...$). The reason for considering the `small' model is that it is better behaved w.r.t. subsystems, making subsystem-wise proofs easier. 

\subsection{Single-partite models and the standard mapping}\label{sec:RNQTsingle}
We begin by recalling the `standard' mapping between complex and real matrices and reviewing its relevant properties. This is the mapping upon which the isomorphism between RNQT and RQT will be built. 
\begin{definition}[Standard mapping]\label{def:std_mapping}
    Let $I$ and $J$ be the matrices in $\rr^{2 \times 2}$ defined as $I = \left(\begin{smallmatrix}
        1 & 0 \\ 0 & 1
    \end{smallmatrix}\right)$ and $J = \left(\begin{smallmatrix}
        0 & -1 \\ 1 & 0
    \end{smallmatrix}\right)$. 
    Then, we refer to the following as the \textbf{standard (complex-to-real) mapping}:
    \begin{equation}
    \begin{aligned}
        \Gamma : \cc^{d \times d} &\rightarrow \rr^{2d \times 2d} \:, \\
        A & \mapsto \G{A} :=  I \ktensor \ReP{A} + J \ktensor \ImP{A} \:.
    \end{aligned}            
    \end{equation}
\end{definition}
Remark that the standard mapping is very common in the mathematics literature, so we skip the proofs of its properties (the reader can find them in our previous work if needed~\cite{RNQT}). 
\begin{prop}[Properties of the standard mapping]\label{prop:gamma_properties}
    The map $\Gamma$ is a $^*$-algebra homomorphism, i.e., $\forall A,B \in \cc^{d\times d}$, $\forall a,b \in \rr$
    \begin{subequations}
            \begin{gather}
                \G{a A + bB } = a\G{A} + b\G{B} \:; \\
                \G{AB} = \G{A}\G{B} \:; \\
                \G{A^\dag} = \G{A}^T \:.
            \end{gather}
    \end{subequations}
    In addition, when the vector spaces are endowed with their standard Frobenius inner products and the domain of $\Gamma$ is restricted to the self-adjoint matrices, then $\Gamma$ is
    \begin{subequations}
    \begin{enumerate}
        \item positive-(semi-)definitness-preserving,
        \begin{equation}
            A \geq 0 \quad \Rightarrow \quad \G{A} \geq 0 \:; 
        \end{equation}
        \item unital,
        \begin{equation}
            \G{\id_d} = \id_{2d} \:;
        \end{equation}
        \item trace-rescaling,
        \begin{equation}
            \TrX{}{A} = \frac{1}{2}\TrX{}{\G{A}} \:; 
        \end{equation}
        \item and isometric up to a constant, i.e.,
        \begin{equation}
            \TrX{}{A^\dag \: B} = \frac{1}{2} \TrX{}{\G{A}^T \: \G{B}} \:.
        \end{equation}
    \end{enumerate}
    \end{subequations}
\end{prop}
\begin{prop}\label{prop:Herm_image}
    The image of the real-linear subspace of Hermitian matrices $\HermC{n}$ through the standard mapping is the subspace of real matrices characterised by the following:
    \begin{equation}
        A \in \cc^{d \times d}: \quad A^\dag = A \iff \G{A} = \G{A}^T \:\wedge \: J_d\G{A} = \G{A}J_d \:,
    \end{equation}
    where we used the notation $J_d = J \ktensor \id_d$ \:.
\end{prop}
In this case, $J_d$ is called a complex structure in $\rr^{2d\times 2d}$ since it plays the role of the scalar multiplication by $i$, $\G{i\id_d} = J_d$. For a real matrix to be in the image of the complex theory, it must be well-behaved with respect to the complex structure, meaning it should commute with it. We call these the `tetradionic' or `special' matrices, and, in particular, we defined~\cite{RNQT}
\begin{definition}[Special Symmetric matrices]\label{def:special_sym}
    The subspace of symmetric matrices in $\rr^{2d \times 2d}$ which commute with $J_d$ as in \cref{prop:Herm_image} is called \textbf{special symmetric} and noted $\SYR{d}$.
\end{definition}
This subspace is the image through $\Gs$ of the complex Hermitian matrices, $\G{\HermC{d}} = \SYR{d}$. It is straightforward to check that both are $d^2$ real vector spaces and thus in bijection. Because of the above, this bijection is furthermore an isometry of inner product spaces (w.r.t. the Frobenius inner product), so their positive (semi-)definite and trace-normalised subsets will also be in bijection. 

More generally, we call the elements of RQT that are well-behaved with the complex structure its \textbf{QT-image}. The QT-image of an RQT model with a single system of dimension $2d$ is therefore its RNQT sub-model, and thus isomorphic to a QT model with a single system of dimension $d$. 
\singleRQTissingleQT*
\begin{proof}
    Let the pair of mappings be defined from the standard mapping as 
    \begin{subequations}
        \begin{align}
            & M_A = \frac{1}{2}\Gamma \:;\\
            & E_A = \Gamma \:.
        \end{align}
    \end{subequations}
    Then, by \cref{prop:Herm_image}, the QT-image is made of all the special symmetric states and effects, i.e. the states and effects that commute with the matrix $J_{d_A} = J \ktensor \id_{d_A}$. 
    From \cref{prop:gamma_properties} $M_A$ preserves the trace and positive semi-definiteness, making it a bijection between the QT-image subset of the RQT states and the corresponding set of QT states. 
    Similarly, $E_A$ is a bijection between the QT-image subset of the RQT effects and the set of QT effects because it is positive semi-definiteness-preserving. It furthermore maps RQT POVMs to QT POVMs since $ E_A(\sum_a A_a) = \sum_a E_A(A_a)$ follows from the map being an algebra homomorphism and $E_A(\id_{2d_A}) = \id_{d_A}$ follows by unitality. 
    Finally, it also follows from \cref{prop:gamma_properties} that the pair of bijections define an isometry:
    \begin{equation}
        \TrX{}{M(\rho)E(A_{a|x})} = \frac{1}{2}\TrX{}{\Gamma\{\rho\}^T\Gamma\{A_{a|x}\}} = \TrX{}{\rho^\dag \: A_{a|x}} = \TrX{}{\rho \: A_{a|x}}\:.
    \end{equation}
\end{proof}
Simply put, the single-partite RNQT can be seen as restricting the RQT to its QT-image, i.e., the image of states and effects through $(\frac{1}{2}\Gs, \Gs )$.

\subsection{Local models and the \textit{n}-fold standard mapping}
We now extend the proof by showing that the local RNQT models are always embedded in local RQT models. By `local models', we mean the models representing network scenarios that assume a single source and several local (i.e. spacelike) measurements. The source produces a multipartite system whose parts are shared between the locally measuring parties.

\subsubsection{The \textit{n}-fold standard mapping}
When dealing with multipartite systems, the bijection from quantum theory to its real image provided by the standard mapping may cause trouble. At first glance, $\Gs$ requires an extra 2-by-2 matrix to encode the phase of a system, which one can physically interpret as a fiducial system~\cite{McKague}. But what happens when several parties are involved? For example, the state of a system shared by Alice and Bob will live in $\G{\cc^{n \times n}_A \otimes \cc^{m \times m}_B} = \rr^{2\times 2} \otimes \rr^{n \times n}_A \otimes \rr^{m \times m}_B$, raising the question of to which party the extra system $\rr^{2\times 2}$ should be associated? 
Of course, this is purely a problem of representation, which can be circumvented. Recall that the fiducial system simply encodes the (complex) phase freedom of quantum theory, so one should not attach too much ontology to it. Actually, interpreting it as a physical system shared by all parties would be equivalent to interpreting the phase freedom of (the pure-state representation of) quantum theory (using Hilbert spaces) as a physical phenomenon. 

So one can instead promote the globally shared phase-carrying fiducial system into a delocalised representation using a multipartite fiducial system~\cite{McKague}. To do so, instead of using a single $2 \times 2$ matrix for the representation of $(1, i)$, we pick a larger dimensional representation, say a 4-dimensional representation $(\I{2},\J{2}) \subset \rr^{4 \times 4}$ such that the extra system can be split evenly between Alice and Bob, $\rr^{4 \times 4} \cong \rr^{2 \times 2}_A \otimes \rr^{2 \times 2}_B$. The important property that must be required from this representation is that it is global, in the sense explained in the previous section: suppose the parties share such a bipartite element, e.g. $\J{2}_{AB}$, then the action of an element of a local 2-d representation acting on it from Alice's side, e.g. $(J_A \ktensor I_B)\J{2}_{AB}$, should be equivalent to the situation in which it is acting on Bob's side, but also if it were acting globally, e.g. $(J_A \ktensor I_B)\J{2} = (I_A \ktensor J_B)\J{2}_{AB} = \J{2}_{AB}\J{2}_{AB}$. This stabiliser-state-like representation is precisely the kind of representation reviewed in the \cref{sec:real_units}.

\begin{definition}[$n$-fold standard mapping]\label{def:std_n_mapping}
    Let $\I{n}$ and $\J{n}$ be defined as in \cref{def:I_J_n}. 
    Then, the following is called the \textbf{\textit{n}-fold standard (complex-to-real) mapping}:
    \begin{equation}\label{eq:Gamma_n}
    \begin{aligned}
        \overline{\Gamma}^{(n)} : \bigotimes_{i=1}^n \cc^{d_{A_i} \times d_{A_i}} &\rightarrow \bigotimes_{i=1}^n \rr^{2d_{A_i} \times 2d_{A_i}} \:, \\
        A & \mapsto \overline{\Gamma}^{(n)}\{A\} :\cong  \I{n} \ktensor \ReP{A} + \J{n} \ktensor \ImP{A} \:.
    \end{aligned}            
    \end{equation}
\end{definition}
We use $:\cong$ instead of $:=$ in the definition because there has been an implicit reordering of the tensor factors from $\bigotimes_{i=1}^n \rr^{2d_{A_i} \times 2d_{A_i}}$ to $\rr^{2^n \times 2^n} \otimes (\bigotimes_{i=1}^n \rr^{d_{A_i} \times d_{A_i}})$. Specifically, every phase-carrying $\rr^{2 \times 2}$ subspace has been moved to the left of the tensor factorisation in order to obtain a concise expression that is reminiscent of \cref{def:std_mapping}.

With the implicit reordering and because the matrices $\I{n}$ and $\J{n}$ are isomorphic to $I$ and $J$, it should be clear that the $n$-fold standard mapping is isomorphic to the standard mapping. Therefore, almost all of the properties proven in \Cref{prop:gamma_properties} apply to it as well. For example, $\TrX{}{\overline{\Gamma}^{(n)}\{A\}} = 2\TrX{}{A}$. The only difference between the standard mapping and its $n$-fold version is that the matrices $I$ and $J$ have been `inflated' into the larger matrices $\I{n}$ and $\J{n}$. As mentioned, this inflation amounts to delocalising the phase-carrying subsystem on $n$ subsystems, like in an encoding of a stabiliser state (or any other quantum error-correcting code). They nonetheless represent the same thing: the complex structure $(1,i)$. In that sense, the matrices $\overline{\Gamma}^{(n)}\{A\}$ will not only be compatible with a \textit{global} complex structure $\J{n}_{d_{A_1}d_{A_2}\dots d_{A_n}}$, it will also be compatible with the \textit{local} complex structure of each party, i.e. the matrices $J_{A_1,2_{A_2}\dots 2_{A_n} d_{A_1}d_{A_2}\dots d_{A_n}}, J_{A_2,2_{A_1}2_{A_3}\dots 2_{A_n} d_{A_1}d_{A_2}\dots d_{A_n}}, \dots$ (recall that the notation $J_{A_1,2_{A_2}\dots 2_{A_n} d_{A_1}d_{A_2}\dots d_{A_n}}$ means $J_{A_1} \ktensor I_{A_2} \ktensor \dots I_{A_n} \ktensor \id_{d_{A_1}} \ktensor \id_{d_{A_2}} \ktensor \dots \ktensor \id_{d_{A_n}}$).

That way, we can define the image of complex matrices representing multipartite systems while keeping each system locally accessible. All single-fold results and definitions then readily generalise to the $n$-fold case:
\begin{definition}[$n$-partite Special Symmetric matrices]\label{def:special_sym_global}
    In the $n$-fold tensor product space $\bigotimes_{i=1}^n \rr^{2d_{A_i} \times 2d_{A_i}}$, the subspace of symmetric matrices commuting with the matrix $\J{n}_{d_{A_1}d_{A_2}\dots d_{A_n}} \cong \J{n} \ktensor \id_{d_{A_1}} \ktensor \id_{d_{A_2}} \ktensor \dots \ktensor \id_{d_{A_n}}$ is called $n$-partite special symmetric.
\end{definition}
\begin{prop}[Properties of the $n$-fold standard mapping]\label{prop:gamma_n_properties}
    The map $ \overline{\Gamma}^{(n)}$ is a $^*$-algebra homomorphism. 
    In addition, when the vector spaces are endowed with their standard Frobenius inner products, then for any $ A,B \in \bigotimes_{i=1}^n \cc^{d_i \times d_i}$,  $ \overline{\Gamma}^{(n)}$ is
    \begin{subequations}
    \begin{enumerate}
        \item positive-semi-definitness-preserving,
        \begin{equation}
            A \geq 0 \quad \Rightarrow \quad \overline{\Gamma}^{(n)}\{A\} \geq 0 \:; 
        \end{equation}
        \item isomorphic to a unital map ($d = \prod_i d_i$),
        \begin{equation}
            \overline{\Gamma}^{(n)}\{\id_{d}\} \:\cong\: \I{n} \ktensor \id_{d} \:;
        \end{equation}
        \item trace-rescaling,
        \begin{equation}
            \TrX{}{A} = \frac{1}{2}\TrX{}{\overline{\Gamma}^{(n)}\{A\}} \:; 
        \end{equation}
        \item and isometric up to a constant, i.e.,
        \begin{equation}
            \TrX{}{A^\dag \: B} = \frac{1}{2} \TrX{}{\overline{\Gamma}^{(n)}\{A\}^T \: \overline{\Gamma}^{(n)}\{B\}} \:.
        \end{equation}
    \end{enumerate}
    \end{subequations}
\end{prop}

We shall now see how the $n$-fold mapping is better suited to deal with the composition of multiple systems. To do so, we need to address the tensor-product decomposition of the real image of an $n$-partite complex matrix.

\subsubsection{Real representation of the complex Kronecker product}

The $n$-fold standard mapping is a way to apply the standard mapping to multipartite systems such that the decomposition of the space into tensor factors is respected. In the following, it is accordingly assumed that every space has a preferred tensor factorization (given by the model), so that any matrix $A \in \bigotimes_{i=1}^n \kk^{d_{A_i} \times d_{A_i}}$ is known to be $n$-partite and the dimension $d_{A_i}$ of each of its $n$ factors is known. Also, each locally-defined space of matrices over the reals will be assumed to contain an image of a complex space through the standard mapping, therefore every space of real matrices will be assumed to be of even dimension. There is no loss of generality in doing so, as no pre-requisite but finiteness was taken on the dimensions of the real spaces. Limiting our analysis to even-dimensional spaces, and, a fortiori, to only the image of complex spaces in their corresponding real spaces, is enough to prove everything we need.

In the following, we will encounter expressions involving the Kronecker product of $k$-fold standard mappings with $j$-fold ones. As we already did in the previous section, in order to avoid clustered equations with tensor-factor swaps everywhere, we will assume that all such swaps are implicit. In particular, we will take as a convention to always gather the $\rr^{2\times 2}$ subsystems encoding the phase onto the left of expressions. For example, an equation like
\begin{equation}
    \begin{aligned}
        &\overline{\Gamma}^{(j)}\{A\} \ktensor \overline{\Gamma}^{(k)}\{B\}\\
        & =\left( (\I{k} \ktensor \ReP{A} + \J{k} \ktensor \ImP{A}) \ktensor (\I{j} \ktensor \ReP{B} + \J{j} \ktensor \ImP{B} \right)\\
        & = \big( \I{k} \ktensor \I{j} \ktensor  \ReP{A} \ktensor \ReP{B} + \I{k} \ktensor \J{j}  \ktensor \ReP{A} \ktensor \ImP{B} \\
        & \qquad + \J{k} \ktensor \I{j} \ktensor \ImP{A} \ktensor \ReP{B} + \J{k} \ktensor \J{j}  \ktensor \ImP{A} \ktensor \ImP{B}\big)
    \end{aligned}
\end{equation}
will be assumed to hold, although the second equality is only true up to the permutation of the tensor factors that allowed us to bring the $\I{j}$ and $\J{j}$ matrices to the left.
    
In our previous work~\cite{RNQT}, we remarked that since $(\I{n})^2 =\I{n}$, it is a projector. This is precisely the projector sending to the quotient space of matrices global with respect to $(\I{n},\J{n})$. In other words, since $\I{n}$ is global w.r.t. any $(\I{k},\J{k})$, $k\leq n$, then $\I{n} X \I{n}$ will also be global; by the adjoint action of $\I{n}$, the matrix $X$ has been projected into the space of all matrices global w.r.t. local representation of $(1,i)$. 
For $n$-fold special symmetric matrices, this globality of the representation is encoded by the following lemma:
\begin{lemma}\label{lemm:special_sym_projector_n}
    Let $\widetilde{A} \in \bigotimes_{i=1}^n \rr^{2d_{A_i} \times 2d_{A_i}} \cong \rr^{2^n d_A \times 2^n d_A}, d_A = \prod_i d_{A_i}$ be a $n$-fold special symmetric matrix, then the following holds (up to tensor-factor permutations):
    \begin{subequations}
    \begin{enumerate}
        \item 
        \begin{equation}
            \widetilde{A} = \I{n}_{d_A} \:\widetilde{A} = \widetilde{A}\: \I{n}_{d_A} = \I{n}_{d_A} \:\widetilde{A}\: \I{n}_{d_A}\:;
        \end{equation}
        \item $\forall k \leq n,$
        \begin{equation}
             (\J{k} \ktensor I^{\otimes (n-k)})_{d_A} \widetilde{A} = \J{n}_{d_A} \: \widetilde{A} \:.
        \end{equation}
    \end{enumerate}
    \end{subequations}
\end{lemma}
\begin{proof}
    Item 1 follows from \cref{prop:gamma_n_properties}: since the $n$-fold special symmetric matrix in $\rr^{2^n d_A \times 2^n d_A}$ are isomorphic to the Hermitian matrices in $\cc^{d_A \times d_A}$, it must have the form (once again, up to tensor-factor permutations)
    \begin{equation}
        \widetilde{A} = \I{n} \ktensor \ReP{A} + \J{n} \ktensor \ImP{A} \:,
    \end{equation}
    for some Hermitian $A$ in $\cc^{d_A \times d_A}$. Then,
    \begin{multline}
        \I{n}_{d_A} \:\widetilde{A}\: \I{n}_{d_A} = (\I{n} \ktensor \id_{d_A})(\I{n} \ktensor \ReP{A} + \J{n} \ktensor \ImP{A})(\I{n} \ktensor \id_{d_A})\\
        = (\I{n})^3 \ktensor \ReP{A} + (\I{n}\J{n}\I{n}) \ktensor \ImP{A})\\
        = \I{n} \ktensor \ReP{A} + \J{n} \ktensor \ImP{A} \\
        = \widetilde{A} \:,
    \end{multline}
    and we simply used the algebraic relations between $\I{n}$ and $\J{n}$. 
    The second item is proven similarly, but this time using \cref{lemm:I_J_global}
\end{proof}

Such projection can be used to make the real Kronecker product well-behaved with respect to all local irreps of the complex structure. That way, the local phase freedom present in complex quantum theory is preserved when going to the real representation. This was one of the main technical points of our paper~\cite{RNQT}, as well as of Ref.~\cite{Barrios2025}, although both these works focused on the `phase-localising' version of the projection, i.e. one reducing the number of subsystems encoding the phase from one per party to only one for all of the parties. This composition, respectively denoted $\otimes_r$ and $\otimes_F$ in the aforementioned works, was favoured for dimensionality reasons as well as for its property of being the image of the Kronecker product through the standard mapping, $\G{A \ktensor B} = \G{A} \otimes_r \G{B}$.

Differently here, the composition rule we considered will be called the `$R$-product' or `$\Rprod$-product' and denoted $\Rprod$.
\begin{definition}[$R$-product]\label{def:Rprod}
    Let $\widetilde{A} \in \bigotimes_{i=1}^j \rr^{2d_{A_i} \times 2d_{A_i}} \cong \rr^{2^jd_A \times 2^jd_A}$ and $\widetilde{B} \in \bigotimes_{l=1}^k \rr^{2d_{B_l} \times 2d_{B_l}} \cong \rr^{2^kd_B \times 2^k d_B}$ be two real matrices, where $d_A := \prod_i d_{A_i}$ and $ d_B := \prod_l d_{B_l}$, and let $\I{j+k}_{d_Ad_B} = \I{j+k} \ktensor \id_{d_A} \ktensor\id_{d_B}$ with $\I{j+k}$ be defined as in \Cref{def:I_J_n}. Then, the $\Rprod$-product of $\widetilde{A} $ and $\widetilde{B} $ is the matrix $\widetilde{A} \Rprod \widetilde{B} \in \left(\left(\bigotimes_{i=1}^j \rr^{2d_{A_i} \times 2d_{A_i}} \right) \otimes \left(\bigotimes_{l=1}^k \rr^{2d_{B_l} \times 2d_{B_l}}\right)\right) \cong \rr^{2^{j+k}d_Ad_B \times 2^{j+k}d_Ad_B}$ defined as
    \begin{equation}
    \begin{gathered}\label{eq:Rprod}
        \forall \widetilde{A} \in \bigotimes_{i=1}^j \rr^{2d_{A_i} \times 2d_{A_i}}, \: \forall \widetilde{B} \in \bigotimes_{l=1}^k \rr^{2d_{B_l} \times 2d_{B_l}},\\
         (\widetilde{A} \Rprod \widetilde{B}) := \I{j+k}_{d_Ad_B}(\widetilde{A} \ktensor \widetilde{B})\I{j+k}_{d_Ad_B} \:.
    \end{gathered}
    \end{equation}
\end{definition} 
It should be evident that, like the $\ktensor$ and the $\otimes_r$, this is a real-bilinear and associative composition rule. 
Similarly to how the $\otimes_r$-product is the image of the complex Kronecker product through the standard mapping $\Gamma$, we can show that the $\Rprod$-product is its image through the $n$-fold standard mapping. Indeed, the $\Rprod$ reduces to $\otimes_r$ when $n=1$ and notice that its definition implies 
\begin{subequations}
    \begin{align}
        & \I{j+k}_{d_Ad_B} = \I{j}_{d_A} \Rprod \I{k}_{d_B} = - \J{j}_{d_A} \Rprod\J{k}_{d_B}\:;\\
        & \J{j+k}_{d_Ad_B} = \I{j}_{d_A} \Rprod \J{k}_{d_B} = \J{j}_{d_A} \Rprod\I{k}_{d_B} \:.
    \end{align}
\end{subequations}
From there, it is direct to show that the $R$-product is the image through the $n$-fold standard mapping of the complex Kronecker product.
\begin{prop}\label{prop:RprodisKprod}
    The $R$-product is the image of the complex Kronecker product through the $\overline{\Gamma}^{(n)}$ map. In symbols:
    \begin{equation}
    \begin{gathered}
        \forall j,k, \quad \forall A \in \bigotimes_{i=1}^j \cc^{d_{A_i} \times d_{A_i}}, \: \forall B \in \bigotimes_{l=1}^k \cc^{d_{B_l} \times d_{B_l}}, \\  
        \overline{\Gamma}^{(j+k)}\{A \ktensor B\} \cong \overline{\Gamma}^{(j)}\{A\} \Rprod \overline{\Gamma}^{(k)}\{B\} \:.
    \end{gathered}
    \end{equation}
\end{prop}
\begin{proof}
Let $d_A = \prod_i d_{A_i}$ and $d_B = \prod_i {d_{B_i}}$, then:
    \begin{multline}
        \overline{\Gamma}^{(j)}\{A\} \Rprod \overline{\Gamma}^{(k)}\{B\} = \I{j+k}_{d_Ad_B}\left( \overline{\Gamma}^{(j)}\{A\} \ktensor \overline{\Gamma}^{(k)}\{B\} \right) \I{j+k}_{d_Ad_B}\\
        = \I{j+k}_{d_Ad_B}\left( (\I{k} \ktensor \ReP{A} + \J{k} \ktensor \ImP{A}) \ktensor (\I{j} \ktensor \ReP{B} + \J{j} \ktensor \ImP{B} \right) \I{j+k}_{d_Ad_B}\\
        = \I{j+k}_{d_Ad_B}\big( (\I{k} \ktensor \ReP{A} \ktensor \I{j} \ktensor \ReP{B} + \J{k} \ktensor \ImP{A} \ktensor \J{j} \ktensor \ImP{B})
        \\+ ( \I{k} \ktensor \ReP{A} \ktensor  \J{j} \ktensor \ImP{B}  +  \J{k} \ktensor \ImP{A} \ktensor \I{j} \ktensor \ReP{B})\big) \I{j+k}_{d_Ad_B}\\
        = \big(\I{j+k} \ktensor (\ReP{A} \ktensor \ReP{B} - \ImP{A} \ktensor \ImP{B})\big) \I{j+k}_{d_Ad_B}
        \\+ \big( \J{j+k} \ktensor (\ReP{A} \ktensor \ImP{B} +  \ImP{A} \ktensor \ReP{B})\big) \I{j+k}_{d_Ad_B}\\
        =\I{j+k} \ktensor \ReP{A \ktensor B} + \J{j+k} \ktensor \ImP{A \ktensor B}\\
        = \overline{\Gamma}^{(j+k)}\{A \ktensor B\} \:.
    \end{multline}
    In the above, some re-ordering of factors was necessary to pass from one step to the next.
\end{proof}
As first shown in Ref.~\cite{RNQT}, it entails that every property of the Kronecker product is true for the $R$-product when restricted to real matrices that have a preimage through the standard mapping, and most of its properties remain true in the general case. The following corollary summarises the ones we need.
\begin{corollary} \label{coro:RNQT_cross_norm}
    When acting on $n$-fold special symmetric matrices, the $R$-product obeys the mixed-product rule (or interchange law) and commutes with the trace up to a factor of 2:
    \begin{subequations}
        \begin{gather}
            \forall j,k \leq n:\: j+k =n, \quad \forall A, A' \in \bigotimes_{i=1}^j \cc^{d_{A_i} \times d_{A_i}}, \: \forall B,B' \in \bigotimes_{l=1}^k \cc^{d_{B_l} \times d_{B_l}}, \notag \\
            \begin{multlined}
                (\overline{\Gamma}^{(j)}\{A\} \Rprod \overline{\Gamma}^{(k)}\{B\})(\overline{\Gamma}^{(j)}\{A'\} \Rprod \overline{\Gamma}^{(k)}\{B'\} ) \\= (\overline{\Gamma}^{(j)}\{A\}\overline{\Gamma}^{(j)}\{A'\})\Rprod (\overline{\Gamma}^{(k)}\{B\}\overline{\Gamma}^{(k)}\{B'\} ) \:;
            \end{multlined}
            \label{eq:mixed_product}\\
            \TrX{}{\overline{\Gamma}^{(j)}\{A\} \Rprod \: \overline{\Gamma}^{(k)}\{B\}} = \frac{1}{2} \TrX{}{\overline{\Gamma}^{(j)}\{A\}}\TrX{}{\overline{\Gamma}^{(k)}\{B\}}\:.
            \label{eq:cross_norm}
        \end{gather}
    \end{subequations}
\end{corollary}

\subsubsection{From QT to RNQT models}\label{sec:RNQT_model}
With the introduction of the $R$-product, we can characterise the image of an $n$-partite QT model through the $n$-fold standard mapping.
\QTisRNQT*
\begin{proof}
    Using the pair of maps $(M,E) = (\frac{1}{2}\overline{\Gamma}^{(n)},\overline{\Gamma}^{(n)})$ and using the system-composition rules defined by
    \begin{subequations}
        \begin{align}
            &\cdot \rtensor \cdot := 2(\cdot \Rprod \cdot) \:; \\
            &\cdot \rtensordual \cdot := \cdot \Rprod \cdot \:;
        \end{align}
    \end{subequations}
    the result follows from applying the properties proven in \cref{prop:gamma_properties,prop:gamma_n_properties,prop:RprodisKprod}.
\end{proof}
With this proposition, what remains to show is that the RNQT composition rules can be safely substituted by the RQT composition rules. That way, the RNQT becomes the QT-image of the RQT, and the inclusion of a sub-model isomorphic to QT within RQT is proven.

\subsubsection{Local RNQT models are embedded in local RQT models}\label{sec:RNQTlocal}
Before moving on to the general case, let us focus on the image of an $n$-partite, single-source, local QT model through the $n$-fold standard mapping.
\begin{corollary}\label{lemm:RNQT_model_local}
    The image of the states and of the effects of the $n$-partite, single-source, local QT model of matrix dimension $d_Ad_B \dots$ through the pair of maps $(\frac{1}{2}\overline{\Gamma}^{(n)}, \overline{\Gamma}^{(n)} )$ is a sub-model of the single-partite RQT model of matrix dimension $2d_A2d_B \dots$ that has underwent the following modification:
    \begin{itemize}
        \item The effect composition rule $\ktensor$ has been replaced by $\cdot \rtensordual \cdot := \cdot \Rprod \cdot $ .
    \end{itemize}
\end{corollary}

The QT image with this modified effect composition rule will be referred to as the local RNQT model. With the results of the previous sections, we already have enough material to show that the local RNQT model can be seen as the ``$n$-fold QT-image'' of the RQT model. We now show that every multipartite, single-source, local RQT model has a sub-model equivalent to a local QT model of $(1/2)^n$-th the matrix dimension. As there is only one source, the only modification that should be addressed to pass from the RNQT defined by \cref{lemm:RNQT_model_local} to the QT-image of RQT is the different composition rule for effects. 

The next proposition does that by proving that the $\Rprod$- and $\ktensor$-products of effects are operationally equivalent for single-source RQ(N)T models. As it turns out, the $n$-fold QT-image is a subspace which is invariant under the projection $\I{n}_{d_Ad_B\dots}$, and the projection $\I{n}_{d_Ad_B\dots}$ `swallows smaller projections' by definition, like e.g. $\I{n}_{d_Ad_B\dots} (\I{1}_{d_A} \ktensor \I{n-1}_{d_B\dots}) \cong \I{n}_{d_Ad_B\dots}$. This implies that in an inner product of a multipartite state in the $n$-fold QT-image with the Kronecker product of any real POVMs, the Kronecker products will be turned into $R$-products.
\localRQTislocalQT*
\begin{proof}
    Similar to the RNQT case, set $M_{\Aa\Bb\dots} = 2\overline{\Gamma}^{(n)}$. By \cref{prop:gamma_n_properties}, $M_{\Aa\Bb\dots}(\rho_{\Aa\Bb\dots})$ is a valid RQT state. On the other hand, set $E_{\Xx}(X_{\alpha|\mu}) = \Gamma\{X_{\alpha|\mu}\}$ for every local party $\Xx$; by \cref{prop:gamma_properties}, this transforms the QT POVMs into RQT POVMs. 
    
    Then, it remains to show that the inner products are preserved. For every $a,b, \dots$: 
    \begin{multline}
        \TrX{}{M_{\Aa\Bb...}(\rho_{\Aa\Bb\ldots}) (E_{\Aa}(A_{a|x}) \ktensor E_{\Bb}(B_{b|y}) \ktensor ...)}\\
        = 2\TrX{}{\overline{\Gamma}^{(n)}\{\rho_{\Aa\Bb\ldots}\} \Big(\Gamma\{A_{a|x}\} \ktensor \Gamma\{B_{b|y}\} \ktensor ...\Big)}\\
        = 2\TrX{}{\I{n}_{d_{\Aa}d_{\Bb}\dots}\overline{\Gamma}^{(n)}\{\rho_{\Aa\Bb\ldots}\} \I{n}_{d_{\Aa}d_{\Bb}\dots}\Big(\Gamma\{A_{a|x}\} \ktensor \Gamma\{B_{b|y}\} \ktensor ...\Big)}\\
        = 2\TrX{}{\overline{\Gamma}^{(n)}\{\rho_{\Aa\Bb\ldots}\} \I{n}_{d_{\Aa}d_{\Bb}\dots}\Big(\Gamma\{A_{a|x}\} \ktensor \Gamma\{B_{b|y}\} \ktensor ...\Big)\I{n}_{d_{\Aa}d_{\Bb}\dots}}\\
        = 2\TrX{}{\overline{\Gamma}^{(n)}\{\rho_{\Aa\Bb\ldots}\} \Big(\Gamma\{A_{a|x}\} \Rprod \Gamma\{B_{b|y}\} \Rprod ...\Big)}\\
        = 2\TrX{}{\overline{\Gamma}^{(n)}\{\rho_{\Aa\Bb\ldots}\} \overline{\Gamma}^{(n)}\{A_{a|x} \ktensor B_{b|y} \ktensor ...\}}\\
        = \TrX{}{\rho_{\Aa\Bb\ldots} \: (A_{a|x} \ktensor B_{b|y} \ktensor ...)}
    \end{multline}
    where we first substituted the mappings by their expression using the standard mappings, then we successively used \cref{lemm:special_sym_projector_n}, cyclicity of the trace, \cref{def:Rprod}, \cref{prop:RprodisKprod}, and \cref{prop:gamma_n_properties} to pass from each line to the next.
\end{proof}

It remains to show that the proof generalises for any network structure. If the situation involves multiple sources each measured locally, it reduces to applying \cref{prop:localRQTislocalQT} over each source. So the only remaining cases are those with joint measurements across several sources. 

\subsection{Multilocal models and relaxing the theory-dependency of RQT models}\label{sec:RNQTmultilocal}

The multilocal model generalises the local model by allowing for several sources. These models represent the network scenarios that assume multiple local measurements from multiple independent sources. The sources can produce multipartite systems that are shared among the measuring parties.

\subsubsection{RNQT models and operationally independent sources}
As we did in the previous section, the first step is to apply the $n$-fold standard mapping to the multilocal  
QT model, and then to split the mapping into $n$ single-fold standard mappings $\overline{\Gamma}^{(n)} = \Gamma \Rprod \Gamma \Rprod \dots \Rprod \Gamma$ via \cref{prop:RprodisKprod} to obtain a representation-local image of QT as real matrices, i.e., a RNQT model as in \cref{prop:QTisRNQT}.
\begin{corollary}\label{lemm:RNQT_model_multilocal}
    The image of the states and of the effects of the $n$-partite, multiple-source, multilocal QT model of matrix dimension $d_Ad_B \dots$ through the pair of maps $(\frac{1}{2}\overline{\Gamma}^{(n)}, \overline{\Gamma}^{(n)} )$ form a sub-model in an $n$-partite multilocal RQT model of matrix dimension $2d_A2d_B \dots$ that has underwent the following modifications:
    \begin{itemize}
        \item The state composition rule $\ktensor$ has been replaced by $\cdot \rtensor \cdot := 2(\cdot \Rprod \cdot) $;
        \item The effect composition rule $\ktensor$ has been replaced by $\cdot \rtensordual \cdot := \cdot \Rprod \cdot $ .
    \end{itemize}
\end{corollary}
The next step would be to show that the modifications can be removed while preserving all predicted distributions. 
Nevertheless, there is no hope that this can be done directly: as it was shown by Renou et al., if the joint sources are represented by a state in Kronecker product form, then there is no possibility to represent all QT models within RQT. Their proof~\cite{Renou2021} presented a counterexample in the bilocal scenario (two sources, three measurements, among which one measurement is jointly applied on one subsystem from each source).

But this only means there is no way to represent all QT models within the naive RQT model, which assumes independent sources to be represented as (Kronecker-)product states. As argued in the main text, representing the joint state of independent sources via a state in a pure Kronecker product state is operationally hard to justify: operationally independent states result in independent distributions (i.e., totally uncorrelated outcomes) for all local measurements, so they seem to be the right candidate to represent the joint state of several systems prepared by independent procedures. Yet, for RQT \cref{prop:independent_sources} shows that the pure Kronecker-product states are only a subset of these states. Assuming that independent sources are represented by product states, therefore, amounts to forbidding a large number of would-be independent preparation procedures.

By consequence, the multilocal RQT model with product-state independence will be weakened into the multilocal RQT model with operational independence, which is a better representation of the experiment in the sense that it no longer arbitrarily forbids many preparation procedures. The goal of this section is to show that the multilocal RNQT models obtained through \cref{lemm:RNQT_model_multilocal} are equivalent to multilocal RQT models with operationally independent sources.

\subsubsection{The `contamination' property of the R-product}\label{sec:RNQT_contamination}
We first need to generalise one of the key elements in the proof that RNQT embeds into RQT in the local case: the property that the $R$-product `contaminates' Kronecker products when applied to matrices that have a preimage through the ($n$-fold) standard mapping. Equations involving mixtures of the $R$-product and the Kronecker product are equivalent to the same equations with all the $\ktensor$-products replaced by $\Rprod$-products:
\begin{equation}
    \begin{aligned}
        \forall j,k \leq n:\: & j+k =n, \quad \forall A, A' \in \bigotimes_{i=1}^j \cc^{d_{A_i} \times d_{A_i}}, \: \forall B,B' \in \bigotimes_{l=1}^k \cc^{d_{B_l} \times d_{B_l}},  \\
        &(\overline{\Gamma}^{(j)}\{A\} \Rprod \overline{\Gamma}^{(k)}\{B\})(\overline{\Gamma}^{(j)}\{A'\} \ktensor \overline{\Gamma}^{(k)}\{B'\} ) \\
        &= (\overline{\Gamma}^{(j)}\{A\} \ktensor \overline{\Gamma}^{(k)}\{B\})(\overline{\Gamma}^{(j)}\{A'\} \Rprod \overline{\Gamma}^{(k)}\{B'\} )\\
        &= (\overline{\Gamma}^{(j)}\{A\} \Rprod \overline{\Gamma}^{(k)}\{B\})(\overline{\Gamma}^{(j)}\{A'\} \Rprod \overline{\Gamma}^{(k)}\{B'\} )\:.
    \end{aligned}
\end{equation}
This happens because the Kronecker product between any $j-$ and $k-$fold special symmetric matrices $\overline{\Gamma}^{(j)}\{A\}$ and $ \overline{\Gamma}^{(k)}\{B\}$ commutes with the projector $\I{j+k}_{d_Ad_B}$ (where $d_A = \prod_i^j d_{A_i}$, $d_B = \prod_l^k d_{B_l}$), 
\begin{equation}\label{eq:Rprod_Special_Sym}
    \I{j+k}_{d_Ad_B}(\overline{\Gamma}^{(j)}\{A\} \ktensor \overline{\Gamma}^{(k)}\{B\})= (\overline{\Gamma}^{(j)}\{A\} \ktensor \overline{\Gamma}^{(k)}\{B\})\I{j+k}_{d_Ad_B} \:. 
\end{equation}
This, in turn, implies that any of the two sides of the above defines an $R$-product, for instance
\begin{equation}\label{eq:Rprod_Special_Sym_bis}
    \I{j+k}_{d_Ad_B}(\overline{\Gamma}^{(j)}\{A\} \ktensor \overline{\Gamma}^{(k)}\{B\}) = \overline{\Gamma}^{(j)}\{A\} \Rprod \overline{\Gamma}^{(k)}\{B\} \:.
\end{equation}
Indeed since $\overline{\Gamma}^{(j)}\{A\}\J{j}_{d_A} = \J{j}_{d_A}\overline{\Gamma}^{(j)}\{A\}$ and $\overline{\Gamma}^{(k)}\{B\}\J{k}_{d_B} = \J{k}_{d_B}\overline{\Gamma}^{(k)}\{B\})$ and using \cref{lemm:I_J_global} and \cref{coro:I_J_groupping} gives
\begin{multline}
    \overline{\Gamma}^{(j)}\{A\} \Rprod \overline{\Gamma}^{(k)}\{B\} = \I{j+k}_{d_Ad_B}(\overline{\Gamma}^{(j)}\{A\} \ktensor \overline{\Gamma}^{(k)}\{B\})\I{j+k}_{d_Ad_B} \\
    =\frac{1}{2}(\I{k}_{d_A} \ktensor \I{j}_{d_B}-\J{k}_{d_A} \ktensor \J{j}_{d_B})(\overline{\Gamma}^{(j)}\{A\} \ktensor \overline{\Gamma}^{(k)}\{B\})\frac{1}{2}(\I{k}_{d_A} \ktensor \I{j}_{d_B}-\J{k}_{d_A} \ktensor \J{j}_{d_B})\\
    =\frac{1}{4}(\I{k}_{d_A} \ktensor \I{j}_{d_B}-\J{k}_{d_A} \ktensor \J{j}_{d_B})(\I{k}_{d_A} \ktensor \I{j}_{d_B}-\J{k}_{d_A} \ktensor \J{j}_{d_B})(\overline{\Gamma}^{(j)}\{A\} \ktensor \overline{\Gamma}^{(k)}\{B\})\\
    =(\I{j+k}_{d_Ad_B})^2(\overline{\Gamma}^{(j)}\{A\} \ktensor \overline{\Gamma}^{(k)}\{B\})\\
    =\I{j+k}_{d_Ad_B}(\overline{\Gamma}^{(j)}\{A\} \ktensor \overline{\Gamma}^{(k)}\{B\}) \:.
\end{multline}
When the matrices do not have a preimage as complex matrices, this relation no longer holds in general.
Similarly to how the properties \eqref{eq:mixed_product} and \eqref{eq:cross_norm} are also not true for all real matrices. 

Yet, these hold in sufficiently many cases for our proof techniques to apply. First, \cref{eq:cross_norm} holds as soon as one of the matrices involved is symmetric, i.e.  $\forall A \in \bigotimes_i^j\rr^{2d_{A_i} \times 2d_{A_i}}$, $\forall B \in \bigotimes_l^k\rr^{2d_{B_l} \times 2d_{B_l}}$,
\begin{equation}\label{eq:Rprod_distributes_in_trace}
    {A} ={A}^T \vee {B} = {B}^T \Longrightarrow \TrX{}{{A} \Rprod {B}} = \frac{1}{2}\TrX{}{{A} \ktensor {B}} = \frac{1}{2}\TrX{}{{A}}\TrX{}{{B}} \:.
\end{equation}
This can be obtained from the definition by a direct computation
\begin{equation}
    \TrX{}{{A} \Rprod {B}} = \frac{1}{2}(\TrX{}{{A}}\TrX{}{{B}} - \TrX{}{\J{j}_{d_A}{A}}\TrX{}{\J{k}_{d_B}{B}})\:,
\end{equation}
and using that $\TrX{}{\J{j}_{d_A}A} = 0$ if $A^T = A$ since $\J{j}_{d_A}$ is an antisymmetric matrix.

As for \cref{eq:mixed_product}, it still holds if an overall trace is taken and one of the sides is globally special symmetric. This result, phrased as a lemma below, is the main ingredient to prove every statement of the main text. As is the case with every proof concerning the $R$-product to some extent, the proof of this lemma is but a consequence of the $\I{n}$ matrix being a projector.
\begin{lemma}[Equivalence of traces of R-products with traces of Kronecker products]\label{lemm:Rprod_contaminates}
    Let $ A,L \in \bigotimes_i\rr^{2d_{A_i} \times 2d_{A_i}}$ and $ B,R \in \bigotimes_l\rr^{2d_{B_l} \times 2d_{B_l}}$ be four matrices. Then, consider the trace of the $\Rprod$-product of $A$ and $B$ multiplied by the Kronecker product of $L$ and $R$, $\TrX{}{(A \Rprod B)(L \ktensor R)}$. It has the following properties: 
    \begin{enumerate}
        \item It is equivalent to the trace of two $\Rprod$-products, 
    \begin{equation}\label{eq:R_contamination}
        \TrX{}{(A \Rprod B)(L \ktensor R)} = \TrX{}{(A \Rprod B)(L \Rprod R)} \:.
    \end{equation}
    \item The interchange law can be used if $A$ and $B$ are special symmetric, 
    \begin{equation}\label{eq:R_contamination_intoR}
        \begin{gathered}
            \forall A\in \SYR{d_A}, \: \forall B\in \SYR{d_B}\:,\\
         \TrX{}{(A \Rprod B)(L \ktensor R)} = \TrX{}{(AL) \Rprod (BR)} \:,
        \end{gathered}
    \end{equation}
    \item If $R$ and $L$ are furthermore symmetric, then the trace commutes with the $\Rprod$-product,
    \begin{equation}\label{eq:R_contamination_into_ind}
        \begin{gathered}
            \forall A\in \SYR{d_A}, \:  \forall B\in \SYR{d_B},\: \forall L = L^T, \: \forall R = R^T, \\
            \TrX{}{(A \Rprod B)(L \ktensor R)} 
            = \frac{1}{2}\TrX{}{AL}\TrX{}{BR} \:.
        \end{gathered}
    \end{equation}
    In this last case, the trace of an $\Rprod$-product times a Kronecker product is equivalent to the trace of two Kronecker products (up to a factor of 2),
    \begin{equation}\label{eq:R_contamination_intoK}
        \begin{gathered}
            \forall A\in \SYR{d_A}, \:  \forall B\in \SYR{d_B},\: \forall L = L^T, \: \forall R = R^T, \\
            \TrX{}{(A \Rprod B)(L \ktensor R)} = \frac{1}{2}\TrX{}{(A \ktensor B)(L \ktensor R)}\:.
        \end{gathered}
    \end{equation}
    \end{enumerate}
\end{lemma}
\begin{proof}
    Suppose $A$ is $j$-partite and $B$ is $k$-partite, the first equation is a consequence of $\I{j+k}_{d_Ad_B}$ being a projector,
    \begin{multline}
        \TrX{}{(A \Rprod B)(L \ktensor R)} 
        = \TrX{}{\I{j+k}_{d_Ad_B}(A \ktensor B)\I{j+k}_{d_Ad_B}(L \ktensor R)} \\
        = \TrX{}{(\I{j+k}_{d_Ad_B})^2(A \ktensor B)(\I{j+k}_{d_Ad_B})^2(L \ktensor R)} \\
        = \TrX{}{\I{j+k}_{d_Ad_B}(A \ktensor B)\I{j+k}_{d_Ad_B}\I{j+k}_{d_Ad_B}(L \ktensor R)\I{j+k}_{d_Ad_B}} \\ 
        = \TrX{}{(A \Rprod B)(L \Rprod R)}\:.
    \end{multline}
    The second equation \eqref{eq:R_contamination_intoR} is proven by using that $\I{j+k}_{d_Ad_B} = \frac{1}{2}(\I{k}_{d_A} \ktensor \I{j}_{d_B}-\J{k}_{d_A} \ktensor \J{j}_{d_B})$. The right-hand side extends into 
    \begin{multline}
        \TrX{}{(AL) \Rprod (BR)}=\TrX{}{\I{j+k}_{d_Ad_B} ((AL) \ktensor (BR))\I{j+k}_{d_Ad_B}} \\
        =\frac{1}{2} \TrX{}{((AL) \ktensor (BR)) - ((\J{j}_{d_A}AL) \ktensor (\J{k}_{d_B}BR))} \:.
    \end{multline}
    Using the first equation \eqref{eq:R_contamination}, the left-hand side extends into
    \begin{multline}
        \TrX{}{(A \Rprod B)(L \ktensor R)}=\TrX{}{(A \Rprod B)(L \Rprod R)} \\
        = \TrX{}{\I{j+k}_{d_Ad_B}(A \ktensor B)\I{j+k}_{d_Ad_B}(L \ktensor R)\I{j+k}_{d_Ad_B}}\\
        =\frac{1}{4}\mathrm{tr}\big[(AL) \ktensor (BR) - (A\J{j}_{d_A}L) \ktensor (B\J{k}_{d_B}R) \\
        - (\J{j}_{d_A}AL) \ktensor (\J{k}_{d_B}BR) +  (\J{j}_{d_A}A\J{j}_{d_A}L) \ktensor (\J{k}_{d_B}B\J{k}_{d_B}R) \big]\:.
    \end{multline}
    From there, one can see that if $A$ and $B$ commute with, respectively, $\J{j}_{d_A}$ and $\J{k}_{d_B}$ then
    \begin{equation}
        \TrX{}{(A \Rprod B)(L \ktensor R)}=\frac{1}{4} \TrX{}{2((AL) \ktensor (BR)) - 2((\J{j}_{d_A}AL) \ktensor (\J{k}_{d_B}BR))} \:,
    \end{equation}
    and so the two sides become equivalent.

    Finally, to prove the third and fourth equations, \cref{eq:R_contamination_into_ind,eq:R_contamination_intoK}, we apply \cref{eq:Rprod_distributes_in_trace}.
\end{proof}

\subsubsection{Allowing operationally independent sources}
The interpretation of each of the two equations of the last item of the previous lemma has two important consequences.
\begin{corollary}
    A $n$-partite RQT state in $R$-product form is operationally independent.
\end{corollary}
\begin{corollary}
    A $n$-partite RQT state in Kronecker product form is locally indistinguishable from the corresponding RQT state in $R$-product form.
\end{corollary}
With these, it is straightforward to weaken the multilocal RNQT model into a `less modified' submodel of RQT with operationally independent sources.
\begin{lemma}\label{lemm:RNQT_model_multilocal_op_indep}
    The image of the states and of the effects of the $n$-partite multilocal QT model of matrix dimension $d_Ad_B \dots$ through the pair of maps $(\frac{1}{2}\overline{\Gamma}^{(n)}, \overline{\Gamma}^{(n)} )$ form a sub-model in an $n$-partite multilocal RQT model of matrix dimension $2d_A2d_B \dots$ that has underwent the following modifications:
    \begin{itemize}
        \item The state of independent sources is represented by operationally independent states.
        \item The $k$-partite joint effects on subsystems $\Xx\Yy\dots$ sum up to $\I{k}_{d_Xd_Y\dots}$ instead of $\id_{2^kd_Xd_Y\dots}$.
    \end{itemize}
\end{lemma}
\begin{proof}
    We start with the states. From \cref{lemm:RNQT_model_multilocal}, any QT product state on the $n$-parties, say $\rho^{\cc}_{\Xx} \ktensor \sigma^{\cc}_{\Yy}$ with $\Xx$ encompassing the first $k\leq n$ subsystems, is mapped to an $R$-product RQT state via
    \begin{equation}
        \frac{1}{2}\overline{\Gamma}^{(n)}\{\rho^{\cc}_{\Xx} \ktensor \sigma^{\cc}_{\Yy}\} = 2 \left( \frac{1}{2}\overline{\Gamma}^{(k)}\{\rho^{\cc}_{\Xx}\} \Rprod \frac{1}{2}\overline{\Gamma}^{(n-k)}\{ \sigma^{\cc}_{\Yy}\} \right)\:,
    \end{equation}
    where $\frac{1}{2}\overline{\Gamma}^{(k)}\{\rho^{\cc}_{\Xx}\} := \rho_\Xx$ is a valid RQT state of the subsystem shared by the first $k$ parties; $\frac{1}{2}\overline{\Gamma}^{(n-k)}\{ \sigma^{\cc}_{\Yy}\} := \sigma_\Yy $, one of the following $n-k$ parties; and $2 \left( \cdot \Rprod \cdot \right) := \cdot \rtensor \cdot$ is the image of the QT state-composition rule. 
    
    Because of \cref{lemm:Rprod_contaminates}, these $R$-product states are always operationally independent. Hence, the first modification of \cref{lemm:RNQT_model_multilocal} can be weakened to simply requiring the state of independent sources to be operationally independent states.
    
    We now turn our attention on the effects. By lemma \cref{lemm:RNQT_model_multilocal} again, any QT product effect, say $X^\cc_x \ktensor Y^\cc_y$, is mapped to a $R$-product of real matrices via
    \begin{equation}
        \overline{\Gamma}^{(n)}\{X^\cc_x \ktensor Y^\cc_y\} = \overline{\Gamma}^{(k)}\{X^{\cc}_{x}\} \Rprod \overline{\Gamma}^{(n-k)}\{ Y^{\cc}_{y}\} \:,
    \end{equation}
    where $\overline{\Gamma}^{(k)}\{X^{\cc}_{x}\}:= X_x$ is a positive semi-definite matrix part of a set $\{X_x\}_x$ of positive semi-definite real matrices summing up to $\sum_x X_x = \I{k} \ktensor \id_{d_X} = \I{k}_{d_X}$; and $\overline{\Gamma}^{(n-k)}\{ Y^{\cc}_{y}\} := Y_y$ is a positive semi-definitite matrix part of a set summing up to $\sum_y Y_y = \I{n-k} \ktensor \id_{d_Y}$.

    Finally we can use the same proof as \cref{prop:localRQTislocalQT} --namely making the projector $\I{n}_{d_Xd_Y}$ appear on the state side of any inner product, i.e, $\rho_\Xx \rtensor \sigma_\Yy = \I{n}_{d_Xd_Y}(\rho_\Xx \rtensor \sigma_\Yy)\I{n}_{d_Xd_Y}$, and then passing this projector on the effect side via the cyclicity of the trace in order to turn the $\ktensor$ of effects into an $\Rprod$-- to show that the $\Rprod$-product of effects can be replaced by a Kronecker product without altering any predicted probability distribution. 
\end{proof}

By the corollaries of \cref{lemm:Rprod_contaminates}, an RQT state of the form $\rho_\Xx \rtensor \sigma_\Yy$ is operationally independent (as well as locally indistinguishable from the state $\rho_\Xx \ktensor \sigma_\Yy$), which explains the first modification of the model. This modification can also be thought of as `operational independence allows using $\rho_\Xx \rtensor \sigma_\Yy$ as the image of the QT state $\rho_\Xx^\cc \rtensor \sigma_\Yy^\cc$ in lieu of $\rho_\Xx \ktensor \sigma_\Yy$' to a certain extent (namely, within Born rules involving local measurements, but keep in mind that the $R$-product and Kronecker-product states still technically correspond to different independent preparation procedures that a global measurement could distinguish).

The other difference with the local case is that the effects must also be modified. To understand why it is the case, see how one passes from RNQT to RQT in a local scenario, for example the Bell one:
\begin{equation}\label{eq:example_Bell}
    \begin{aligned}
        &\TrX{}{\rho_{AB}^\cc (A_a^\cc \ktensor B_b^\cc)} \\
        &\overset{(\text{Lem. }\ref{lemm:RNQT_model_local})}{=} \frac{1}{2}\TrX{}{\overline{\Gamma}^{(2)}\{\rho_{AB}^\cc\} (\Gamma\{A_a^\cc\} \rtensordual \Gamma\{B_b^\cc\})} \\
        &\overset{(\text{Prop. \ref{prop:RprodisKprod}})}{=} \frac{1}{2}\TrX{}{\overline{\Gamma}^{(2)}\{\rho_{AB}^\cc\} (\Gamma\{A_a^\cc\} \ktensor \Gamma\{B_b^\cc\})} \\
        &\overset{(\text{Prop. \ref{prop:gamma_properties}, \ref{prop:gamma_n_properties}})}{=} \TrX{}{\rho_{AB}^\rr (A_a^\rr \ktensor B_b^\rr)} \:,
    \end{aligned}
\end{equation}
such that $ \rho_{AB}^\rr = \frac{1}{2}\overline{\Gamma}^{(2)}\{\rho_{AB}^\cc\}$, $A_a^\rr = \Gamma\{A_a^\cc\}$, and $ B_b^\rr = \Gamma\{B_b^\cc\}$ are valid elements of RQT. This succession of steps is essentially the proof of \cref{prop:singleRQTissingleQT}.
Now compare this to a multilocal scenario, for example the bilocal one:
\begin{equation}\label{eq:example_bilocal}
    \begin{aligned}
        &\TrX{}{\rho_{AB_1}^\cc \ktensor \sigma_{B_2C}^\cc)(A_a^\cc \ktensor B_b^\cc \ktensor C_z^\cc)} \\ 
        &\overset{(\text{Lem. }\ref{lemm:RNQT_model_multilocal})}{=} \frac{1}{4}\TrX{}{(\overline{\Gamma}^{(2)}\{\rho_{AB_1}^\cc\} \Rprod \overline{\Gamma}^{(2)}\{\sigma_{B_2C}^\cc\})(\Gamma\{A_a^\cc\} \Rprod \overline{\Gamma}^{(2)}\{B_b^\cc\} \Rprod \Gamma\{C_z^\cc\})}\\
        &\overset{(\text{Lem. }\ref{lemm:RNQT_model_multilocal_op_indep})}{=} \frac{1}{4}\TrX{}{(\overline{\Gamma}^{(2)}\{\rho_{AB_1}^\cc\} \Rprod \overline{\Gamma}^{(2)}\{\sigma_{B_2C}^\cc\})(\Gamma\{A_a^\cc\} \ktensor \overline{\Gamma}^{(2)}\{B_b^\cc\} \ktensor \Gamma\{C_z^\cc\})} \:,
    \end{aligned}
\end{equation}
We would like to apply \cref{prop:gamma_properties,prop:gamma_n_properties} to finish the argument, but this cannot be done since the joint effects $\{B_b = \overline{\Gamma}^{(2)}\{B_b^\cc\}\}_b$ are not part of a POVM: they do not sum up to the identity matrix.

As we have argued that operational independence is a justified relaxation of the state space of multilocal RQT models, it only remains to show that this second modification is not essential to the model. That is, it remains to show that the RNQT effects of the form $\{B_b = \overline{\Gamma}^{(2)}\{B_b^\cc\}\}_b$ can be freely replaced by genuine RQT effects. 


Now, the reason why the joint effects do not sum to the identity matrix in the multilocal case is that they act on parts of independent states. As a result, their subsystems of definition cannot be bundled into a single global system. Thus, we must use the $k$-fold standard mapping $\overline{\Gamma}^{(k)}$ instead of the one-fold one $\Gamma$ in order for the complex phase factor to be properly represented across the local R(N)QT states and effects. 

\subsubsection{Relocalising the \textit{n}-fold standard mapping}
As a last step for the proof, we thus need to find a way to reduce the $k$-fold standard mappings to a mapping akin to the (single-fold) standard mapping, but in a way that does not interfere with the real representation of the complex phase.

This is not as big as an issue as it may appear: $\I{n} \ktensor \id_{d_A} \in \rr^{2^n d_A \times 2^n d_A}$ is isomorphic to $\id_{2d_A}$; we only went into this larger matrix space because the operators behave nicely w.r.t. the tensor factorisation. 
Recall from previous works~\cite{McKague,RNQT,Barrios2025} that the $n-1$ extra $\rr^{2\times2}$ spaces that the $\overline{\Gamma}^{(n)}$ map requires when passing from the complex to the real 
are there such that each subsystem has its own $\rr^{2\times2}$ `phase-encoding' extra part. Contrastingly, the $\Gamma$ map only requires one such `phase-encoding' extra part that is shared among all subsystems. 
The matrices $\I{n}$ and $\J{n}$ are then but non-local representations of $1$ and $i$ that act on every `phase-encoding' part of each subsystem so to emulate the global complex factors. 
As for the $R$-product, it is the correct composition rule for this emulation, as it ensures that the real representation of two local phases is combined into a global one, unlike the Kronecker product~\cite{RNQT,Barrios2025}.

This is fundamentally why the $R$-product contaminates the Kronecker product: in a multiplication of several operators, it is sufficient that only one operator treats the phase properly (i.e. can be interpreted as an $R$-product) for the whole multiplication to treat the phase properly, as the phase will just commute with every other operator until reaching the `good' one. 
Take for example the Bell scenario \cref{eq:example_Bell}, the effects are in Kronecker product, so we have
\begin{equation}
    \Gamma\{A_a^\cc\} \ktensor \Gamma\{e^{i \theta} B_b^\cc\} \neq \Gamma\{e^{i \theta}A_a^\cc\} \ktensor \Gamma\{ B_b^\cc\} \:.
\end{equation}
This has somehow `recluded' the phase on one side of the tensor factorisation. On the other hand, since the state is global, it can be seen as having been applied the $2$-fold mapping, and thus all its phase factors can move freely from one side to the other.
The contamination property, then, is just the fact that the `more global representations' (those with an interpretation as the highest-fold standard mapping) can be used as intermediaries to `pass the phase around'. In the Bell case, it is the joint state that is used as the intermediary:
\begin{multline}
    \TrX{}{\overline{\Gamma}^{(2)}\{\rho_{AB}^\cc\}(\Gamma\{A_a^\cc\} \ktensor \Gamma\{e^{i \theta} B_b^\cc\})} \\
    = \TrX{}{\overline{\Gamma}^{(2)}\{e^{i \theta}\rho_{AB}^\cc\}(\Gamma\{A_a^\cc\} \ktensor \Gamma\{ B_b^\cc\})} \\
    = \TrX{}{\overline{\Gamma}^{(2)}\{\rho_{AB}^\cc\}(\Gamma\{e^{i \theta}A_a^\cc\} \ktensor \Gamma\{ B_b^\cc\})} \:,
\end{multline}
explaining why the effects do not need to be in $R$-product form for the phase to be passed around properly. 

Now the fact is that operationally independent states are always in $R$-product form, so they can always be used to pass the phase around for the effects. The use of $k$-fold mappings on the effect side, which is why \cref{lemm:RNQT_model_multilocal_op_indep} has modified effects, is therefore unnecessary. For example, in the bilocal model, we reached
\begin{multline}
    \TrX{}{(\rho_{AB_1}^\cc \ktensor \sigma_{B_2C}^\cc)(A_a^\cc \ktensor B_b^\cc \ktensor C_z^\cc)}\\
    =\frac{1}{2}\TrX{}{(\overline{\Gamma}^{(2)}\{\rho_{AB_1}^\cc\} \Rprod \overline{\Gamma}^{(2)}\{\sigma_{B_2C}^\cc\})(\Gamma\{A_a^\cc\} \ktensor \overline{\Gamma}^{(2)}\{B_b^\cc\} \ktensor \Gamma\{C_z^\cc\})} \:,
\end{multline}
but the use of $\overline{\Gamma}^{(2)}$ at Bob's effect, $\overline{\Gamma}^{(2)}\{B_b^\cc\}$, is purely superfluous if it were not for ensuring that the dimensions of states and effects match. Thus, there is no need to delocalize the phase across Bob's subsystems in this situation; we can simply use a different representation that relocalizes Bob's phase within one of his subsystems, and let the state side handle moving it around.

On that account, we need a mapping that `relocalises' the $n$-fold mapping without changing its dimension. This mapping is inspired by but slightly different from the mapping $\tau_{\Rprod}$ introduced in our previous work~\cite{RNQT}, and that is used to collapse an $n$-fold standard mapping into a single-fold one.
\begin{definition}[Relocalisation map]\label{def:relocalisation_map}
    Let $\widetilde{X} \in \rr^{2^nd_X \times 2^nd_X}$ be a matrix on a space with a $n$-fold factorization like $\rr^{2^nd_X \times 2^nd_X} \cong \otimes_{i=1}^n \left(\rr^{2 \times 2} \otimes \rr^{d_{X_i} \times d_{X_i}}\right)$ and where $d_X = \prod_i d_{X_i}$. Then, for a given factor $X_j$, $1 \leq j \leq n$, we say that the map $\tau^{(n)}_{X_j}: \rr^{2^n d_X \times 2^n d_X} \rightarrow \rr^{2^n d_X \times 2^n d_X}$ defined by:
    \begin{multline}
        \tau^{(n)}_{X_j}\{\widetilde{X}\} :\cong I^{\ktensor n}  \ktensor \TrX{2_{X_1}2_{X_2}\dots 2_{X_n}}{\frac{1}{2}\I{n}_{d_{X_1}d_{X_2}\dots d_{X_n}}\:\widetilde{X}}\\
        + (I^{\ktensor (j-1) } \ktensor J \ktensor I^{\ktensor (n-j)} )\ktensor  \TrX{2_{X_1}2_{X_2}\dots 2_{X_n}}{\frac{1}{2}\J{n}_{d_{X_1}d_{X_2}\dots d_{X_n}}\:\widetilde{X}} \:,
    \end{multline}
    is called the \textbf{relocalization (at system $X_j$) map}. Here the $\TrX{2_{X_1}2_{X_2}\dots 2_{X_n}}{\cdot}$ symbols stands for taking the trace only over the 2-by-2 systems of each fold.
\end{definition}
Simply put, the relocalisation map replaces the delocalized state of the $n$ 2-by-2 phase-carrying-subsystems by a version that is localised at a chosen subsystem $X_j$. It does so by replacing the pair $\{\I{n},\J{n}\}$ defined over $\otimes_{i=1}^n \left(\rr^{2 \times 2}_{X_i} \right)$ by the pair $\{\I{1},\J{1}\}$ at factor $X_j$ and then padding all the other factors with $2\times 2$ identity matrices $I$.

To lessen clutter, if, in the following, the subscript of the relocalisation map is omitted, it should always be understood as the map that relocalises the phase in the leftmost factor, e.g., $\tau^{(n)} = \tau^{(n)}_{X_1}$ in the above equation.

The relocalisation map is a real-linear map that, when composed with an $n$-fold mapping, has exactly the properties we need:
\begin{lemma}\label{lemm:relocalisation}
    The map $\tau^{(n)} \circ \overline{\Gamma}^{(n)}: \bigotimes_{i=1}^n \cc^{d_i \times d_i} \rightarrow \bigotimes_{i=1}^n \rr^{2d_i \times 2d_i}$ is a positive-semi-definitness-preserving $*$-algebra homomorphism. 
    In addition, it is unital,
        \begin{equation}
            \overline{\Gamma}^{(n)}\{\id_{d}\} = \id_{2^n d}\:,
        \end{equation}
    ($d = \prod_{i=1}^n d_i$), and it verifies the following for all $A,B \in \bigotimes_{i=1}^n \cc^{d_i \times d_i}$:
        \begin{equation}
            \TrX{}{\overline{\Gamma}^{(n)}\{A\}\overline{\Gamma}^{(n)}\{B\}} = \TrX{}{\overline{\Gamma}^{(n)}\{A\}(\tau^{n} \circ \overline{\Gamma}^{(n)})\{B\}}\:.
        \end{equation}
\end{lemma}
\begin{proof}
    Since the relocalisation map $\tau^{(n)}$ is, by definition, a `trace-and-replace' map that does not act on the subsystems where $X$ is defined, it should be obvious why composing it with $\overline{\Gamma}^{(n)}$ is again a positive-semi-definiteness-preserving $*$-algebra homomorphism. Unitality follows by directly applying the definitions. 
    As for the last property, a direct computation leads to 
    \begin{equation}
        \begin{aligned}
            &\TrX{}{\overline{\Gamma}^{(n)}\{A\}(\tau^{n} \circ \overline{\Gamma}^{(n)})\{B\}} \\
            &= \TrX{}{(\I{n} \ktensor \ReP{A} + \J{n} \ktensor \ImP{A}) \: \tau^{n}\{(\I{n} \ktensor \ReP{B} + \J{n} \ktensor \ImP{B}\}} \\
            &= \mathrm{tr}\Big[(\I{n} \ktensor \ReP{A} + \J{n} \ktensor \ImP{A}) \\
            & \hphantom{= \mathrm{tr}\Big[(} \{( (I \ktensor I^{\ktensor (n-1)})\ktensor \ReP{B} + (J \ktensor I^{\ktensor (n-1)}) \ktensor \ImP{B}\}\Big]\\
            &= \mathrm{tr}\Big[\I{n} \ktensor (\ReP{A}\ReP{B}) + \J{n} \ktensor (\ReP{A}\ImP{B}) \\
            &\hphantom{= \mathrm{tr}\Big[(} + \J{n} \ktensor (\ImP{A}\ReP{B}) - \I{n} \ktensor \ImP{A}\ImP{B}\Big]\\
            &=\TrX{}{\I{n}\ktensor \ReP{AB} + \J{n}\ktensor \ImP{AB}}\\
            &=\TrX{}{\overline{\Gamma}^{(n)}\{AB\}} = \TrX{}{\overline{\Gamma}^{(n)}\{A\}\overline{\Gamma}^{(n)}\{B\}} \:,
        \end{aligned}
    \end{equation}
    where we used the definitions for the first and second equalities, \cref{lemm:I_J_global} for the third, definitions for the fourth and fifth, and \cref{prop:gamma_n_properties} for the last one.
\end{proof}
With this, we are ready to prove that QT embeds into RQT as RNQT also in the multilocal model.

\subsubsection{Multilocal RNQT models are embedded in multilocal RQT models}

With the introduction of the $R$-product, we can characterise the image of an $n$-partite, single-source, local QT model through the $n$-fold standard mapping.
\begin{lemma}\label{lemm:RNQT_model_multilocal_op_indep_tau}
    The image of the states and of the effects of the $n$-partite multilocal QT model of matrix dimension $d_Ad_B \dots$ through the pair of maps $(\frac{1}{2}\overline{\Gamma}^{(n)}, \overline{\Gamma}^{(n)} )$ form a sub-model in an $n$-partite multilocal RQT model of matrix dimension $2d_A2d_B \dots$ that has underwent the following modification:
    \begin{itemize}
        \item The state of independent sources is represented by operationally independent states.
    \end{itemize}
\end{lemma}
\begin{proof}
    Starting from \cref{lemm:RNQT_model_multilocal_op_indep}, we only have to prove that there is a way to remove the modification on effects. This is exactly what the localisation map has been developed for: by \cref{lemm:relocalisation} we can apply the relocalisation map on every effect without changing any predicted distribution\footnote{In case this is not clear: the choice of party for the relocalisation is totally arbitrary. Pick your favourite.}. Since this map still preserves sums and positivity, but makes the $n$-fold standard mapping unital, the effects now satisfy their RQT definition.
\end{proof}

This brings us to the main technical statement of this~\app{} which underlies the proof of~\Cref{theo:QTnetworksareRQTnetworks}.

\MulitilocalRQTisMultilocalQT*
\begin{proof}
    \Cref{lemm:RNQT_model_multilocal_op_indep_tau} guarantees that the QT-image of the RQT model is in isometric bijection with the QT model. The mappings are explicitly obtained as 
    \begin{equation}
        \begin{aligned}
            &M_{\Aa\Bb\dots} = \frac{1}{2}\overline{\Gamma}^{(n)}\:;\\
            \forall i, 1 \leq i \leq m\::\:& E_{M_i} = \tau^{(\abs{M_i})} \circ \overline{\Gamma}^{(\abs{M_i})}\:.
        \end{aligned}
    \end{equation}
\end{proof}
This last proposition exhausts all possible network structures for nonlocal inequalities. Since the standard mapping is an algebra homomorphism, the proof will readily extend to iterative scenarios, as we do in \cref{sec:proof}. 

\begin{remark}
    At the end of Ref.~\cite{Weilenmann2025}, the authors present a proposition that appears similar to \cref{prop:MulitilocalRQTisMultilocalQT}. The conceptual difference is that their proposition `externalises' the phase-encoding subsystems. The $n$ 2-by-2 systems are treated as an additional resource shared by the sources, instead of some non-physical `entanglement of representation' internal to the description of independent sources. 
    
    The technical difference is that this proposition only considers pure states and simplifies the measurements greatly: by seeing the phase-carrying subsystems as external resources, they effectively allow their measurements to be represented by matrices not summing up to the identity---in contrast to the POVM formalism of both QT and RQT. 
    In summary, their proposition is better seen as a special case of \cref{lemm:RNQT_model_multilocal_op_indep} rather than one of \cref{prop:MulitilocalRQTisMultilocalQT}. These simplifications explain why their proof is significantly more concise than ours\footnote{If one ignores the conceptual differences to focus only on the technical content of each result, Proposition 1 in Ref.~\cite{Weilenmann2025} has been independently derived in our previous work~\cite{RNQT} as well as in the work of Barrios et al.~\cite{Barrios2025}. All these results extended upon the works of Myrheim~\cite{Myrheim_1999} and McKague et al.~\cite{McKague}}.
\end{remark} 

\section{Real representation of the complex structure}\label{sec:real_units}
In this appendix, we detail the $2^n$-by-$2^n$ matrix representation of the pair $\{1,i\}$ used to define the $n$-fold standard mapping, and we prove some of its properties needed for our purposes.

\begin{definition}[Two-dimensional real representation of $(1,i)$]\label{def:I_J}
    The following two matrices in $\rr^{2 \times 2}$,
    \begin{subequations}
        \begin{align}
            &I := \begin{pmatrix} 1 &0 \\ 0 & 1 \end{pmatrix} \:,\\
            &J := \begin{pmatrix} 0 & -1 \\ 1 & 0 \end{pmatrix} \:,
        \end{align}
    \end{subequations}
    form what we call a two-dimensional real representation of the complex units $(1,i)$.
\end{definition}
The reason for the name is self-explanatory and is only presented as a definition in order to condense the representation-theory background. These matrices obey the same algebraic rules as the units of $\cc$, i.e., $I^2 = -J^2 = 1$ and $IJ = JI = J$. They are therefore used to represent a complex number as a real $2 \times 2$ matrix. This representation is indeed based on seeing the field $\cc$ as the two-dimensional real vector space spanned by the real $1$ and imaginary $i$ units, and provides an equivalent representation as the two-dimensional subspace of $\rr^{2 \times 2}$ spanned by $I$ and $J$. 

We will need a `global' generalisation of this representation. By this, we mean a representation on an $n$-fold tensor product space which generalises the matrices $I$ and $J$ into their `stabiliser code' logical counterparts $\I{n}$ and $\J{n}$. To do so, we define $(I,J) \in \rr^{2 \times 2}$ as  the single-fold representation $(\I{1},\J{1})$. The $n$-fold generalization $(\I{n},\J{n}) \in \bigotimes_{k=1}^n \rr^{2 \times 2}_k$ is defined by its globality property, meaning that we want these to be irreducible representations of $(1,i)$ such that they treat any local application (with respect to the real Kronecker product) of a $k$-fold representation $(\I{k},\J{k})$, $ k \leq n$, as if it were the corresponding $n$-fold representation. For example, we would need $\I{n} (\J{k} \ktensor \underbrace{I \ktensor I \ktensor \ldots}_{n-k \: \mathrm{terms}} ) $ to be equivalent to $\I{n}\J{k}$.
A recursive definition of such matrices is given by the following.
\begin{definition}[$n$-fold two-dimensional real representation of $(1,i)$]\label{def:I_J_n}
    For $n\in \nn_0$, the matrices $(\I{n},\J{n})$ in $\rr^{2n \times 2n}$ are defined via the following recursive rule:
    \begin{subequations}\label{eq:def_I_J}
        \begin{gather}
        (1,i) \in \cc \times \cc \mapsto (\I{n},\J{n}) \in \rr^{2n\times 2n} \times \rr^{2n \times 2n}\:: \notag \\
            \I{1} := I  \:, \quad \I{n}{} := \frac{1}{2}\left(\I{n-1} \ktensor I  - \J{n-1} \ktensor J\right)\:;\\
            \J{1} := J  \:, \quad \J{n}{} := \frac{1}{2}\left(\J{n-1} \ktensor I + \I{n-1} \ktensor J\right)  \:;
        \end{gather}
    \end{subequations}
    with $ I = \left(\begin{smallmatrix} 1 & 0 \\ 0 & 1 \end{smallmatrix}\right)$
    and $J = \left(\begin{smallmatrix}  0 & -1 \\ 1 & 0 \end{smallmatrix}\right)$.
\end{definition}
\begin{prop}
    The matrices $(\I{n},\J{n})$ provide an irreducible representation in $\rr^{2n \times 2n}$ of the cyclic group $\zz_2$ under matrix multiplication, similarly to how $(1,i)$ provides one in $\cc$.
\end{prop}
\begin{proof}
    The following are straightforward to show by induction:
     \begin{subequations}
        \begin{align}
            &\I{n}  = \big(\I{n} \big)^2 = - \big(\J{n}\big)^2 \:; \\
            &\J{n} = \J{n}\I{n} = \I{n}\J{n} \:.
        \end{align}
    \end{subequations}
    The proof of irreducibility can be done via character theory; see our previous work~\cite{RNQT}.
\end{proof}

We will need some notation to prove the properties of these matrices relevant to this \doc. First of all, the notation $A^{\otimes n} := \underbrace{A \otimes A \otimes \ldots \otimes A}_{n}$ will be used for an $n$-fold application of a given combination rule to the matrix $A$. For example, instead of writing $\underbrace{I \ktensor I \ktensor \ldots \ktensor I}_n \in \underbrace{\rr^{2 \times 2} \otimes \rr^{2 \times 2} \otimes \ldots \otimes \rr^{2 \times 2}}_n$, we shall use the more compact $I^{\ktensor n} \in \bigotimes_{i=1}^n \rr^{2 \times 2}$.

We will also need to address independent tensor factors in a composite space. To do so, we associate each factor with a party like Alice, Bob, Charlie, etc. as a handle to individually address factors. That way, we can quickly refer to the dimension of a factor by using the notation $d_A,$ to mean that Alice is associated with a space of $d_A \times d_A$-dimensional matrices.


Finally, we will need a notation to `pad' a matrix with a Kronecker product of unit matrices in order for it to reach a certain size (this procedure is sometimes called an \textit{amplification} in the literature); to do so, we will use a subscript with the corresponding dimension: $A_{d_B} := A \ktensor \id_{d_B}$. Because the Kronecker product is associative and because $\id_{d_Bd_C} = \id_{d_B} \ktensor \id_{d_C}$, several such subscripts can be combined, like in $A_{d_Bd_C} = A \ktensor \id_{d_B} \ktensor \id_{d_C}$ for instance. 
By mixing this notation with the previous ones, we shall then write expressions like 
\begin{subequations}
    \begin{align}
        &I_{A,2_B2_C} = I_{B,2_A2_C} = I_{C,2_A2_B} = I_A \ktensor I_B \ktensor I_C \:; \\
        &I_{A,2_B2_D} \ktensor J_C = I_A \ktensor I_B \ktensor J_C \ktensor I_D \:.
    \end{align}
\end{subequations}
Also, when there is no risk of confusion, the labels may be dropped. In the first line, we could have written $I_{2_B2_C}$ or $ I^{\ktensor 3} = I \ktensor I \ktensor I$ instead of the leftmost term in each equation.

\begin{prop}[Characterisation of the $n$-fold representation]
    \begin{subequations}\label{eq:I_J_characterisation}
        \begin{align}
            &\I{n}{} = \sum_{\substack{\Vec{i}=(i_1 i_2 ... i_n) \in \{0,1\}^n,\\ \sum_k i_k \text{ is even}}} \frac{1}{2^{n-1}} (-1)^{\frac{\sum_k i_k}{2}}J^{i_1} \ktensor J^{i_2} \ktensor \ldots \ktensor J^{i_n}\:; \label{eq:I_characterisation}\\
            &\J{n}{} = \sum_{\substack{\Vec{i}=(i_1 i_2 ... i_n) \in \{0,1\}^n,\\ \sum_k i_k \text{ is odd}}} \frac{1}{2^{n-1}} (-1)^{\frac{\sum_k i_k-1}{2}}J^{i_1} \ktensor J^{i_2} \ktensor \ldots \ktensor J^{i_n}\:; \label{eq:J_characterisation}
        \end{align}
    \end{subequations}
\end{prop}
\begin{proof}
    Using \cref{def:I_J_n}, the cases $n=1$, i.e. $\I{1}=I, \J{1}=J$ and $n=2$, i.e. $\I{2}=\frac{1}{2}\left( I \ktensor I - J\ktensor J \right)$, $\J{2}=\frac{1}{2}\left( I \ktensor J + J\ktensor I \right)$ can be proven to fit \cref{eq:I_J_characterisation} by direct inspection. We shall now use the $n=2$ case to prove the induction on $n$.

    Suppose it holds for $n-1$, then for the case $n$, we have for $\I{n}$,
    \begin{equation}
        \begin{aligned}
            \I{n}{} = \frac{1}{2}\Big(&\big(\sum_{\substack{\Vec{i} \in \{0,1\}^{n-1},\\ \sum_k i_k \text{ is even}}} \frac{1}{2^{n-2}} (-1)^{\frac{\sum_k i_k}{2}}J^{i_1} \ktensor J^{i_2} \ktensor \ldots \ktensor J^{i_{n-1}}\big) \ktensor I\\
            &- \big( \sum_{\substack{\Vec{i} \in \{0,1\}^{n-1},\\ \sum_k i_k \text{ is odd}}} \frac{1}{2^{n-2}} (-1)^{\frac{\sum_k i_k-1}{2}}J^{i_1} \ktensor J^{i_2} \ktensor \ldots \ktensor J^{i_{n-1}} \big) \ktensor J \Big) \\
            = &\big(\sum_{\substack{\Vec{i} \in \{0,1\}^{n-1},\\ \sum_k i_k \text{ is even}}} \frac{1}{2^{n-1}} (-1)^{\frac{\sum_k i_k}{2}}J^{i_1} \ktensor J^{i_2} \ktensor \ldots \ktensor J^{i_{n-1}} \ktensor I \big) \\
            &- \big( \sum_{\substack{\Vec{i} \in \{0,1\}^{n-1},\\ \sum_k i_k \text{ is odd}}} \frac{1}{2^{n-1}} (-1)^{\frac{\sum_k i_k-1}{2}}J^{i_1} \ktensor J^{i_2} \ktensor \ldots \ktensor J^{i_{n-1}} \ktensor J \big) \:. \\
            = &\big(\sum_{\substack{\Vec{j} \in \{0,1\}^{n-1}, j_n =0,\\ \sum_k^{n-1} j_k \text{ is even}}} \frac{1}{2^{n-1}} (-1)^{\frac{\sum_k j_k}{2}}J^{j_1} \ktensor J^{j_2} \ktensor \ldots \ktensor J^{i_{n-1}} \ktensor J^{j_n} \big) \\
            &- \big( \sum_{\substack{\Vec{j} \in \{0,1\}^{n}, j_n = 1\\ \sum_k^{n-1} j_k \text{ is odd}}} \frac{1}{2^{n-1}} (-1)^{\frac{\sum_k i_k-1}{2}}J^{i_1} \ktensor J^{i_2} \ktensor \ldots \ktensor J^{i_{n-1}} \ktensor J^{j_n} \big) \:. \\
        \end{aligned}
    \end{equation}
    This is exactly $\I{n}$ as in \cref{eq:I_characterisation}, since the first term is the sum up to $n$ but only for the cases where the last bit is fixed to 0, $\Vec{i}=(i_1, i_2, ... i_{n-1}, \: i_n=0)$ whereas the second is the one for the cases where it is fixed to 1, $\Vec{i}=(i_1, i_2, ... i_{n-1}, \: i_n= 1)$. Since $i_n$ can only have a value in $\{0,1\}$, this covers all cases proving the induction.
    
    The proof of \eqref{eq:J_characterisation} is a similar induction.
\end{proof}
This immediately implies
\begin{corollary}
    The matrices $(\I{n}, \J{n})$ are tensor-factor-permutation-independent.     
\end{corollary}
since the set of all even (resp. odd) words $\Vec{i}\in \{0,1\}^n$ is permutation-invariant.

Also, the same inductive techniques as in the proof of \cref{eq:I_J_characterisation} can be used to further regroup the terms on the right as a $k$-fold representation. 
\begin{corollary}\label{coro:I_J_groupping}
    For every $k \in \nn_0, k \leq n$, the definition \eqref{eq:def_I_J} verifies:
    \begin{subequations}\label{eq:I_J_groupping}
        \begin{gather}
            \I{n}{} = \frac{1}{2}\left(\I{n-k} \ktensor \I{k}  - \J{n-k} \ktensor \J{k}\right)\:;\\
            \J{n}{} = \frac{1}{2}\left(\J{n-k} \ktensor \I{k} + \I{n-k} \ktensor \J{k}\right)  \:.
        \end{gather}
    \end{subequations}
\end{corollary}

This last property can be used to show that representation on more subsystems acts globally in the following sense: 
\begin{lemma}\label{lemm:I_J_global}
    The matrices $(\I{n},\J{n})$ are a global representation of $(1,i)$ in the sense that they behave properly when multiplying elements of a smaller-dimensional representation $(\I{k},\J{k})$ with $k \leq n$:
     \begin{subequations}
        \begin{align}
            \forall k \leq n,& \notag \\
            &(\I{k} \ktensor I^{\otimes (n-k)})\I{n} = \I{n}(\I{k} \ktensor I^{\otimes (n-k)}) = (\I{n})^2\:;\\
            &(\J{k} \ktensor I^{\otimes (n-k)})\J{n}{} = \J{n}{}(\J{k} \ktensor I^{\otimes (n-k)}) = (\J{n}{})^2 \:;\\
            &(\I{k} \ktensor I^{\otimes (n-k)})\J{n} = \J{n}(\I{k} \ktensor I^{\otimes (n-k)}) = \J{n} \:;\\
            & (\J{k} \ktensor I^{\otimes (n-k)})\I{n}{} = \I{n}{}(\J{k} \ktensor I^{\otimes (n-k)}) = \J{n}\:.
        \end{align}
    \end{subequations}
\end{lemma}
\begin{proof}
    The proof is the same for every equation and relies on \cref{eq:I_J_groupping}. We prove the first line as an example:
    \begin{multline}
        \I{n} (\I{k} \ktensor I^{\otimes (n-k)})= \frac{1}{2}\left(\I{k} \ktensor \I{n-k}  - \J{k} \ktensor \J{n-k}\right)(\I{k} \ktensor I^{\otimes (n-k)})\\
        = \frac{1}{2}\left((\I{k})^2 \ktensor \I{n-k}  - (\J{k}\I{k}) \ktensor \J{n-k}\right)\\
        = (\I{k} \ktensor I^{\otimes (n-k)})\frac{1}{2}\left(\I{k} \ktensor \I{n-k}  - \J{k} \ktensor \J{n-k}\right)\\
        = (\I{k} \ktensor I^{\otimes (n-k)})\I{n} \\
        = \frac{1}{2}\left((\I{k})^2 \ktensor \I{n-k}  - (\J{k}\I{k}) \ktensor \J{n-k}\right)\\
        = \frac{1}{2}\left(\I{k} \ktensor \I{n-k}  - \J{k} \ktensor \J{n-k}\right) \\ = (\I{n})^2\\
        =\I{n}\:.
    \end{multline}
\end{proof}
In a more technically accurate language, a global representation is the tensor product of representations of algebras. In this case, $(\I{n},\J{n})$ is the $n-$fold tensor product of the 2-dimensional irreducible representation $(I,J)$ of $\cc$ seen as a real algebra. In terms of representation of algebras, the above lemma is actually a necessary condition for $(\I{n},\J{n})$ to be called a tensor product of representations.

\section{Proof of \cref{theo:RQTisQT}}\label{sec:proof} 
This appendix shows how our main theorem is implied by \cref{prop:MulitilocalRQTisMultilocalQT}. That is, 
how that every network scenario that has a QT explanation also has an RQT explanation implies that every experiment with a QT explanation will have one with an RQT explanation. 

In \cref{sec:proof_setup}, we make more explicit what is meant by the statement:
\RQTisQT*
To do so, we introduce a general description of experimental protocols using the quantum channel formalism. With this general QT description of `an experiment', we can then introduce its RQT counterpart and obtain a technically precise rephrasing of the above theorem statement. 

In \cref{sec:proof_proof}, we then prove the theorem by generalising the proof technique of \cref{prop:MulitilocalRQTisMultilocalQT}, restated below for completeness.
\MulitilocalRQTisMultilocalQT*

\subsection{Setup and statement of the theorem}\label{sec:proof_setup}
\paragraph{Quantum-Theoretical description of the experimental setup.} 
Consider an $n$-partite quantum system associated with the joint Hilbert space $\mathcal H=\bigotimes_{j=1}^n \mathcal H_{S_j}$, where each of its subsystems $S_1,\ldots,S_n$ is associated with a tensor factor of the Hilbert space $\mathcal{H}_{S_1}, \ldots ,\mathcal{H}_{S_n}$. 
Let $\rho^{(0)}$ be the arbitrary initial state of the joint system\footnote{This initial state can be pictured as any prior on the `state of the global past'. For all intents and purposes, it can describe the electrical noise in the experimentalists' devices before they turn them on, just as it can describe the universe right after the Big Bang.} 
The dynamics of the experiment consist of the successive operations made by parties on some part(s) of the system. We organise the operations into rounds, each round consisting of the parties applying a quantum channel (encompassing the deterministic operations (some of) the parties have performed during this round, like e.g., preparing the system in a given state), followed by them potentially applying a generalised measurement (encompassing the probabilistic operations (some of) the parties have performed during this round, like e.g., measuring the part of the system). 

These operations are associated with measurement outcomes, and these outcomes can influence the operation performed at a subsequent round. In other words, if $x_T$ represents the outcome(s) obtained at round $T$, then the operations at round $T+1$ may be conditioned according to the string $\mathbf{x}_{\leq T}:= (x_1,x_2,\ldots,x_T)$ of all the outcomes obtained during the $T$ previous rounds. 

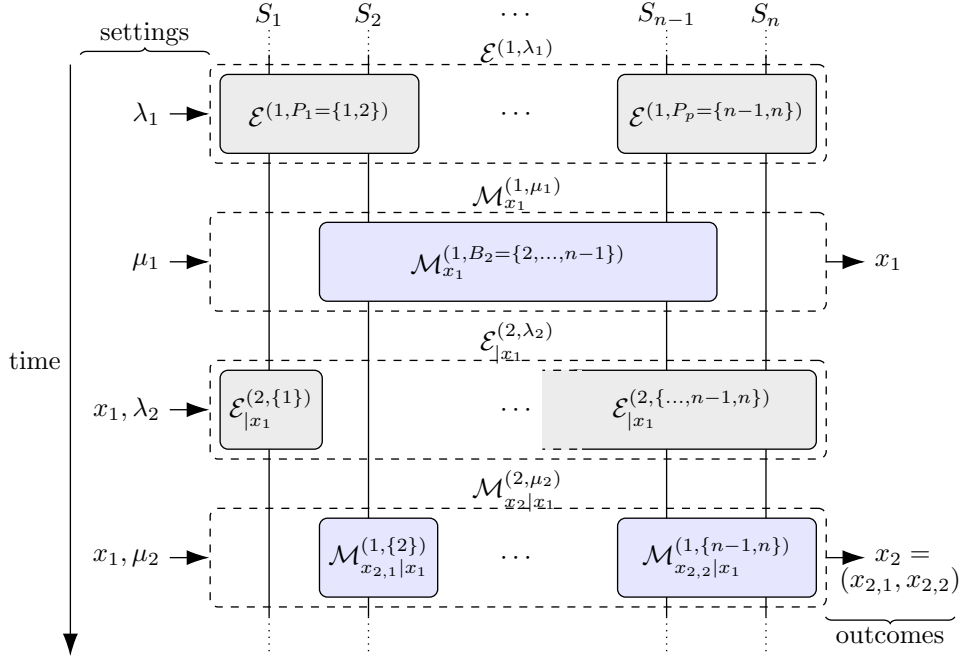
\begin{figure}[htb]
	\centering
	\begin{tikzpicture}[
		x=0.05\linewidth,
		y=0.05\linewidth,
		>=Latex,
		line width=0.5pt,
		worldline/.style={gray!60},
		dotchan/.style={
			draw,
            dashed,
			rounded corners=3pt,
			minimum height=.1\linewidth,
			inner xsep=8pt,
			inner ysep=3pt
		},
        chan/.style={
			draw,
			rounded corners=3pt,
			fill=gray!15,
			minimum height=.08\linewidth,
			inner ysep=3pt
		},
		meas/.style={
			draw,
			rounded corners=3pt,
			fill=blue!10,
			minimum height=.08\linewidth,
			inner ysep=3pt
		}
		]
		
		\node at (0,1) {$S_1$};
		\node at (2,1) {$S_2$};
		\node at (5,1) {$\cdots$};
		\node at (8,1) {$S_{n-1}$};
		\node at (10,1) {$S_n$};

		\foreach \x in {0,2,8,10} {
			\draw (\x,0.2) -- (\x,-11.2);
            \draw[dotted] (\x,0) -- (\x,0.7);
            \draw[dotted] (\x,-11.2) -- (\x,-11.9);
		}
		
		\draw[-{Latex[length=3mm,width=2mm]}] (-4,0) -- (-4,-12);
		\node[left] at (-4,-6) {time};

        \path[draw,decorate,decoration=brace] (-3.8,0.2) -- (-1.2,0.2)
        node[midway,above]{settings};
        \path[draw,decorate,decoration=brace] (13.8,-11.2) -- (11.2,-11.2) 
        node[midway,below]{outcomes};
        \node[dotchan, minimum width=.62\linewidth, anchor=west] at (-1.2,-1) {$\cdots$};
        \node[above] at (5,-0.15) {$\mathcal E^{(1,\lambda_1)}$};
        \draw[-{Latex[length=3mm,width=2mm]}] (-2,-1) -- (-1.2,-1);
		\node[left] at (-2,-1) {$\lambda_1$};
		\node[chan, minimum width=.2\linewidth, anchor=west] at (-1,-1) {$\mathcal E^{(1,P_1 = \{1,2\})}$};
		\node[chan, minimum width=.2\linewidth, anchor=west] at (7,-1) {$\mathcal E^{(1,P_p = \{n-1,n\})}$};
		
        \node[dotchan, minimum width=.62\linewidth, anchor=west] at (-1.2,-4) {$\cdots$};
        \node[above] at (5,-3.15) {$\mathcal M^{(1,\mu_1)}_{x_1}$};
        \draw[-{Latex[length=3mm,width=2mm]}] (-2,-4) -- (-1.2,-4);
		\node[left] at (-2,-4) {$\mu_1$};
        \draw[-{Latex[length=3mm,width=2mm]}] (11.2,-4) -- (12,-4);
		\node[right] at (12,-4) {$x_1$};
		\node[meas, minimum width=.4\linewidth, anchor=west] at (1,-4) {$\mathcal M^{(1,B_2 = \{2,\ldots,n-1\})}_{x_1}$};

        \node[dotchan, minimum width=.62\linewidth, anchor=west] at (-1.2,-7) {$\cdots$};
        \node[above] at (5,-6.25) {$\mathcal E^{(2,\lambda_2)}_{|x_1}$};
        \draw[-{Latex[length=3mm,width=2mm]}] (-2,-7) -- (-1.2,-7);
		\node[left] at (-2,-7) {$x_1,\lambda_2$};
		\node[chan, minimum width=.1\linewidth, anchor=west] at (-1,-7) {$\mathcal E^{(2,\{1\})}_{|x_1}$};
		\node[chan, minimum width=.25\linewidth, anchor=west] at (6,-7) {$\mathcal E^{(2,\{\dots,n-1,n\})}_{|x_1}$};
        \node[anchor=center, minimum height=.081\linewidth, minimum width=.04\linewidth, fill=gray!15] at (5.9,-7) {};
        \draw[dashed] (6.2,-6.2) --  (5.5,-6.2);
        \draw[dashed] (6.2,-7.8) --  (5.5,-7.8);

        \node[dotchan, minimum width=.62\linewidth, anchor=west] at (-1.2,-10) {$\cdots$};
        \node[above] at (5,-9.25) {$\mathcal M^{(2,\mu_2)}_{x_2|x_1}$};
        \draw[-{Latex[length=3mm,width=2mm]}] (-2,-10) -- (-1.2,-10);
		\node[left] at (-2,-10) {$x_1,\mu_2$};
        \draw[-{Latex[length=3mm,width=2mm]}] (11.2,-10) -- (12,-10);
		\node[right] at (12,-10) {$x_2=$};
        \node[right,below] at (12.7,-10) {$(x_{2,1},x_{2,2})$};
		\node[meas, minimum width=.1\linewidth, anchor=west] at (1,-10) {$\mathcal M^{(1,\{2\})}_{x_{2,1}|x_1}$};
		\node[meas, minimum width=.2\linewidth, anchor=west] at (7,-10) {$\mathcal M^{(1,\{n-1,n\})}_{x_{2,2}|x_1}$};
		
	\end{tikzpicture}
	\caption{
		Schematic space--time representation of the alternating channel--measurement
		process on an $n$-partite system. Vertical lines denote worldlines of subsystems, with time flowing towards the bottom of the page. The overall operations on the $n$ subsystems are represented by dashed boxes. Each operation carries a classical locality label (shown on the left), encoding the partition of subsystems across which the corresponding channel or measurement factorises into non-trivial sub-operations. These sub-operations are represented by grey boxes when they are channels and blue boxes when they are measurements, with the horizontal extent of each box indicating the subset of subsystems on which the corresponding sub-operation acts nontrivially. Measurement rounds additionally produce classical outcomes (shown on the right). Note that when several local measurements occur, as is the case in the last operation (measurement phase of the second round), the measurement outcomes are split into sub-outcomes, $x_t = (x_{t,1},x_{t,2},\dots)$. Each sub-outcome $x_{t,i}$ corresponds to what a local party has observed when measuring their respective set of subsystems $B_i$.
	}
	\label{fig:alternating-channel-measurement}
\end{figure}
The operations may have some local constraints; it can be assumed that the channel applied by the parties at a given round splits into two sub-channels in Kronecker product because the parties, together with the subsystems they have been given, are assumed spacelike separated. 
Hence, in addition to the measurement outcomes, we explicitly associate with each operation a piece of classical \emph{locality data}. More precisely, at round $t$ and conditional on the previous outcome history $\mathbf{x}_{<t}=(x_1,\ldots,x_{t-1})$, the channel is specified not only by a completely positive trace-preserving map $\mathcal E^{(t,\lambda_t)}_{|\mathbf{x}_{<t}}:\mathcal L(\mathcal H)\to\mathcal L(\mathcal K)$, with $\mathcal{K}$ another $n$-partite Hilbert space\footnote{Bear in mind that the dimension of the Hilbert space associated with a subsystem can be reduced to $1$ after the action of a channel in order to represent the discarding of it. Conversely, the action of a channel can also increase the dimension from $1$ to any positive integer in order to represent the re-preparation of the subsystem. For example, for the subsystem $S_i$ associated with Hilbert space $\mathcal{H}_i$ before the action of $\mathcal E^{(t,\lambda_t)}_{|\mathbf{x}_{<t}}$ and with $\mathcal{K}_i$ afterwards, if $\mathrm{dim}\left(\mathcal{H}_i\right) > \mathrm{dim}(\mathcal{K}_i) = 1$ then $S_i$ was discarded during the operation.}, but also by a classical locality label $\lambda_t$ which may depends on the outcome history ${\mathbf{x}_{<t}}$. This label determines a partitioning 
$\mathcal{P}(\lambda_t) = \{P_1, P_2, \dots, P_p\}$
of the set of parties $\{S_1,\ldots,S_n\}$ into $p \leq n$ subsets, with the interpretation that the channel factorises across these blocks:
\begin{equation}
    \mathcal E^{(t,\lambda_t)}_{|\mathbf{x}_{<t}} = \bigotimes_{\alpha=1}^p{\!\vphantom{\bigotimes}}_K\: \mathcal E^{(t,P_\alpha)}_{|\mathbf{x}_{<t}}.
\end{equation}
Here, singleton blocks and identity factors are permitted, so that channels acting nontrivially only on a subset of parties are included as special cases. For example, the singleton partition $\mathcal{P}(\lambda_t) =\{P_1\}$ such that $P_1=\{1,\ldots,n\}$ corresponds to a fully global channel. For another example, a label of the form
\begin{equation}
\mathcal{P}(\lambda_t) = \bigl\{\{5,6\},\{11,12,13\},\{1\},\{2\},\ldots\bigr\}
\end{equation}
indicates a channel that factorises into one component acting on $S_5S_6$,
another acting on $S_{11}S_{12}S_{13}$, and local or identity components on
the remaining parties.

Similarly, the measurement at round $t$ is specified by POVMs of the form $\big\{M_{x_t|\mathbf{x}_{<t}}^{(t,\mu_t)} = \bigl(N_{x_t|\mathbf{x}_{<t}}^{(t,\mu_t)}\bigr)^\dagger
N_{x_t|\mathbf{x}_{<t}}^{(t,\mu_t)}\big\}_{x_t}$, where $\{M_{x_t|\mathbf{x}_{<t}}^{(t,\mu_t)}\}_{x_t}$ is a valid POVM, i.e, $\sum_{x_t} \bigl(N_{x_t|\mathbf{x}_{<t}}^{(t,\mu_t)}\bigr)^\dagger
N_{x_t|\mathbf{x}_{<t}}^{(t,\mu_t)} =\id$, where $x_t$ is the observed outcome, where the ``$|\mathbf{x}_{<t}$'' label indicates that the string of previous outcomes $\mathbf{x}_{<t}$ might be used as a setting to condition the choice of the POVM, and where $\mu_t$ is the classical locality label. This label determines a partition of $\{S_1,\ldots,S_n\}$ into $q \leq n$ subset of parties,
\begin{equation}
    \mathcal{P}(\mu_t)=\{B_1,\ldots,B_q\}
\end{equation}
such that, whenever the measurement is (block-)local with respect to this partition, the corresponding measurement operators factorise as
\begin{equation}
    N_{x_t|\mathbf{x}_{<t}}^{(t,\mu_t)} = \bigotimes_{\beta=1}^q{\!\vphantom{\bigotimes}}_K\: N_{x_t|\mathbf{x}_{<t}}^{(t,B_\beta)}.
\end{equation}
That way, product measurements, partially local measurements, and fully global measurements are all described within the same notation.

If $\rho^{(t-1)}_{|\mathbf{x}_{<t}}$ denotes the normalised state immediately before the application of $\mathcal{E}^{(t,\lambda_t)}_{|\mathbf{x}_{<t}}$, then the conditional probability distribution of the outcome $x_t$ obtained after the subsequent measurement $\{M_{x_t|\mathbf{x}_{<t}}^{(t,\mu_t)}\}$ is given by
\begin{equation}
    p(x_t|\mathbf{x}_{<t}) = \TrX{}{ N_{x_t|\mathbf{x}_{<t}}^{(t,\mu_t)}\:  \mathcal{E}^{(t,\lambda_t)}_{|\mathbf{x}_{<t}}\bigl(\rho^{(t-1)}_{|\mathbf{x}_{<t}}\bigl) \:\left( N_{x_t|\mathbf{x}_{<t}}^{(t,\mu_t)}\right)^\dag}
\end{equation}
and the corresponding post-measurement state is
\begin{equation}
    \rho^{(t)}_{|\mathbf{x}_{<(t+1)}} = \frac{N_{x_t|\mathbf{x}_{<t}}^{(t,\mu_t)}\:  \mathcal{E}^{(t,\lambda_t)}_{|\mathbf{x}_{<t}}\bigl(\rho^{(t-1)}_{|\mathbf{x}_{<t}}\bigl) \:\left( N_{x_t|\mathbf{x}_{<t}}^{(t,\mu_t)}\right)^\dag}{p(x_t|\mathbf{x}_{<t})} \:.
\end{equation}

Equivalently\footnote{We have chosen to represent measurement via POVMs to remain consistent with the rest of the text, but one could readily generalise the proof to general quantum instruments, i.e. to any collection of completely positive trace non-increasing maps that sum up to a quantum channel.}, each outcome $x_t$ defines a trace-non-increasing completely positive map,
\begin{equation}
    \mathcal M_{x_t|\mathbf{x}_{<t}}^{(t,\mu_t)}(\cdot) = N_{x_t,\mu_t}^{(t,\mathbf x_{<t})} (\cdot) \bigl(N_{x_t,\mu_t}^{(t,\mathbf x_{<t})}\bigr)^\dagger.
\end{equation} 
Hence, a single realised run specified by an outcome string $\mathbf x=(x_1,\ldots,x_T)$, is naturally described by the composition of the $2T$ outcome-conditioned subchannels,
\begin{equation}
    \mathcal T_{\mathbf x} := \mathcal M_{x_T|\mathbf{x}_{<T}}^{(T,\mu_T)} \circ \mathcal E_{x_T|\mathbf{x}_{<T}}^{(T,\lambda_T)} \circ \ldots \circ \mathcal M_{x_1}^{(1,\mu_1)} \circ \mathcal E^{(1,\lambda_1)} \:,
\end{equation}
such that the joint probability distribution of the outcome string is given by
\begin{equation}
    p(\mathbf x) = \TrX{}{\mathcal T_{\mathbf x}\bigl(\rho^{(0)}\bigr)}\:.
\end{equation}
Hence, in a completely realised experimental run represented by $\mathcal{T}_{\mathbf x}$, the full classical transcript consists of the sequence of triples
\[
(\lambda_1,\mu_1,x_1),\,
(\lambda_2,\mu_2,x_2),\,
\ldots,\,
(\lambda_T,\mu_T,x_T)\:,
\]
so that the alternating channel--measurement process may equivalently be viewed
as a sequence of quantum operations together with their associated classical locality labels\footnote{Remark that the locality labels may depend on previous outcomes, like e.g., $\lambda_2 = \lambda_2(x_1)$ and $\mu_T = \mu_T(\mathbf{x}_{<T}) = \mu_T(x_1,x_2,\dots, x_{T-1})$.} and measurement outcomes. 
In this sense, once one conditions on a particular sequence of observed outcomes, the alternating channel--measurement process may be viewed simply as a sequence of outcome-labelled quantum operations acting on the multipartite system.

\paragraph{Equivalent Real Quantum Theory description of the experimental setup.} 
We associate with each QT operation appearing in the above adaptive process a corresponding RQT operation. Concretely, for each channel
\begin{equation}
\mathcal E^{(t,\lambda_t)}_{|\mathbf{x}_{<t}}:\mathcal L(\mathcal H)\to \mathcal L(\mathcal K),
\end{equation}
we introduce a corresponding RQT channel
\begin{equation}
\mathsf R\!\mathcal E^{(t,\lambda_t)}_{|\mathbf{x}_{<t}}
:\mathcal L(\widetilde{\mathcal H})\to \mathcal L(\widetilde{\mathcal K}),
\end{equation}
where $\widetilde{\mathcal H}$ is a real Hilbert space, and for each measurement subchannel
\begin{equation}
    \mathcal M_{x_t|\mathbf{x}_{<t}}^{(t, \mu_t)}: \mathcal L(\mathcal H)\to \mathcal L(\mathcal H),
\end{equation}
we introduce a corresponding RQT subchannel
\begin{equation}
\mathsf R\!\mathcal M_{x_t|\mathbf{x}_{<t}}^{(t, \mu_t)}
:\mathcal L(\widetilde{\mathcal H})\to \mathcal L(\widetilde{\mathcal H}).
\end{equation}
(If needed, see Ref.~\cite{Chiribella_2023} for a review of RQT channels, i.e. completely positive and trace-preserving maps between operators acting on real Hilbert spaces.)  

Importantly, these RQT operations must have the same locality structure as their QT counterparts. That is, if the locality label $\lambda_t$ specifies the partition
	\begin{equation}
	\mathcal{P}(\lambda_t)=\{P_1,\ldots,P_p\}\:,
	\end{equation}
such that 
	\begin{equation}
	\mathcal E^{(t,\lambda_t)}_{|\mathbf{x}_{<t}} = \bigotimes_{\alpha=1}^p{\!\vphantom{\bigotimes}}_K\:\mathcal E^{(t,P_\alpha)}_{|\mathbf{x}_{<t}}\:,
	\end{equation}
then
	\begin{equation}
	\mathsf R\!\mathcal E^{(t,\lambda_t)}_{|\mathbf{x}_{<t}} =
	\bigotimes_{\alpha=1}^m{\!\vphantom{\bigotimes}}_K\:
	\mathsf R\!\mathcal E^{(t,P_\alpha)}_{|\mathbf{x}_{<t}} \:.
	\end{equation}
Likewise, if the locality label $\mu_t^{\mathbf{x}_{<t}}$ specifies the partition
	\begin{equation}
	\mathcal{P}(\mu_t)=\{B_1,\ldots,B_q\}\:,
	\end{equation}
and the measurement subchannel factorises accordingly, then
	\begin{equation}
	\mathsf R\!\mathcal M_{x_t|\mathbf{x}_{<t}}^{(t, \mu_t)} =\bigotimes_{\beta=1}^q{\!\vphantom{\bigotimes}}_K\: \mathsf R\!\mathcal M_{x_t|\mathbf{x}_{<t}}^{(t, B_\beta)}\:.
	\end{equation}
In particular, fully global, partially local, and fully product operations are all mapped to RQT operations with exactly the same block structure.

\paragraph{Theorem statement.} We can now state the implicit technical content of \cref{theo:RQTisQT} (here denoted~\Cref{thm2 in app}), which is our main structural result.

\begin{theorem}[Locality-preserving RQT reproduction of adaptive QT statistics]\label{thm2 in app}
	Consider an arbitrary $n$-partite QT adaptive process of the form
	described above, specified by an initial state $\rho^{(0)}$, together with a family of channels $\mathcal E^{(t,\lambda_t)}_{|\mathbf{x}_{<t}} $ and measurement subchannels $\mathcal M_{x_t|\mathbf{x}_{<t}}^{(t, \mu_t)}$, for all rounds $t$ and all prior outcome strings $\mathbf{x}_{<t}$.
	
	Then for every triple $\big(\rho^{(0)}, \big\{\mathcal E^{(t,\lambda_t)}_{|\mathbf{x}_{<t}}\big\}_t, \big\{\mathcal M_{x_t|\mathbf{x}_{<t}}^{(t, \mu_t)}\big\}_t\big)$ associated with the classical data $\big(\mathbf{x}_{\leq t}, \{\lambda_t\}_t, \{\mu_t\}_t\big)$, there exists a RQT counterpart $\big(\rho_{\mathbb R}^{(0)}, \big\{\mathsf R\!\mathcal E^{(t,\lambda_t)}_{|\mathbf{x}_{<t}}\big\}_t, \big\{\mathsf R\!\mathcal M_{x_t|\mathbf{x}_{<t}}^{(t,\mu_t)}\big\}_t\big)$ associated with the same classical data and such that:
	\begin{enumerate}
		\item
		the RQT operations have the same locality structure as their QT counterparts, in the sense encoded by the same locality labels
		$\lambda_t$ and $\mu_t$;
		
		\item
		the resulting RQT adaptive process reproduces exactly the same
		measurement statistics at every round. Equivalently, for every round $t$, every
		prior history $\mathbf{x}_{<t}$, and every outcome $x_t$,
		\begin{equation}
		p(x_t\mid \mathbf{x}_{<t})
		=
		p_{\mathbb R}(x_t\mid \mathbf{x}_{<t}) \:,
		\end{equation}
		and hence, for every finite outcome string
		$\mathbf{x}_{\le T}=(x_1,\ldots,x_T)$,
		\begin{equation}
		p(\mathbf{x}_{\le T})
		=
		p_{\mathbb R}(\mathbf{x}_{\le T}) \:;
		\end{equation}
		
		\item
		the correspondence is componentwise and context-independent: each
		RQT channel
		$\mathsf R\!\mathcal E^{(t,\lambda_t)}_{|\mathbf{x}_{<t}}$
		depends only on the corresponding QT channel
		$\mathcal E^{(t,\lambda_t)}_{|\mathbf{x}_{<t}}$, and not on the initial state nor
		on any other operations appearing elsewhere in the protocol; likewise, each
		RQT measurement subchannel
		$\mathsf R\!\mathcal M_{x_t|\mathbf{x}_{<t}}^{(t, \mu_t)}$
		depends only on the corresponding QT measurement subchannel
		$\mathcal M_{x_t|\mathbf{x}_{<t}}^{(t, \mu_t)}$; and the initial real-image
		state $\rho_{\mathbb R}^{(0)}$ depends only on the initial QT state $\rho^{(0)}$, and not on the subsequent choice of channels or
		measurements.
	\end{enumerate}
	
	Thus every finite adaptive multipartite QT test admits a
	RQT counterpart with identical operational statistics and identical
	locality pattern at each stage of the protocol.
\end{theorem}

\subsection{Proof of the theorem}\label{sec:proof_proof}
The proof, like the one of \cref{prop:MulitilocalRQTisMultilocalQT}, consists in applying the complex-to-real map developed in detail in \cref{sec:RNQT}, as if we were trying to map QT to RNQT. However, to respect the local structure of the iterative protocol, the RQT composition rule --the Kronecker product $\ktensor$--, is used in lieu of the RNQT composition rule --the $\rtensordual$-product in this case. The mapping will work nevertheless because the initial real state will always be an $\rtensor$-product state (or a linear sum thereof), and such a state will `contaminate' all the other $\ktensor$-products into becoming $\rtensordual$-products. 

Specifically, consider the first step of the protocol, where the state is updated from $\rho^{(0)}$ to
\begin{equation}
    \rho^{(1)}_{x_1} = \bigl(\mathcal{M}^{(1,\mu_1)}_{x_1} \circ \mathcal{E}^{(1,\lambda_1)}\bigr)\bigl(\rho^{(0)}\bigr) \:,
\end{equation}
with a probability of $p(x_1) = \TrX{}{\bigl(\mathcal{M}^{(1,\mu_1)}_{x_1} \circ \mathcal{E}^{(1,\lambda_1)}\bigr)\bigl(\rho^{(0)}\bigr)}$. Define the trace-normalised version of this state as
\begin{equation}
    \rho^{(1)}_{|x_1} = \frac{1}{p(x_1)} \rho^{(1)}_{x_1} \:.
\end{equation}
This is the resulting state of the system after the outcome $x_1$ was observed during round $t=1$, and this is the state on which the operations of step $t=2$ will be applied.

Suppose, for simplicity, that $\lambda_1$ defines a bipartition and so does $\mu_1$ (the proof generalises to any partition mutatis mutandis). Assume that $\lambda_1$ cuts the set of parties into $P_1 = \{1,2,3,\dots, k\}$ and $P_2 = \{k+1, k+2, \dots, n\}$, this means that there exists a pair of channels $\mathcal{E}^{(1,P_1)}$ and $\mathcal{E}^{(1,P_2)}$ such that $\mathcal{E}^{(1,\lambda_1)} = \mathcal{E}^{(1,P_1)} \ktensor \mathcal{E}^{(1,P_2)}$. 
Similarly, assume that $\mu_1$ cuts the set of parties into two other, arbitrary partitions $B_1$ and $B_2$ such that $\abs{B_1} = j$ and $\abs{B_2} = n-j$. This means that there exists a pair of instruments $\{\mathcal{M}_{x_{1,1}}^{(1,B_1)}\}_{x_{1,1}}$ and $\{\mathcal{M}_{x_{1,2}}^{(1,B_2)}\}_{x_{1,2}}$ and a permutation of the Hilbert space factorisation $\Pi_{\mu_1}$ such that the instrument factorises as 
\begin{equation}
    \left\{\mathcal{M}_{x_1}^{(1,\mu_1)}\right\}_{x_1} = \left\{\Pi_{\mu_1}^\dag \circ \big(\mathcal{M}_{x_{1,1}}^{(1,B_1)} \ktensor \mathcal{M}_{x_{1,2}}^{(1,B_2)}\big) \circ \Pi_{\mu_1}  \right\}_{(x_{1,1},x_{1,2})} \:,
\end{equation}
and such that the outcome $x_1$ accordingly factorises into a pair of random variables $x_1 = (x_{1,1},x_{1,2})$.
Here, the permutation amounts $\Pi_{\mu_1}$ to a redistribution of the subsystems between the first operation and the first measurement; with respect to \cref{fig:alternating-channel-measurement}, picture it as a reordering of the wires before and after the measurement such that the pair of quantum instruments (the blue boxes) are represented side by side. 

(Remark in passing that a single round can be interpreted as a network scenario, with the operation representing the preparation of the sources, and the instrument their measurement; the parameter $\lambda_1$ then tracks how the sources are distributed, and $\mu_1$ how the measurements are. The map $\Pi_{\mu_1}$ then encodes the topology of the network, i.e., how the subsystems are distributed from the sources to the parties.)

When the locality constraints have been applied the first step takes the form
\begin{equation}
    \rho^{(1)}_{x_1} = \left[\Pi_{\mu_1}^\dag \circ \big(\mathcal{M}_{x_{1,1}}^{(1,B_1)} \ktensor \mathcal{M}_{x_{1,2}}^{(1,B_2)}\big) \circ \Pi_{\mu_1}  \circ \big(\mathcal{E}^{(1,P_1)} \ktensor \mathcal{E}^{(1,P_2)}\big) \right](\rho^{(0)}) \:.
\end{equation}
Recall that a map that applies a permutation $\sigma$ of the Hilbert space factors acts on operators as the adjoint action of a swap gate, i.e.,
\begin{equation}
    \begin{gathered}
        \Pi_{\sigma} : \mathcal{L}(\mathcal{H}_1 \otimes \mathcal{H}_2 \otimes \dots \otimes \mathcal{H}_n) \rightarrow \mathcal{L}(\mathcal{H}_{\sigma(1)} \otimes \mathcal{H}_{\sigma(2)} \otimes \dots \otimes \mathcal{H}_{\sigma(n)}) ::\\
        \Pi_{\sigma}(\rho) = S_\sigma \rho S_\sigma^\dag \qquad \text{where}\\
        S_\sigma: \mathcal{H}_1 \otimes \mathcal{H}_2 \otimes \dots \otimes \mathcal{H}_n \rightarrow\mathcal{H}_{\sigma(1)} \otimes \mathcal{H}_{\sigma(2)} \otimes \dots \otimes \mathcal{H}_{\sigma(n)} ::\\
        S_\sigma(\ket{\psi_1} \ktensor \ket{\psi_2}\ktensor \dots \ket{\psi_n}) = \ket{\psi_{\sigma(1)}} \ktensor \ket{\psi_{\sigma(2)}}\ktensor \dots \ket{\psi_{\sigma(n)}} \:.
    \end{gathered}
\end{equation}
Hence, using the operator-sum representation of the quantum channels, $\mathcal{E}^{(1,P_1)} = \sum_i E_i^{(1,P_1)} \cdot (E_i^{(1,P_1})^\dag$ and $\mathcal{E}^{(1,P_2)} = \sum_l E_l^{(1,P_2)} \cdot (E_l^{(1,P_2)})^\dag$, this expression for $\rho^{(1)}$ is equivalent to
\begin{multline}
    \rho^{(1)}_{x_1} = \sum_{i,l}
    \left[S_{\mu_1}^\dag \big(N_{x_{1,1}}^{(1,B_1)} \ktensor N_{x_{1,2}}^{(1,B_2)}\big) S_{\mu_1}  \big({E}^{(1,P_1)}_{i} \ktensor {E}^{(1,P_2)}_{l}\big) \right]\rho^{(0)} \\
    \left[S_{\mu_1}^\dag \big(N_{x_{1,1}}^{(1,B_1)} \ktensor N_{x_{1,2}}^{(1,B_2)}\big) S_{\mu_1}  \big({E}^{(1,P_1)}_{i} \ktensor {E}^{(1,P_2)}_{l}\big) \right]^\dag \:.
\end{multline}
(For conciseness, the repeated term on the right of $\rho^{(0)}$ will be abbreviated into $[...]$ in the following.)

We now construct the corresponding real operations. First, we apply the $n$-fold standard mapping to $\rho^{(1)}_{x_1}$ as in \cref{def:std_n_mapping} with a factor of $1/2$ for trace-normalisation purposes. Using the fact that the $n$-fold mapping is a $*$-algebra homomorphism (\cref{prop:gamma_n_properties}), it gives
\begin{multline}
     \frac{1}{2}\overline{\Gamma}^{(n)}\{\rho^{(1)}_{x_1}\} = \\
     \sum_{i,l} \left[\overline{\Gamma}^{(n)}\{S_{\mu_1}\}^T \overline{\Gamma}^{(n)}\{N_{x_{1,1}}^{(1,B_1)} \ktensor N_{x_{1,2}}^{(1,B_2)}\} \overline{\Gamma}^{(n)}\{S_{\mu_1}\}  \overline{\Gamma}^{(n)}\{{E}^{(1,P_1)}_{i} \ktensor {E}^{(1,P_2)}_{l}\} \right]\\
    \frac{1}{2}\overline{\Gamma}^{(n)}\{\rho^{(0)}\} \Big[...\Big]^T \:.
\end{multline}
Again by \cref{prop:gamma_n_properties}, states of the form $\frac{1}{2}\overline{\Gamma}^{(n)}\{\rho^{(t)}_{\mathbf{x}_{\leq t}}\}$ must be valid RQT state, whereas $\overline{\Gamma}^{(n)}\{S_{\mu_1}\}$ is an orthogonal matrix since $S_{\mu_1}$ is a unitary matrix (see Ref.~\cite{RNQT} for a proof), hence $\overline{\Gamma}^{(n)}\{S_{\mu_1}\} \:\cdot\: \overline{\Gamma}^{(n)}\{S_{\mu_1}\}^T$ is a valid RQT operation. 

Next, we focus on the two maps that are in Kronecker product form. In accordance with \cref{prop:RprodisKprod}, we can decompose the $n$-fold standard mappings that are being applied into a pair of smaller-fold-mappings using the R-product:
\begin{multline}
     \frac{1}{2}\overline{\Gamma}^{(n)}\{\rho^{(1)}_{x_1}\} =\\
     \sum_{i,l} \Big[\overline{\Gamma}^{(n)}\{S_{\mu_1}\}^T \big(\overline{\Gamma}^{(j)}\{N_{x_{1,1}}^{(1,B_1)}\} \Rprod \overline{\Gamma}^{(n-j)}\{N_{x_{1,2}}^{(1,B_2)}\}\big) \overline{\Gamma}^{(n)}\{S_{\mu_1}\} \\
     \big(\overline{\Gamma}^{(k)}\{{E}^{(1,P_1)}_{i}\} \Rprod \overline{\Gamma}^{(n-k)}\{{E}^{(1,P_2)}_{l}\}\big) \Big]
    \frac{1}{2}\overline{\Gamma}^{(n)}\{\rho^{(0)}\} \Big[ \dots \Big]^T \:.
\end{multline}
Then, we can use the fact that this expression is equivalent to the one where the $\Rprod$-products have been replaced by $\ktensor$-products,
\begin{multline}
     \frac{1}{2}\overline{\Gamma}^{(n)}\{\rho^{(1)}_{x_1}\} =\\
     \sum_{i,l} \Big[\overline{\Gamma}^{(n)}\{S_{\mu_1}\}^T \big(\overline{\Gamma}^{(j)}\{N_{x_{1,1}}^{(1,B_1)}\} \ktensor \overline{\Gamma}^{(n-j)}\{N_{x_{1,2}}^{(1,B_2)}\}\big) \overline{\Gamma}^{(n)}\{S_{\mu_1}\} \\
     \big(\overline{\Gamma}^{(k)}\{{E}^{(1,P_1)}_{i}\} \ktensor \overline{\Gamma}^{(n-k)}\{{E}^{(1,P_2)}_{l}\}\big) \Big]
    \frac{1}{2}\overline{\Gamma}^{(n)}\{\rho^{(0)}\} \Big[ \dots \Big]^T \:.
\end{multline}
This follows from the contamination property discussed in \cref{sec:RNQT_contamination}: one can make the representation of the complex matrix unit appear on the left of the initial state $\overline{\Gamma}^{(n)}\{\rho^{(0)}\} = \overline{\Gamma}^{(n)}\{\id \rho^{(0)}\} = \overline{\Gamma}^{(n)}\{\id \} \overline{\Gamma}^{(n)}\{\rho^{(0)}\} = \I{n}_d \overline{\Gamma}^{(n)}\{\rho^{(0)}\}$ (where $d$ is the total dimension of the Hilbert space and the notation $\I{n}_d $ is, up to permutations of tensor factors, equivalent to the matrix $\I{n}\ktensor \id_d$ with $\I{n}$ defined by \cref{def:I_J_n}). This matrix $\I{n}_d$ is idempotent (see \cref{sec:real_units}; it is actually a projector that defines a subspace of the real Hilbert space in which the image of QT lives), commutes with Kronecker products of the standard mapping (\cref{eq:Rprod_Special_Sym}), and turns them into $\Rprod$-product (\cref{eq:Rprod_Special_Sym_bis}). One can therefore make this matrix move to the left in order to `contaminate' every Kronecker product into $\Rprod$-products it encounters along the way, and once it has reached the end of the expression, it can be reintegrated into any other matrix, like, e.g., via, $\I{n}_d\overline{\Gamma}^{(n)}\{S_{\mu_1}\}^T = \overline{\Gamma}^{(n)}\{\id\}\overline{\Gamma}^{(n)}\{S_{\mu_1}\}^T= (\overline{\Gamma}^{(n)}\{S_{\mu_1}\}\overline{\Gamma}^{(n)}\{\id^\dag\})^T = \overline{\Gamma}^{(n)}\{S_{\mu_1}\id\} = \overline{\Gamma}^{(n)}\{S_{\mu_1}\}$.

With this done, we can remark that the matrices $\overline{\Gamma}^{(k)}\{{E}^{(1,P_1)}_{i}\}$ and $ \overline{\Gamma}^{(n-k)}\{{E}^{(1,P_2)}_{l}\}$ are real matrices which have exactly the sought structural form for the real operation $\mathsf R\!\mathcal E^{(1,\lambda_1)}$. Indeed, define
\begin{subequations}
    \begin{align}
        & \mathsf R\!\mathcal E^{(1,\lambda_1)} := \mathsf R\!\mathcal E^{(1,P_1)} \ktensor \mathsf R\!\mathcal E^{(1,P_2)} \:;\\
        &\mathsf R\!\mathcal E^{(1,P_1)}(\cdot) := \sum_i \overline{\Gamma}^{(k)}\{{E}^{(1,P_1)}_{i}\} \:\cdot\: (\overline{\Gamma}^{(k)}\{{E}^{(1,P_1)}_{i}\})^T \:;\\
        &\mathsf R\!\mathcal E^{(1,P_2)}(\cdot) := \sum_l  \overline{\Gamma}^{(n-k)}\{{E}^{(1,P_2)}_{l}\} \:\cdot\: (\overline{\Gamma}^{(n-k)}\{{E}^{(1,P_2)}_{l}\})^T\:.
    \end{align}
\end{subequations}
By construction, the two smaller maps are in operator-sum forms that only involve real matrices, hence the total $\mathsf R\!\mathcal E^{(1,\lambda_1)}$ is indeed a CP map split according to $\lambda_1$ into the Kronecker product of two CP maps. One can then use \cref{prop:gamma_n_properties} to show that these three maps are trace-preserving on QT-image RQT states, and finally, one can combine this proposition with \cref{lemm:Rprod_contaminates} to extend the proof to every RQT state.

The same thing cannot be directly done with the matrices $\overline{\Gamma}^{(j)}\{N_{x_{1,1}}^{(1,B_1)}\}$ and $ \overline{\Gamma}^{(n-j)}\{N_{x_{1,2}}^{(1,B_2)}\}$ because these do no longer define valid POVMs as the multiple fold standard mapping is not unital. Fortunately, this problem has already been addressed for the proof of \cref{prop:MulitilocalRQTisMultilocalQT}. The standard mapping $\overline{\Gamma}^{n}$ can be composed with the relocalisation map $\tau^{(n)}$ (see \cref{def:relocalisation_map}) to make it unital. Since it is direct to check that $(\tau^{(n)} \circ \overline{\Gamma}^{(n)})\{\rho\} \I{n}_d = \overline{\Gamma}^{(n)}\{\rho\}$ from the definitions, the same contamination argument as above can be used to argue that adding this map to the expression will not change its outcome: as 
\begin{multline}
    \Big(\big(\tau^{j}\circ \overline{\Gamma}^{(j)}\big)\{N_{x_{1,1}}^{(1,B_1)}\} \ktensor \big( \tau^{(n-j)} \circ\overline{\Gamma}^{(n-j)}\big)\{N_{x_{1,2}}^{(1,B_2)}\}\Big) (\I{n}_d)^2 = \\
    \Big(\big(\tau^{j}\circ \overline{\Gamma}^{(j)}\big)\{N_{x_{1,1}}^{(1,B_1)}\} \rtensor \big( \tau^{(n-j)} \circ\overline{\Gamma}^{(n-j)}\big)\{N_{x_{1,2}}^{(1,B_2)}\}\Big) \I{n}_d =\\
    \Big(\big(\tau^{j}\circ \overline{\Gamma}^{(j)}\big)\{N_{x_{1,1}}^{(1,B_1)}\} \Rprod \big( \tau^{(n-j)} \circ\overline{\Gamma}^{(n-j)}\big)\{N_{x_{1,2}}^{(1,B_2)}\}\Big) \big(\I{j}_{d_{B_1}} \Rprod\I{n-j}_{d_{B_2}}\big) = \\
    \Big(\big(\tau^{j}\circ \overline{\Gamma}^{(j)}\big)\{N_{x_{1,1}}^{(1,B_1)}\}\I{j}_{d_{B_1}} \Big) \Rprod \Big( \big( \tau^{(n-j)} \circ\overline{\Gamma}^{(n-j)}\big)\{N_{x_{1,2}}^{(1,P_2)}\}\I{n-j}_{d_{B_2}}\Big) = \\
    \overline{\Gamma}^{(j)}\{N_{x_{1,1}}^{(1,B_1)}\} \Rprod \overline{\Gamma}^{(n-j)}\{N_{x_{1,2}}^{(1,B_2)}\} \:,
\end{multline}
it follows that
\begin{multline}
     \frac{1}{2}\overline{\Gamma}^{(n)}\{\rho^{(1)}_{x_1}\} =\\
     \sum_{i,l} \Big[\overline{\Gamma}^{(n)}\{S_{\mu_1}\}^T \Big( \big(\tau^{j}\circ \overline{\Gamma}^{(j)}\big)\{N_{x_{1,1}}^{(1,B_1)}\} \ktensor \big( \tau^{(n-j)} \circ\overline{\Gamma}^{(n-j)}\big)\{N_{x_{1,2}}^{(1,B_2)}\}\Big) \overline{\Gamma}^{(n)}\{S_{\mu_1}\} \\
     \big(\overline{\Gamma}^{(k)}\{{E}^{(1,P_1)}_{i}\} \ktensor \overline{\Gamma}^{(n-k)}\{{E}^{(1,P_2)}_{l}\}\big) \Big]
    \frac{1}{2}\overline{\Gamma}^{(n)}\{\rho^{(0)}\} \Big[ \dots \Big]^T \:.
\end{multline}
We are now ready to define the real counterpart of the measurement subchannel. Let  
\begin{subequations}
    \begin{align}
        &\mathsf R\!\mathcal M_{x_1}^{(1, \mu_1)} := \mathsf R\! \Pi_{\mu_1}^T \circ \big(\mathsf R\!\mathcal M_{x_{1,1}}^{(1, B_1)} \ktensor  R\!\mathcal M_{x_{1,2}}^{(1, B_2)}\big) \circ \mathsf R\! \Pi_{\mu_1} \:;\\
        & \mathsf R\!\mathcal M_{x_{1,1}}^{(1, B_1)}(\cdot) :=  \big(\tau^{j}\circ \overline{\Gamma}^{(j)}\big)\{N_{x_{1,1}}^{(1,B_1)}\}  \:\cdot\: \big(\big(\tau^{j}\circ \overline{\Gamma}^{(j)}\big)\{N_{x_{1,1}}^{(1,B_1)}\} \big)^T\:;\\
        & \mathsf R\!\mathcal M_{x_{1,2}}^{(1, B_2)}(\cdot) := \big( \tau^{(n-j)} \circ\overline{\Gamma}^{(n-j)}\big)\{N_{x_{1,2}}^{(1,B_2)}\} \:\cdot \: \big( \big( \tau^{(n-j)} \circ\overline{\Gamma}^{(n-j)}\big)\{N_{x_{1,2}}^{(1,B_2)}\}\Big) \big)^T\:;\\
        & \mathsf R\! \Pi_{\mu_1}(\cdot) := \overline{\Gamma}^{(n)}\{S_{\mu_1}\} \:\cdot \: \big(\overline{\Gamma}^{(n)}\{S_{\mu_1}\}\big)^T \:; \\
        & \mathsf R\! \Pi_{\mu_1}^T(\cdot) := \big(\overline{\Gamma}^{(n)}\{S_{\mu_1}\}\big)^T \:\cdot \: \overline{\Gamma}^{(n)}\{S_{\mu_1}\} \:.
    \end{align}
\end{subequations}
Despite the convoluted form of the mappings, it is quite direct to check that each side of the Kronecker product defines a valid RQT POVM since the standard mapping is a *-algebra homomorphism; for example, $\left(\big(\tau^{j}\circ \overline{\Gamma}^{(j)}\big)\{N_{x_{1,1}}^{(1,B_1)}\}\right)^T \big(\tau^{j}\circ \overline{\Gamma}^{(j)}\big)\{N_{x_{1,1}}^{(1,B_1)}\}$ is a positive matrix because it is of the form $M^TM$ with $M$ real, and when summed over $x_{1,1}$ it gives the identity because:
\begin{equation}
    \begin{aligned}
        \sum_{x_{1,1}} \left(\big(\tau^{j}\circ \overline{\Gamma}^{(j)}\big)\{N_{x_{1,1}}^{(1,B_1)}\}\right)^T &\big(\tau^{j}\circ \overline{\Gamma}^{(j)}\big)\{N_{x_{1,1}}^{(1,B_1)}\}\\ 
        &= \sum_{x_{1,1}} \big(\tau^{j}\circ \overline{\Gamma}^{(j)}\big)\{\big(N_{x_{1,1}}^{(1,B_1)}\big)^\dag\} \big(\tau^{j}\circ \overline{\Gamma}^{(j)}\big)\{N_{x_{1,1}}^{(1,B_1)}\}\\
        &= \sum_{x_{1,1}} \big(\tau^{j}\circ \overline{\Gamma}^{(j)}\big)\{\big(N_{x_{1,1}}^{(1,B_1)}\big)^\dag N_{x_{1,1}}^{(1,B_1)}\}\\
        &= \big(\tau^{j}\circ \overline{\Gamma}^{(j)}\big)\{\sum_{x_{1,1}} \big(N_{x_{1,1}}^{(1,B_1)}\big)^\dag N_{x_{1,1}}^{(1,B_1)}\}\\
         &= \big(\tau^{j}\circ \overline{\Gamma}^{(j)}\big)\{\id_d\} = \id_{2^n d} \:,\\
    \end{aligned}
\end{equation}
where the only not-yet-used property appeared after the second equality, namely that $\tau^{j}\circ \overline{\Gamma}^{(j)}$ is also an algebra homomorphism. Since the map $\tau^{j}\circ \overline{\Gamma}^{(j)}$ is isomorphic to the single-fold standard mapping, this property follows from \cref{prop:gamma_properties}. 
Consequently, both maps define proper RQT instruments; it should then be no surprise that $\mathsf R\!\mathcal M_{x_1}^{(1, \mu_1)}$, which is their composition using the Kronecker product and their composition with the orthogonal map $\mathsf R\! \Pi_{\mu_1}(\cdot)$ is also a valid RQT instrument, since the Kronecker product is the composition rule of RQT and since the orthogonal maps are to RQT what unitary maps are to QT: noiseless evolution.

It remains to prove that the statistics of the outcome $x_1$ are preserved, i.e. that 
\begin{equation}
    p(x_1) = p_\rr(x_1) \:,
\end{equation}
where
\begin{subequations}
    \begin{align}
        & p(x_1) = \TrX{}{\rho^{(1)}_{x_1}} = \TrX{}{\bigl(\mathcal{M}^{(1,\mu_1)}_{x_1} \circ \mathcal{E}^{(1,\lambda_1)}\bigr)\bigl(\rho^{(0)}\bigr) }\:;\\
        & p_{\rr}(x_1) = \TrX{}{\rho^{(1)}_{\rr , x_1}} = \TrX{}{\bigl(\mathsf R\!\mathcal{M}^{(1,\mu_1)}_{x_1} \circ \mathsf R\!\mathcal{E}^{(1,\lambda_1)}\bigr)\bigl(\rho^{(0)}_\rr\bigr) }\:.
    \end{align}
\end{subequations}
Once $\rho^{(0)}_\rr$ has been set to $\frac{1}{2}\overline{\Gamma}^{(n)}\{\rho^{(0)}\}$, this equality directly follows from the trace preservation of the standard mapping (up to a $\frac{1}{2}$ factor; see \cref{prop:gamma_n_properties}). Indeed, our construction of the real mappings relied on showing the equality $\bigl(\mathsf R\!\mathcal{M}^{(1,\mu_1)}_{x_1} \circ \mathsf R\!\mathcal{E}^{(1,\lambda_1)}\bigr)\bigl(\rho^{(0)}_\rr\bigr) = \frac{1}{2} \overline{\Gamma}^{(n)}\{\rho^{(1)}_{x_1}\}$. Taking the trace of each side, we get
\begin{equation}
    \TrX{}{\bigl(\mathsf R\!\mathcal{M}^{(1,\mu_1)}_{x_1} \circ \mathsf R\!\mathcal{E}^{(1,\lambda_1)}\bigr)\bigl(\rho^{(0)}_\rr\bigr)} =\TrX{}{\frac{1}{2} \overline{\Gamma}^{(n)}\{\rho^{(1)}_{x_1}\}} \:.
\end{equation}
The left-hand side is $p_\rr(x_1)$ by definition, and the aforementioned trace-preservation is used on the right-hand side as $\TrX{}{\frac{1}{2} \overline{\Gamma}^{(n)}\{\rho^{(1)}_{x_1}\}} = \TrX{}{\rho^{(1)}_{x_1}}$, leading to the sought
\begin{equation}
    p_\rr(x_1) = p(x_1) \:.
\end{equation}
The probabilities are indeed preserved. 

This proves the theorem for the first time step. The proof is then directly generalised by remarking that the construction at step $T$ is not dependent on the outcome string $\mathbf{x_{<T}}$, one can then see the step $T-1$ as the zeroth step and run the same construction as for the first step to obtain an RQT explanation for the step $T$. This concludes the proof. 


\section{RQT description with equivalent QT description}\label{sec:RQTisnotQT}
While the methods of this article have been concerned with showing that the QT models can always be explained by RQT models, little has been said about the reverse inclusion. 
At first glance, it might appear trivial: any RQT model is a QT model when `one allows complex numbers back in the maths'; i.e., by using the identity embedding $\iota:\rr^{n\times n} \rightarrow \cc^{n \times n}: \rho \mapsto \iota(\rho) = \rho$. However, this reasoning is once again only true under the assumption that independent sources only produce systems that are product-state independent (\cref{def:independent_sources}). 

If, however, the operational meaning of independence is enforced, one is no longer free to map RQT to QT simply by using the identity embedding from the reals to the complexes. This is because some operationally independent states lose this property under this embedding, as implied by \cref{prop:independent_sources}. 

Take for example the matrix $\rho \in \rr^{2 \times 2}$ defined as
\begin{equation}
    \rho := \frac{I}{2} \ktensor \dyad{0} = \frac{1}{4} (I \ktensor I + I \ktensor Z) \:,
\end{equation} 
where we use the notation $\dyad{0} = \frac{I+Z}{2} = \left(\begin{smallmatrix}
    1 & 0 \\ 0 & 0
\end{smallmatrix}\right)$. This matrix is a valid RQT state in its own right, but it is also the image through the map $\frac{1}{2}\Gs$ of the QT state $\dyad{0} \in \cc^{2 \times 2}$. 
Now, considering two such states, one can define a bipartite state as their Kronecker product,
\begin{equation}
    \Psi_{\Xx\Yy}^{K} := \rho_\Xx \ktensor \rho_\Yy \:,
\end{equation}
or as their $R$-product:
\begin{equation}
    \Psi_{\Xx\Yy}^{\textup{R}} := \rho_\Xx \rtensor \rho_\Yy \:.
\end{equation}
As we mentioned in the main text (and proved at \cref{prop:locallyIndistinguishable,lemm:Rprod_contaminates}), these two states are locally indistinguishable, and both are operationally independent. It means that for every pair of RQT POVMs $\{X_\alpha\}_\alpha,\{Y_{\beta}\}_\beta \subset \rr^{4 \times 4}$, they verify
\begin{equation}
    \TrX{}{\Psi_{\Xx\Yy}^K(X_{\alpha} \ktensor Y_{\beta})} = \TrX{}{\Psi_{\Xx\Yy}^{\textup{R}}(X_{\alpha} \ktensor Y_{\beta})} = \TrX{}{\rho_\Xx \: X_\alpha} \TrX{}{\rho_\Yy \: Y_\beta} \:. 
\end{equation}

This plays a role when trying to embed these states back in QT: $\iota$ maps only product states to product states, but only the products states are operationally independent in QT. Hence, while $\iota(\Psi_{\Xx\Yy}^K)$ maps an operationally independent state of RQT to one of QT, it will not be able to map $\Psi_{\Xx\Yy}^{\textup{R}}$ to one.

To see it concretely, we can turn the $\rtensor$-product into an expression involving $\ktensor$-products:
\begin{equation}
    \begin{aligned}
        \Psi_{\Xx\Yy}^{\textup{R}} = &\Big[\left( \frac{I}{2} \ktensor \frac{I + Z}{2} \right)_\Xx \ktensor \left( \frac{I}{2} \ktensor \frac{I + Z}{2} \right)_\Yy \Big] \\
        &-  \Big[\left( \frac{J}{2} \ktensor \frac{I + Z}{2} \right)_\Xx \ktensor \left( \frac{J}{2} \ktensor \frac{I + Z}{2} \right)_\Yy \Big]\:.
    \end{aligned}
\end{equation}
If we now allow the presence of complex numbers, i.e. if we see $\Psi_{\Xx\Yy}^{\textup{R}}$ as a QT state in $\cc^{16 \times 16}$, it suddenly becomes separable\footnote{It is worth remarking that while $\rtensor$-product states are usually entangled, any pair of states which have a QT-image as a pair of symmetric states will be separable under a similar decomposition. That is, $\Gamma \Rprod \Gamma$ maps (R)QT product states in $\rr^{2 \times 2} \otimes \rr^{2\times 2}$ to separable QT states in $\cc^{4\times 4} \otimes \cc^{4\times 4}$.} over the partition $(\Xx)(\Yy)$:
\begin{equation}
    \begin{aligned}
        \Psi_{\Xx\Yy}^{\textup{R}} = & \frac{1}{2} \Big[ \left(\frac{I + i J}{2} \ktensor \frac{I + Z}{2} \right)_\Xx \ktensor  \left(\frac{I + i J}{2} \ktensor \frac{I + Z}{2} \right)_\Yy \Big] \\
        &+\frac{1}{2}  \Big[ \left(\frac{I - i J}{2} \ktensor \frac{I + Z}{2} \right)_\Xx \ktensor  \left(\frac{I - i J}{2} \ktensor \frac{I + Z}{2} \right)_\Yy\Big] \:.
    \end{aligned}
\end{equation}
This is equivalent to:
\begin{equation}
    \Psi_{\Xx\Yy}^{\textup{R}} = \frac{1}{2} \dyad{+ i,0}_\Xx \ktensor \dyad{+i,0}_\Yy + \frac{1}{2} \dyad{-i,0}_\Xx \ktensor \dyad{-i,0}_\Yy \:,
\end{equation}
where $\ket{\pm i} = \frac{1}{\sqrt{2}}(\ket{0} \pm \ket{1})$. This sudden separability means that $\iota(\Psi_{\Xx\Yy}^{\textup{R}})$ is no longer operationally independent, so the operational meaning of this state has been changed when passing from an RQT model to a QT model. The two models are thus describing different situations, and therefore $\iota$ cannot be used to obtain a QT description of RQT.

In addition to providing a simple argument why the RQT description cannot be directly lifted into a QT description, this example is consequently one of the simplest situations in which the correlations observables in the RQT model cannot be explained by a QT model. Indeed, because $\Psi_{\Xx\Yy}^{\textup{R}}$ has a QT description (it is in the QT-image of the RQT model), $\Psi_{\Xx\Yy}^{K}$ cannot have one. Since these two states cannot have a QT explanation at the same time, some protocol involving them must yield at least one outcome distribution that cannot be explained by QT. 

Here is an example of such a distribution: suppose you are trying to characterise a pair of sources using the statistics of measurement outcomes. Within a QT model of the sources, once every possible local measurement have been performed for a sufficiently large amount of time, you are certain that you have the maximal knowledge of the state; no new information can be inferred about the joint state of the systems outputted by the sources, even if you now consider global measurements. In the corresponding RQT model, however, you can still learn new information by performing a global measurement. 

Take the pair of RQT states $\Psi_{\Xx\Yy}^{\textup{R}}$ and $\Psi_{\Xx\Yy}^{K}$ considered in the example. As mentioned above, these are locally indistinguishable, meaning that all the measurement statistics of local measurements will not be sufficient to distinguish them. If you were to give a QT description of these two RQT states, you would be forced to conclude that they are described by the same QT state. Yet, a two-element RQT POVM such as $\{E_\omega\}_{\omega}$ with $\omega \in \{0,1\}$ and
\begin{equation}
    E_{\omega} = \frac{1}{2} (I\ktensor I)_\Xx \ktensor (I \ktensor I)_\Yy + \frac{(-1)^\omega}{2} (J\ktensor I) _\Xx \ktensor (J \ktensor I)_\Yy\:,
\end{equation}
will yield an outcome distribution able to differentiate the two states because 
\begin{equation}
    \begin{gathered}
        \TrX{}{ \Psi_{\Xx\Yy}^{K} E_0} = \frac{1}{2}\:; \qquad \TrX{}{ \Psi_{\Xx\Yy}^{K} E_1} = \frac{1}{2}\:;\\
        \TrX{}{ \Psi_{\Xx\Yy}^{\textup{R}} E_0} = 0\:; \qquad  \TrX{}{ \Psi_{\Xx\Yy}^{\textup{R}} E_1} = 1 \:.
    \end{gathered}
\end{equation}
Depending on whether the outcome distribution is balanced or peaked, one can know if the state was of the Kronecker or $R$-product form. The sole existence of outcome statistics of bipartite measurements that cannot be explained with local measurements is a specificity in the predictions of the RQT models that QT models will never be able to describe. 

Now the example considered in this appendix is a very simple way of playing with the non-tomographic-locality of RQT. This will most likely not be possible to infer a Bell-like inequality out of it or any generalisation, as it would provide a theory-independent way to test locally tomographic theories against non-locally tomographic ones. Yet, it prompts many questions paving the road for future investigations: can it be shown that there are no theory-independent ways of witnessing these locally-similar-yet-globally-different states? Can the notion of operational independence be strengthened to avoid this issue somehow? Is imposing product-state independence on RQT models the only way to be able to consistently embed them in QT models? 



\end{document}